\documentclass[11pt]{article}
\usepackage{fullpage}
\usepackage[english]{babel}
\usepackage[cmex10]{amsmath}
\usepackage{todonotes}
\usepackage{amsthm}
\usepackage{amssymb}
\usepackage{microtype}
\usepackage{xspace}
\usepackage{comment}
\usepackage{letltxmacro}
\usepackage{comment}
\usepackage{enumerate}
\usepackage{enumitem}
\usepackage{tikz}
\usetikzlibrary{calc}
\usepackage{url}

\usepackage{graphicx}
\usepackage{caption}
\usepackage{subcaption}
\usepackage{float}
\usepackage{tikz}
\usetikzlibrary{shapes,patterns,decorations.pathmorphing}

\definecolor{cA}{RGB}{190,210,255}
\definecolor{cB}{RGB}{200,110,100}
\definecolor{cAB}{RGB}{190,120,190}

\tikzset{ %
	V/.style={circle,draw=black,inner sep=0pt,minimum size=8pt,fill=gray!20},
	A/.style={V,fill=cA},
	B/.style={V,fill=cB},
	ab/.style={fill=cAB,shading=true,left color=cA,right color=cB},
	AB/.style={V,ab},
	T/.style={V,rectangle},
	ABT/.style={T,ab},
	As/.style={draw=cA!50,fill=cA!20,fill opacity=0.5},
	Bs/.style={draw=cB!50,fill=cB!20,fill opacity=0.5},
	zeroborder/.style={draw=gray},
	non/.style={draw=gray,loosely dotted},
	term/.style={draw=gray,decorate,decoration={zigzag}},
}

\newcommand{\Azero}{%
\draw[draw=none,fill=cA] (-2,0) arc (180:360:2 and 0.4) -- (-2,0);%
}
\newcommand{\Bzero}{%
\draw[draw=none,fill=cB] (-2,0) arc (180:0:2 and 0.4) -- (-2,0);%
}

\newcommand{\figspace}{\vspace*{-20pt}} 

\newtheorem{theorem}{Theorem}[section]
\newtheorem{lemma}[theorem]{Lemma}
\newtheorem{myclaim}[theorem]{Claim}

\theoremstyle{definition}
\newtheorem{definition}[theorem]{Definition}

\newtheorem{reduction}{Reduction}
\newtheorem{assumption}{Assumption}
\newtheorem{branching}{Branching step}
\newtheorem{reductionstep}[branching]{Reduction step}

\def\cqedsymbol{\ifmmode$\lrcorner$\else{\unskip\nobreak\hfil
\penalty50\hskip1em\null\nobreak\hfil$\lrcorner$
\parfillskip=0pt\finalhyphendemerits=0\endgraf}\fi}

\newcommand{\executeiffilenewer}[3]{%
\ifnum\pdfstrcmp{\pdffilemoddate{#1}}%
{\pdffilemoddate{#2}}>0%
{\immediate\write18{#3}}\fi%
} 
\newcommand{\oct}{\textsc{OCT}\xspace}

\newcommand{\oddcyc}{\textsc{Odd Cycle Transversal}\xspace}
\newcommand{\edgebip}{\textsc{Edge Bipartization}\xspace}
\newcommand{\edgebipcomp}{\textsc{Edge Bipartization Compression}\xspace}
\newcommand{\termsep}{\textsc{Terminal Separation}\xspace}
\newcommand{\vcalp}{\textsc{VC}-above-\textsc{LP}\xspace}

\newcommand{\Oh}{\mathcal{O}}
\newcommand{\Ohstar}{\Oh^\star}
\newcommand{\terms}{\mathcal{T}}
\newcommand{\Acsp}{\mathbf{A}}
\newcommand{\Bcsp}{\mathbf{B}}
\newcommand{\Az}{A^\circ}
\newcommand{\Bz}{B^\circ}
\newcommand{\Aopt}{A^\ast}
\newcommand{\Bopt}{B^\ast}
\newcommand{\wsp}{\alpha}
\newcommand{\inst}{\mathcal{I}}

\newcommand{\defproblemG}[3]{
  \vspace{2mm}
\noindent\fbox{
  \begin{minipage}{0.96\textwidth}
  #1 \\
  {\bf{Input:}} #2  \\
  {\bf{Goal:}} #3
  \end{minipage}
  }
  \vspace{2mm}
}

\title{Edge Bipartization faster than $2^k$}
\author{%
Marcin Pilipczuk\thanks{University of Warsaw, Poland, \texttt{malcin@mimuw.edu.pl}.
Partially supported by the Centre for Discrete Mathematics and its Applications (DIMAP) at the University of Warwick and by Warwick-QMUL Alliance in Advances in Discrete Mathematics and its Applications.}
\and
Micha\l{} Pilipczuk\thanks{University of Warsaw, Poland, \texttt{michal.pilipczuk@mimuw.edu.pl}. Supported by the Polish National Science Centre grant DEC-2013/11/D/ST6/03073 and Foundation for Polish Science via the START stipend program. During the work on these results, Micha\l{} Pilipczuk has been holding a post-doc position of Warsaw Centre of Mathematics and Computer Science.}
\and
Marcin Wrochna\thanks{University of Warsaw, Poland, \texttt{m.wrochna@mimuw.edu.pl}. Supported by the Polish National Science Centre grant DEC-2013/11/D/ST6/03073.}}
\date{}

\begin{document}
\pagenumbering{gobble}
\thispagestyle{empty}

\date{}

\maketitle

\begin{abstract}
In the {\sc{Edge Bipartization}} problem one is given an undirected graph $G$ and an integer~$k$, and the question is whether $k$ edges can be deleted from $G$ so that it becomes bipartite. In~2006, Guo et al.~\cite{eb-2k} proposed an algorithm solving this problem in time $\Oh(2^k\cdot m^2)$; today, this algorithm is a textbook example of an application of the iterative compression technique. Despite extensive progress in the understanding of the parameterized complexity of graph separation problems in the recent years, no significant improvement upon this result has been yet reported. 

We present an algorithm for {\sc{Edge Bipartization}} that works in time $\Oh(1.977^k\cdot nm)$, which is the first algorithm with the running time dependence on the parameter better than $2^k$. To this end, we combine the general iterative compression strategy of Guo et al.~\cite{eb-2k}, the technique proposed by Wahlstr\"om~\cite{magnus} of using a polynomial-time solvable relaxation in the form of a Valued Constraint Satisfaction Problem to guide a bounded-depth branching algorithm, and an involved Measure\&Conquer analysis of the recursion tree.
\end{abstract}

\clearpage

\pagenumbering{arabic}

\section{Introduction}
The \edgebip problem asks, for a given graph $G$ and integer $k$, whether one can turn $G$ into a bipartite graph using at most $k$ edge deletions. Together with its close relative \oddcyc (\oct), where one deletes vertices instead of edges, \edgebip was one of the first problems shown to admit a fixed-parameter (FPT) algorithm using the technique of {\em{iterative compression}}. In a breakthrough paper~\cite{reed-ic} that introduces this methodology, Reed et al. showed how to solve \oct in time $\Oh(3^k\cdot kmn)$\footnote{Even though Reed et al.~\cite{reed-ic} state their running time as $\Oh(4^k\cdot kmn)$, it is not hard to adjust the analysis to show that the algorithm in fact works in time $\Oh(3^k\cdot kmn)$; see e.g.~\cite{huffner,LokshtanovSS09}.}. In fact, this was the first FPT algorithm for \oct. Following this, Guo et al.~\cite{eb-2k} applied iterative compression to show fixed-parameter tractability of several closely related problems. Among other results, they designed an algorithm for \edgebip with running time $\Oh(2^k\cdot m^2)$. Today, both the algorithms of Reed et al. and of Guo et al. are textbook examples of the iterative compression technique~\cite{our-book,downey-fellows}.

Iterative compression is in fact a simple idea that boils down to an algorithmic usage of induction. In case of \edgebip, we introduce edges of $G$ one by one, and during this process we would like to maintain a solution $F$ to the problem, i.e., $F\subseteq E(G)$ is such that $|F|\leq k$ and $G-F$ is bipartite. When the next edge $e$ is introduced to the graph, we observe that $F\cup \{e\}$ is a solution of size at most $k+1$, that is, at most one too large. Then the task reduces to solving \edgebipcomp: given a graph $G$ and a solution that exceeds the budget by at most one, we are asked to find a solution that fits into the budget. 

Surprisingly, this simple idea leads to great algorithmic gains, as it reduces the matter to a cut problem. Guo et al.~\cite{eb-2k} showed that a simple manipulation of the instance reduces \edgebipcomp to the following problem that we call \termsep: We are given an undirected graph $G$ with a set $\terms$ of $k+1$ disjoint pairs of terminals, where each terminal is of degree $1$ in $G$. The question is whether one can color one terminal of every pair white and the second black in such a way that the minimum edge cut between white and black terminals is at most $k$. Thus, the algorithm of Guo et al.~\cite{eb-2k} boils down to trying all the $2^{k+1}$ colorings of terminals and solving a minimum edge cut problem. For $\oct$, we similarly have a too large solution $X\subseteq V(G)$ of size $k+1$, and we are looking for a partition of $X$ into $(L,R,Z)$, where the size of the minimum vertex cut between $L$ and $R$ in $G-Z$ is at most $k-|Z|$. Thus it suffices to solve $3^{k+1}$ instances of a flow problem.

The search for FPT algorithms for cut problems has been one of the leading directions in parameterized complexity in the recent years. Among these, \oddcyc and \edgebip play a central role; see for instance~\cite{eb-2k,huffner,kebab,eb-approx,KolayP0S16,oct-kernel,LP-guided-dl,reed-ic} and references therein. Of particular importance is the work of Kratsch and Wahlstr\"om~\cite{oct-kernel}, who gave the first (randomized) polynomial kernelization algorithms for \oddcyc and \edgebip. The main idea is to encode the cut problems that arise when applying iterative compression into a matroid with a representation that takes small space. The result of Kratsch and Wahlstr\"om sparked a line of further work on applying matroid methods in parameterized complexity.


Another thriving area in parameterized complexity is the {\em{optimality program}}, probably best defined by Marx in~\cite{marx-opt}. The goal of this program is to systematically investigate the optimum complexity of algorithms for parameterized problem by proving possibly tight lower and upper bounds. For the lower bounds methodology, the standard complexity assumptions used are the {\em{Exponential Time Hypothesis}} ({\em{ETH}}) and the {\em{Strong Exponential Time Hypothesis}} ({\em{SETH}}). In the recent years, the optimality program has achieved a number of successes. For instance, under the assumption of SETH, we now know the precise bases of exponents for many classical problems parameterized by treewidth~\cite{ham-tw-opt,cut-and-count,tw-opt}. To explain the complexity of fundamental parameterized problems for which natural algorithms are based on dynamic programming on subsets, Cygan et al.~\cite{secoco} introduced a new hypothesis resembling SETH, called the {\em{Set Cover Conjecture}} ({\em{SeCoCo}}). See~\cite{bulletin,marx-opt} for more examples.

For our techniques, the most important is the line of work of Guillemot~\cite{miszczu}, Cygan et al.~\cite{above-LP}, Lokshtanov et al.~\cite{LP-guided-dl}, and Wahlstr\"om~\cite{magnus} that developed a technique for designing parameterized algorithm for cut problems called {\em{LP-guided branching}}. The idea is to use the optimum solution to the linear programming relaxation of the considered problem in order to measure progress. Namely, during the construction of a candidate solution by means of a backtracking process, the algorithm achieves progress not only when the budget for the size of the solution decreases (as is usual in branching algorithms), but also when the lower bound on the optimum solution increases. Using this concept, Cygan et al.~\cite{above-LP} showed a $2^k n^{\Oh(1)}$-time algorithm for {\sc{Node Multiway Cut}}. Lokshtanov et al.~\cite{LP-guided-dl} further refined this technique and applied it to improve the running times of algorithms for several important cut problems. In particular, they obtained a $2.3146^k n^{\Oh(1)}$-time algorithm for \oddcyc, which was the first improvement upon the classic $\Oh(3^k\cdot kmn)$-time algorithm of Reed et al.~\cite{reed-ic}. From the point of view of the optimality program, this showed that the base $3$ of the exponent was not the final answer for \oddcyc.

In works~\cite{above-LP,LP-guided-dl} it was essential that the considered linear programming relaxation is half-integral, which restricts the applicability of the technique. Recently, Wahlstr\"om~\cite{magnus} proposed to use stronger relaxations in the form of certain polynomial-time solvable Valued Constraint Satisfaction Problems (VCSPs). Using this idea, he was able to show efficient FPT algorithms for the node and edge deletion variants of {\sc{Unique Label Cover}}, for which natural LP relaxations are not half-integral.

Despite substantial progress on the node deletion variant, for \edgebip there has been no improvement since the classic algorithm of Guo et al.~\cite{eb-2k} that runs in time $\Oh(2^k\cdot m^2)$. 
The main technical contribution of Lokshtanov et al.~\cite{LP-guided-dl} is a $2.3146^k n^{\Oh(1)}$-time algorithm for \textsc{Vertex Cover}
parameterized by the excess above the value of the LP relaxation (\vcalp); the algorithm for \oct is a corollary of this result due to folklore reductions
from \oct to \vcalp via the \textsc{Almost 2-SAT} problem. Thus the algorithm for \oct in fact relies on the LP relaxation for the \textsc{Vertex Cover} problem,
which has very strong combinatorial properties; in particular, it is half-integral. 
No such simple and at the same time strong relaxation is available for \edgebip.
The natural question stemming from the optimality program, whether the $2^k$ term for \edgebip could be improved, was asked repeatedly in the parameterized complexity community, for example by Daniel Lokshtanov at WorKer'13~\cite{open-update},
repeated later at~\cite{open-schl}.


\paragraph*{Our results and techniques.}

In this paper we answer this question in affirmative by proving the following theorem.

\begin{theorem}\label{thm:main}
\edgebip{} can be solved in time $\Oh(1.977^k \cdot nm)$.
\end{theorem}

Thus, the $2$ in the base of the exponent is not the ultimate answer for \edgebip. 

To prove Theorem~\ref{thm:main}, we pursue the approach proposed by Guo et al.~\cite{eb-2k} and use iterative compression to reduce solving \edgebip to solving \termsep (see Section~\ref{sec:prelims} for a formal definition of the latter). This problem has two natural parameters: $|\terms|$, the number of terminal pairs, and $p$, the bound on the size of the cut between white and black terminals. The approach of Guo et al. is to use a simple $\Oh(2^{|\terms|}\cdot pm)$ algorithm that tries all colorings of terminal pairs and computes the size of a minimum cut between the colors. 

The observation that is crucial to our approach is that one can express \termsep as a very restricted instance of the {\sc{Edge Unique Label Cover}} problem. More precisely, in this setting the task is to assign each vertex of $G$ a label from $\{\Acsp,\Bcsp\}$. Pairs of $\terms$ present hard (of infinite cost) inequality constraints between the labels of terminals involved, while edges of $G$ present soft (of unit cost) equality constraints between the endpoints. The goal is to minimize the cost of the labeling, i.e., the number of soft constraints broken. An application of the results of Wahlstr\"om~\cite{magnus} (with the further improvements of Iwata, Wahlstr\"{o}m, and Yoshida~\cite{magnus-arxiv} regarding linear dependency on the input size) immediately gives an $\Oh(4^p\cdot m)$ algorithm for \termsep.

Thus, we have in hand two substantially different algorithms for $\termsep$. If we plug in $|\terms|=k+1$ and $p=k$, as is the case in the instance that we obtain from \edgebipcomp, then we obtain running times $\Oh(2^k\cdot km)$ and $\Oh(4^k\cdot m)$, respectively. The idea now is that these two algorithms present two complementary approaches to the problem, and we would like to combine them to solve the problem more efficiently. To this end, we need to explain more about the approach of Wahlstr\"om~\cite{magnus}.

The algorithm of Wahlstr\"om~\cite{magnus} is based on measuring the progress by means of the optimum solution to the relaxation of the problem in the form of a Valued CSP instance. In our case, this relaxation has the following form: We assign each vertex a label from $\{\bot,\Acsp,\Bcsp\}$, where $\bot$ is an additional marker that should be thought of as {\em{not yet decided}}. The hard constraints have zero cost only for labelings $(\Acsp,\Bcsp)$, $(\Bcsp,\Acsp)$ and $(\bot,\bot)$, and infinite cost otherwise. The soft constraints have cost $0$ for equal labels on the endpoints, $1$ for unequal from $\{\Acsp,\Bcsp\}$, and $\frac{1}{2}$ when exactly one endpoint is assigned $\bot$. Based on previous results of Kolmogorov, Thapper, and \v{Z}ivn\'y~\cite{thapper-zivny}, Wahlstr\"om observed that this relaxation is polynomial-time solvable, and moreover it is {\em{persistent}}: whenever the relaxation assigns $\Acsp$ or $\Bcsp$ to some vertex, then it is safe to perform the same assignment in the integral problem (i.e., only with the ``integral'' labels $\Acsp,\Bcsp$). The algorithm constructs an integral labeling by means of a backtracking process that fixes the labels of consecutive vertices of the graph. During this process, it maintains an optimum solution to the relaxation that is moreover {\em{maximal}}, in the sense that one cannot extend the current labeling by fixing integral labels on some undecided vertices without increasing the cost. This can be done by dint of persistence and polynomial-time solvability: we can check in polynomial time whether a non-trivial extension exists, and then it is safe to fix the labels of vertices that get decided. Thus, when the algorithm considers the next vertex $u$ and branches into two cases, fixing label $\Acsp$ or $\Bcsp$ on it, the optimum cost of the relaxation increases by at least $\frac{1}{2}$ in each branch. Hence the recursion tree can be pruned at depth $2p$, and we obtain an $4^p n^{\Oh(1)}$-time algorithm.

Our algorithm for \termsep applies a similar branching strategy, where at each point we maintain some labeling of the vertices with $\Acsp$, $\Bcsp$, and $\bot$ (undecided). Every terminal pair is either already {\em{resolved}} (assigned $(\Acsp,\Bcsp)$ or $(\Bcsp,\Acsp)$), or {\em{unresolved}} (assigned $(\bot,\bot)$). Using the insight of Wahlstr\"om we can assume that this labeling is maximal.  Intuitively, we look at the unresolved pairs from $\terms$ and try to identify a pair $(s,t)$ for which branching into labelings $(\Acsp,\Bcsp)$ and $(\Bcsp,\Acsp)$ leads to substantial progress. Here, we measure the progress in terms of a potential $\mu$ that is a linear combination of three components:
\begin{itemize}
\item $t$, the number of unresolved terminal pairs;
\item $k$, the current budget for the cost of the sought integral solution; 
\item $\nu$, the difference between $k$ and the cost of the current solution to the relaxation.
\end{itemize}
These ingredients are taken with weights $\wsp_t = 0.59950$, $\wsp_\nu = 0.29774$, and $\wsp_k = 1-\wsp_t-\wsp_\nu = 0.10276$.
Thus, the largest weight is put on the progress measured in terms of the number of resolved terminal pairs. Indeed, we want to argue that if we can identify a possibility of recursing into two instances, where in each of them at least one new terminal pair gets resolved, but in one of them we resolve two terminal pairs, then we can pursue this branching step. 

Therefore, we are left with the following situation: when branching on any terminal pair, only this terminal pair gets resolved in both branches. Then the idea is to find a branching step where the decrease of the auxiliary components of the potential, namely $\nu$ and $k$, is significant enough to ensure the promised running time of the algorithm. Here we apply an extensive combinatorial analysis of the instance to show that finding such a branching step is always possible. In particular, our analysis can end up with a branching not on a terminal pair, but on the label of some other vertex; however, we make sure that in both branches some terminal pair gets eventually resolved. Also, in some cases we localize a part of the input that can be simplified (a {\em{reduction step}}), and then the analysis is restarted. 

To sum up, we would like to highlight two aspects of our contribution.
First, we answer a natural question stemming from the optimality program, showing that $2^k$ is not the final dependency on the parameter for \edgebip.
Second, our algorithm can be seen as a ``proof of concept'' that the LP-guided branching technique, even in the more abstract variant of Wahlstr\"{o}m~\cite{magnus},
can be combined with involved Measure\&Conquer analysis of the branching tree.
Note that in the past Measure\&Conquer and related techniques led to rapid progress in the area of moderately-exponential algorithms~\cite{fomin-kratsch}.

We remark that the goal of the current paper is clearly improving $2^k$ factor, and not optimizing the dependence of the running time on the input size. However, we do estimate it. Using the tools prepared by Iwata, Wahlstr\"om, and Yoishida~\cite{magnus-arxiv}, we are able to implement the algorithm so that it runs in time $\Oh(1.977^k\cdot nm)$. Naively, this seems like an improvement over the algorithm of Guo et al.~\cite{eb-2k} that had quadratic dependence on $m$, however this is not the case. We namely use the recent approximation algorithm for \edgebip of Kolay et al.~\cite{eb-approx} that in time $\Oh(k^{\Oh(1)}\cdot m)$ either returns a solution $F^{\textrm{apx}}$ of size at most $\Oh(k^2)$, or correctly concludes that there is no solution of size $k$. Then we start iterative compression from $G-F^{\textrm{apx}}$ and introduce edges of $F^{\textrm{apx}}$ one by one, so we need to solve the \termsep problem only $\Oh(k^2)$ times. In our case each iteration takes time $\Oh(1.977^k\cdot nm)$, but for the approach of Guo et al. it would take time $\Oh(2^k\cdot km)$. Thus, by using the same idea based on~\cite{eb-approx}, the algorithm of Guo et al. can be adjusted to run in time $\Oh(2^k\cdot k^3m)$. It is just that the newer algorithm of Kolay et al.~\cite{eb-approx} was not known at the time of writing~\cite{eb-2k}.

\paragraph*{Organization of the paper.} In Section~\ref{sec:prelims} we give background on iterative compression and the VCSP-based tools borrowed from~\cite{magnus-arxiv,magnus}. In particular, we introduce formally the \termsep problem and reduce solving \edgebip to it. In Section~\ref{sec:main} we set up the Measure\&Conquer machinery that will be used by our branching algorithm, and we introduce preliminary reductions. In Section~\ref{sec:excess} we prove some auxiliary results on {\em{low excess set}}, which is the key technical notion used in our combinatorial analysis. Finally, we present the whole algorithm in Section~\ref{sec:miesko}. Section~\ref{sec:conc} is devoted to some concluding remarks and open problems.

\section{Preliminaries}\label{sec:prelims}
\subsection{Graph notation}

For all standard graph notation, we refer to the textbook of Diestel~\cite{diestel}.
For the input instance $(G,k)$ of \edgebip{}, we denote $n = |V(G)|$ and $m = |E(G)|$.
As isolated vertices are irrelevant for the \edgebip{} problem, we assume that
$G$ does not contain any such vertices, and hence $n = \Oh(m)$.

\subsection{Cuts and submodularity}

As edge cuts in a graph are the main topic of this work, let us introduce some convenient notation.
In all graphs in this paper we allow multiple edges, but not loops, as they are irrelevant for the problem.
For a graph $G$ and two disjoint vertex sets $A,B \subseteq V(G)$, by $E_G(A,B)$ we denote the set of edges
with one endpoint in $A$ and the second endpoint in $B$. If any of the sets $A$ or $B$ is a singleton, say $A = \{a\}$,
we write $E_G(a,B)$ instead of $E_G(\{a\},B)$. We drop the subscript if the graph $G$ is clear from the context.

For a set $A \subseteq V(G)$, we denote $d(A) = |E(A,V(G) \setminus A)|$.
It is well known that the $d(\cdot)$ function is submodular, that is, for every $A,B \subseteq V(G)$ it holds that
\begin{equation}\label{eq:submodularity}
d(A) + d(B) \geq d(A \cap B) + d(A \cup B).
\end{equation}
In fact, a study of the proof of~\eqref{eq:submodularity} allows us to state the difference in the inequality.
\begin{equation}\label{eq:submodularity-err}
d(A) + d(B) = d(A \cap B) + d(A \cup B) + 2|E(A \setminus B,B \setminus A)|.
\end{equation}
Since $d(\cdot)$ is symmetric (i.e., $d(A) = d(V(G) \setminus A)$), by applying~\eqref{eq:submodularity}
to $A$ and the complement of $B$, we obtain a property sometimes called \emph{posimodularity}: for every $A,B \subseteq V(G)$ it holds that
\begin{equation}\label{eq:posimodularity}
d(A) + d(B) \geq d(A \setminus B) + d(B \setminus A).
\end{equation}
Or, with the error term:
\begin{equation}\label{eq:posimodularity-err}
d(A) + d(B) = d(A \setminus B) + d(B \setminus A) + 2|E(A \cap B,V(G) \setminus (A \cup B))|.
\end{equation}

\subsection{Iterative compression and the compression variant}

Let $(G,k)$ be an input \edgebip{} instance.
The opening step of our algorithm is the standard usage of 
iterative compression. We start by applying the
approximation algorithm of~\cite{eb-approx} that, given $(G,k)$,
in time $\Oh(k^{\Oh(1)} m)$ either correctly concludes that it is a no-instance,
or produces a set $Z \subseteq E(G)$ of size $\Oh(k^2)$ such that $G-Z$ is bipartite.

Let $Z$ be the obtained set, $|Z| = r = \Oh(k^2)$ and $Z = \{e_1,e_2,\ldots,e_r\}$.
Let $G_i = G-\{e_{i+1},e_{i+2},\ldots,e_r\}$ for $0 \leq i \leq r$;
note that $G_r = G$ while $G_0 = G-Z$, which is bipartite.
Our algorithm, iteratively for $i=0,1,\ldots,r$, computes
a solution $X_i$ to the instance $(G_i,k)$, or concludes that
no such solution exists. Clearly, since $G_i$ is a subgraph of $G$,
if we obtain the latter conclusion for some $i$, we can report that $(G,k)$ is a no-instance.

For $i=0$, $G_0=G-Z$ is bipartite, thus $X_0 = \emptyset$ is a solution.
Consider now an instance $(G_i,k)$ for $1 \leq i \leq r$, and assume
that a solution $X_{i-1}$ has already been computed. 
Let $X_i' = X_{i-1} \cup \{e_i\}$.
If $|X_i'| \leq k$, we can take $X_i = X_i'$ and continue.
Otherwise, we can make use of the structural insight given by 
the set $X_i'$ and solve the following problem.

\defproblemG{\edgebipcomp}{A graph $G$, and integer $k$, and a set $X' \subseteq E(G)$ of size $k+1$ such that $G-X'$ is bipartite.}{Compute a set $X \subseteq E(G)$ of size at most $k$ such that $G-X$ is bipartite, or conclude that no such set exists.}

If we could efficiently solve an \edgebipcomp{} instance $(G_i,k,X_i')$, we can take the output solution as $X_i$ and proceed to the next step of this iteration.
Consequently, it suffices to prove the following theorem.

\begin{theorem}\label{thm:comp}
\edgebipcomp{} can be solved in time $\Oh(c^k nm)$ for some constant $c < 1.977$.
\end{theorem}

\subsection{The terminal separation problem}

Following the algorithm of~\cite{eb-2k}, we phrase \edgebipcomp{} as
a separation problem.

Consider a graph $G$ with a family $\terms$ of pairs of terminals in $G$.
A pair $(A,B)$ with $A,B \subseteq V(G)$ is a \emph{terminal separation}
if $A \cap B = \emptyset$ and, for every terminal pair $P$, either one of 
the terminals in $P$ belongs to $A$ and the second to $B$, or $P \subseteq V(G) \setminus (A \cup B)$. 
A terminal separation $(A,B)$ is \emph{integral} if $A \cup B = V(G)$.\footnote{The word \emph{integral} stems from the fact that an integral separation corresponds to a solution to the relaxed \termsep{} problem that actually does
not use the relaxed value $\bot$. In fact, it also corresponds to an integral solution of an LP formulation underlying the algorithmic results of~\cite{thapper-zivny}.}
A terminal separation $(A',B')$ \emph{extends} $(A,B)$
if $A \subseteq A'$ and $B \subseteq B'$. 
The \emph{cost} of a terminal separation $(A,B)$ is defined
as $c(A,B) = (d(A) + d(B))/2$. Note that if $(A,B)$ is integral,
then we have $c(A,B) = d(A) = d(B)$.

We will solve the following separation problem.

\defproblemG{\termsep}{A graph $G$ with a set of terminal pairs $\terms$ such that the pairs are pairwise disjoint and every terminal is of degree at most one in $G$; a terminal separation $(\Az,\Bz)$; and an integer $k$.}{Find an integral terminal separation $(A,B)$ extending $(\Az,\Bz)$ of cost at most $k$, or report that no such separation exists.}

\begin{lemma}\label{lem:compred}
Given an \edgebipcomp{} instance $(G,k,X')$, one can in polynomial
time compute an equivalent instance $(G',\terms,(\Az,\Bz),k')$ of
\termsep, such that $|E(G')| = |E(G)| + \Oh(|X'|)$, 
$|V(G')| = |V(G)| + \Oh(|X'|)$, $|\terms| = |X'|$, 
$\Az=\Bz = \emptyset$, and $k' = k$.
\end{lemma}
\begin{proof}
Let $G'$ be the graph obtained from $G$ by replacing every edge $uv$ in $X'$ with two new vertices $s,t$ and two pendant edges $us,vt$.
Let $\terms$ be the set of vertex pairs $\{s,t\}$ created this way, $\Az=\Bz=\emptyset, k'=k$.
We show this constructed instance is equivalent to the original instance.

If the constructed instance is a yes-instance, let $A,B$ be an integral terminal separation of cost at most $k$.
Take $X=E(A,B)$, and then, for every edge in $X$ incident to a terminal $s$, replace this edge with $uv$, where $uv$ is the edge of $X'$ for which the terminal pair containing $s$ was created.
We claim $X$ is a solution to the original instance (clearly $|X| \leq |E(A,B)| \leq k$).
Indeed, let $L',R'$ be a bipartition of $G-X'$.
We show that $(L'\cap A) \cup (R' \cap B), (R'\cap A) \cup (L' \cap B)$ gives a bipartition of $G-X$.
Suppose that, to the contrary, there is an edge $uv$ in $G-X$ with both endpoints in $(L'\cap A) \cup (R' \cap B)$ (the case of $(R'\cap A) \cup (L' \cap B)$ being symmetrical).
Since all edges with one endpoint in $A\cap V(G)$ and the other in $B \cap V(G)$ were deleted by $X$, we may assume $uv$ is an edge with both endpoints in $L' \cap A$ (the case of $R' \cap B$ being symmetrical).

Since $L'$ is one side of a bipartition of $G-X'$, $uv$ must be an edge in $X'$.
Let $us, vt$ be the corresponding edges to terminals created by the construction.
Since both $u$ and $v$ are in $A$ and exactly one of $s,t$ is in $A$ (as $A,B$ is a terminal separation), one of $us,vt$ must be an edge in $E(A,B)$.
Hence $uv$ was deleted by $X'$, a contradiction.

For the other side, if the original instance is a yes-instance, let $X$ be a set of at most $k$ edges such that $G-X$ has a bipartition $L,R$.
By taking $X$ to be minimal, we can assert that every edge in $X$ has both endpoints on one side of this bipartition.
Let $A$ contain all vertices in $(L \cap L') \cup (R \cap R')$, let $B$ contain all vertices of $(L \cap R')\cup (R \cap L')$.
For every edge $uv$ in $X$, add the corresponding terminal vertices $s,t$ to $A$ and $B$ so that $s$ is in the same set as $u$ and $t$ is in the other one.
Clearly $(A,B)$ is an integral terminal separation extending $(\emptyset,\emptyset)$, it suffices to show that $|E_{G'}(A,B)| \leq |X|$.

Let $e \in E_{G'}(A,B)$.
If neither endpoint of $e$ is a terminal, then let $e=uv$ for $u\in (L \cap L') \cup (R \cap R')$ and $v\in (L \cap R')\cup (R \cap L')$.
Since all edges with both endpoints in $L'$ were in $X'$ and hence replaced by edges to terminals in $G'$, it cannot be that both $u$ and $v$ are in $L'$.
Similarly, for $R'$, so let us assume that $u\in L'$ and $v\in R'$ (the other case is symmetrical).
Then $u\in (L\cap L')$ and $v\in (L \cap R')$, which means both $u$ and $v$ are on the same side of the $(L,R)$ bipartition of $G-X$ but on different sides of the $(L',R')$ bipartition of $G-X'$.
Hence $e\in X \setminus X'$.

Otherwise, assume that some endpoint of $e$ is a terminal.
Recall that we have defined $A$ and $B$ in such a way that the edge $us$ for a terminal pair $\{s,t\}$ is never
cut. Hence, let $e=vt$, and without loss of generality assume that
$v \in  (L \cap L') \cup (R \cap R')$.
(Here we keep the notation that  $s,t$ are two terminal with edges $us$ and $vt$ replacing in the construction an edge $uv$ in $X'$.)
In particular $v\in A$, so $t\in B$. 
By construction of the separation $(A,B)$, it must be that $s$ is in the same side as $u$ and $t$ is on the other side, so $s,u\in A$.
Since $u$ is not a terminal, $u,v\in (L \cap L') \cup (R \cap R')$.
Hence $uv \in X' = E_G(L',L')\cup E_G(R',R')$ implies that $u$ and $v$ are on the same side of the $(L',R')$ partition, thus also on the same side of the $(L,R)$ partition and hence $uv \in X \cap X'$.
Note also that of the two edges that replace $uv$ in the construction, only $e$ is in $E_{G'}(A,B)$, because $s,u \in A$.

Hence every edge in $E_{G'}(A,B)$ is either an edge in $X\setminus X'$ or an edge to a terminal uniquely corresponding to an edge in $X \cap X'$, which implies $|E_{G'}(A,B)| \leq |X|$, concluding the other side of the proof.

\end{proof}

We say that a terminal pair $P$ is \emph{resolved} in a \termsep{} instance
$(G,\terms,(\Az,\Bz),k)$ if $P \subseteq \Az \cup \Bz$, and \emph{unresolved}
otherwise (i.e., $P \subseteq V(G) \setminus (\Az \cup \Bz)$).
Thus, our goal is to design an efficient branching algorithm
for \termsep{}, with parameters being $k$, the excess in the cutset
$k - c(\Az,\Bz)$, and the number $t$ of unresolved terminal pairs.
A precise statement of the result can be found in Section~\ref{sec:main},
where an appropriate progress measure is defined.

\subsection{LP branching}

The starting point in designing an algorithm for \termsep{} using the
aforementioned parameters is the generic LP branching framework
of Wahlstr\"{o}m~\cite{magnus}.

Observe that one can phrase an instance of \termsep{} as a Valued CSP instance,
with vertices being variables over the domain $\{\Acsp,\Bcsp\}$, edges being
soft (unit cost) equality constraints, terminal pairs being hard
(infinite or prohibitive cost) inequality constraints, while membership
in $\Az$ or $\Bz$ translates to hard unary constraints on vertices
of $\Az \cup \Bz$. Observe that this Valued CSP instance is in fact
a \textsc{Unique Label Cover} instance over binary alphabet, with additional
hard unary constraints.

In a relaxed instance, we add to the domain the ``do not know'' value $\bot$,
and extend the cost function for an edge (equality) constraint to be $0$
for both endpoints valued $\bot$, and $\frac{1}{2}$ for exactly one
endpoint valued $\bot$; the hard inequality constraints on terminal pairs
 additionally allows both terminals to be valued $\bot$.
Observe that now feasible solutions $f: V(G) \to \{\bot,\Acsp,\Bcsp\}$ to
the so-defined instance are in one-to-one correspondence with
terminal separations $(A,B) = (f^{-1}(\Acsp),f^{-1}(\Bcsp))$, and
the cost of $f$ equals exactly $c(f^{-1}(\Acsp), f^{-1}(\Bcsp))$.
Furthermore, as shown by Wahlstr\"{o}m~\cite{magnus},
the cost functions in this relaxation are bisubmodular,
which implies the following two corollaries.

\begin{theorem}[\textbf{persistence} \cite{magnus}]\label{thm:persistence}
Let $(G,\terms,(\Az,\Bz),k)$ be a \termsep{} instance, and
let $(A,B)$ be a terminal separation in $G$ of minimum cost among separations
that extend $(\Az,\Bz)$. 
Then there exists an integral separation $(\Aopt,\Bopt)$ that has
minimum cost among all separations extending $(\Az,\Bz)$, with the additional
property that $(\Aopt,\Bopt)$ extends $(A,B)$.
\end{theorem}

We say that a terminal separation $(A,B)$ is \emph{maximal} if
every other separation extending it has strictly larger cost.

\begin{theorem}[\textbf{polynomial-time solvability}, \cite{magnus,magnus-arxiv}]\label{thm:ptime}
Given a \termsep{} instance $(G,\terms,(\Az,\Bz),k)$ with $c(\Az,\Bz) \leq k$, one
can in $\Oh(k^{\Oh(1)} m)$ time find a maximal terminal separation $(A,B)$ in $G$ that has minimum cost among all separations extending $(\Az,\Bz)$.
\end{theorem}

From Theorems~\ref{thm:persistence} and~\ref{thm:ptime} it follows that, while working on a \termsep{} instance $(G,\terms,(\Az,\Bz),k)$, we can always assume that $(\Az,\Bz)$ is a maximal separation: If that is not the case, we can obtain an extending separation $(A,B)$ via Theorem~\ref{thm:ptime}, and set $(\Az,\Bz) := (A,B)$; the safeness of the last step is guaranteed by Theorem~\ref{thm:persistence}.

We remark here that, in the course of the algorithm, we will often merge sets of vertices in the processed graph. For a nonempty set $X \subseteq V(G)$, the operation of \emph{merging $X$ into a vertex} replaces $X$ with a new vertex $x$, and replaces every edge $uv \in E(X,V(G) \setminus X)$, $u \in X$, $v \notin X$, with an edge $xv$. That is, in this process we do not supress multiple edges while identifying some vertices. However, we do supress loops, as they are irrelevant for the problem.
Consequently, we allow the graph $G$ to have multiple edges, but not loops; we remark that both theorems cited in this section work perfectly fine in this setting as well.

\section{The structure of the branching algorithm}\label{sec:main}
In this section we describe the structure of the branching
algorithm for \termsep{}.
Before we state the main result, we introduce the potential
that will measure the progress made in each branching step.

Let $\inst = (G,\terms,(\Az,\Bz),k)$ be a \termsep{} instance, where
$(A_0,B_0)$ is a maximal terminal separation; we henceforth call such an instance \emph{maximal}.
We are interested in keeping track of the following
partial measures:
\begin{itemize}
\item $t_\inst$ is the number of unresolved terminal pairs;
\item $\nu_\inst = k - c(\Az,\Bz)$;
\item $k_\inst = k$.
\end{itemize}

The $\Oh(2^k km)$-time algorithm used in~\cite{eb-2k} can be interpreted
in our framework as an $\Oh(2^{t_\inst} k_\inst m)$-time algorithm
for \termsep{}, while the generic LP-branching algorithm for \textsc{Edge Unique Label Cover}
of~\cite{magnus} can be interpreted as an $\Oh(4^{\nu_\inst} m)$-time
algorithm. As announced in the intro, our main goal is to blend these two algorithms, by analysing the cases where both these algorithms perform badly.

An important insight is that all these inefficient cases happen when $\Az$ and $\Bz$ increase their common boundary. If this is the case, a simple reduction rule is applicable that also reduces the allowed budget $k$; in some sense, with this reduction rule the budget $k$ represents the yet undetermined part of the boundary between $\Aopt$ and $\Bopt$ in the final integral solution $(\Aopt,\Bopt)$.
For this reason, we also include the budget $k$ in the potential.

Formally, we fix three constants $\wsp_t = 0.59950$, $\wsp_\nu = 0.29774$, and
$\wsp_k = 1-\wsp_t-\wsp_\nu = 0.10276$ and define a potential of an instance $\inst$
as
$$\mu_\inst = \wsp_t \cdot t_\inst + \wsp_\nu \cdot \nu_\inst + \wsp_k \cdot k_\inst.$$

Our main technical result, proved in the remainder of this paper, is the following.
\begin{theorem}\label{thm:termsep}
A \termsep{} instance $\inst$ can be solved in time
$\Oh(c^{\mu_\inst} nm)$ for some $c < 1.977$.
\end{theorem}
Observe that if $\inst$ is an instance output by the reduction of Lemma~\ref{lem:compred}, then $t_\inst = |X'| = k+1$, $\nu_\inst = k$ since $\Az=\Bz=\emptyset$, and
$k_\inst = k$. Consequently, $\mu_\inst < k+1$, and Theorem~\ref{thm:main} follows from Theorem~\ref{thm:termsep}.

The algorithm of Theorem~\ref{thm:termsep} follows a typical outline of a recursive branching algorithm. At every step, the current instance
is analyzed, and either it is reduced, or some two-way branching step is performed. The potential $\mu_\inst$ is used to measure the progress of the 
algorithm and to limit the size of the branching tree.

\subsection{Reductions}

We use a number of reductions in our algorithm.
Every reduction decreases $|V(G)|+|\terms|+k$, and after any application of any reduction
we re-run Theorem~\ref{thm:ptime} to ensure that the considered instance is maximal.

The first one is the trivial termination condition.

\begin{reduction}[Terminator Reduction]
If $k_\inst < 0$ or $\nu_\inst < 0$, then we terminate the current branch
with the conclusion that there is no solution.
If $(\Az,\Bz)$ is integral, return it as a solution.
\end{reduction}

Observe that if all terminals are resolved, then
both $(\Az, V(G)\setminus \Az)$
and $(V(G) \setminus \Bz,\Bz)$ are integral separations,
and one of them is of cost at most $c(\Az,\Bz)$. Consequently,
since $\inst$ is maximal, in fact $(\Az,\Bz)$ is integral.
We infer that if the Terminator Reduction does not trigger,
then there exists at least one unresolved terminal pair, i.e., $t_\inst > 0$.

We now provide the promised reduction of the boundary between $\Az$ and $\Bz$.

\begin{reduction}[Boundary Reduction]
If there exists an edge $ab$ with $a \in \Az$, $b \in \Bz$, delete the edge $ab$ and decrease $k$ by one.
If there exist two edges $va,vb$ with $a \in \Az$, $b \in \Bz$, and $v \notin \Az \cup \Bz$, delete both edges $va$ and $vb$, and decrease $k$ by one.
\end{reduction}

\begin{lemma}\label{lem:boundary-red}
Let $\inst = (G,\terms,(\Az,\Bz),k)$ be a maximal \termsep{} instance, and assume that the Boundary Reduction have been applied once, giving a graph $G'$.
Then $\inst' = (G',\terms,(\Az,\Bz),k-1)$ is a maximal \termsep{} instance, equivalent to $\inst$. Furthermore, $\mu_{\inst'} = \mu_\inst - \wsp_k$.
\end{lemma}

\begin{proof}
Observe that whether $(A,B)$ is a terminal separation extending $(\Az,\Bz)$ does not depend on the instance we are looking at: $\inst$ and $\inst'$
differ only in the edgeset of the graph and the budget. For such a separation, by $c(A,B)$ we denote its cost in $\inst$, and by $c'(A,B)$ its cost in $\inst'$.

We claim that for any terminal separation $(A,B)$ extending $(\Az,\Bz)$ it holds that $c(A,B) = c'(A,B) + 1$.
The claim is straightforward if an edge $ab$ is deleted. For the second case, consider subcases depending on where the vertex $v$ lies.
If $v \in A$, then $d_G(A) = d_{G'}(A) + 1$ due to missing edge $vb$, while if $v \notin A$, then also $d_G(A) = d_{G'}(A) + 1$ due to missing edge $va$.
Symmetrically, $d_G(B) = d_{G'}(B)+1$, which proves the claim. Consequently, the instances $\inst$ and $\inst'$ are equivalent, and $(\Az,\Bz)$ remains
a maximal separation. Furthermore, since $c(\Az,\Bz) = c'(\Az,\Bz)+1$, we have $t_{\inst'} = t_\inst$ and $\nu_{\inst'} = \nu_\inst$, hence
$\mu_{\inst'} = \mu_\inst - \wsp_k$.
\end{proof}

It is easy to observe that the Boundary Reduction can be applied exhaustively in linear time.

In a number of reductions in this section,
in a few places in the analysis of different cases in the branching algorithm, as well as in the reduction rules defined in the next section,
we find a set $X \subseteq V(G)$ of at least two vertices
without any terminals, with at least one vertex of $V(G) \setminus (\Az \cup \Bz)$,
for which we can argue that there exists an integral solution $(\Aopt,\Bopt)$ to $\inst$ of minimum cost
such that $X \subseteq \Aopt$ or $X \subseteq \Bopt$. In this case, we identify $X$ into a single vertex (that belongs to $\Az$ if $X \cap \Az \neq \emptyset$
and to $\Bz$ if $X \cap \Bz \neq \emptyset$), and start from the beginning.

Note that after such reduction $(\Az,\Bz)$ may not be a maximal separation if the contracted set $X$ contains at least one vertex of $\Az \cup \Bz$,
and we need to apply Theorem~\ref{thm:ptime} to extend it to a maximal one.
However, note that the operation of merging vertices only shrinks the space
of all terminal separations, and thus the cost of $(\Az,\Bz)$ cannot
decrease with such a reduction (and, consequently, $\nu_\inst$ cannot increase).

We now introduce four simple rules. The first one reduces clearly superfluous pieces of the graph.

\begin{reduction}[Pendant Reduction]
If there exists a vertex set $X \subseteq V(G) \setminus (\Az \cup \Bz)$ that does not contain any terminal and $|N(X)| \leq 1$, then 
delete $X$ from $G$.

If there exists a vertex set $X \subseteq V(G) \setminus (\Az \cup \Bz)$ that does not contain any terminal and $|N(X)| = 2$,
then let $\lambda$ be the size of the minimum (edge) cut between the two vertices of $N(X)$ in $G[N[X]]$.
If $\lambda \leq k$, then replace $X$ with $\lambda$ edges between the two vertices of $N(X)$,
and otherwise identify $N[X]$ into a single vertex.
\end{reduction}
The safeness of the Pendant Reduction is straightforward: in the first case, for any integral separation $(A,B)$ of the reduced graph, 
one can add $X$ to the set $A$ or $B$ that contains $N(X)$, without increasing the cost of the separation, while in the second
case we can do exactly the same if $N(X)$ belongs to the same side $A$ or $B$, and otherwise we can greedily cut $G[N[X]]$ along the minimum cut
between the vertices of $N(X)$.

Observe that an application of a Pendant Reduction does not merge two terminals and does not 
spoil the invariant that every terminal in $G$ is of degree at
most one. If $|N(X)| \leq 1$, then clearly the deletion of $X$ cannot spoil this property.
Otherwise, if $|N(X)|=2$, then $\lambda$ is at most the degree of any vertex of $N(X)$ in $G[N[X]]$;
in particular, if $N(X)$ contains a terminal, then $\lambda \leq 1$ and the vertices of $N(X)$ are not identified.

This reduction also does not decrease $c(\Az,\Bz)$ (and thus does not increase $\nu_\inst$).
This is clear for $N(X) = \emptyset$. For $|N(X)| = 1$ or $\lambda > k$, it can be modelled as identifying $N[X]$ into a single vertex.
Otherwise, for $|N(X)| = 2$ and $\lambda \leq k$, it can be modelled as identifying the sides of a minimum cut in $G[N[X]]$ between vertices of $N(X)$
onto the corresponding elements of $N(X)$.

Let us now argue that the Pendant Reduction can be applied efficiently.
\begin{lemma}
One can in $\Oh(km)$ time find a set $X$ on which the Pendant Reduction is applicable, or correctly conclude that no such set exists.
\end{lemma}
\begin{proof}
First, compute an auxiliary graph $G'$ from $G$ by adding a clique $K$ on four vertices, and making $K$ fully adjacent to $L := \Az \cup \Bz \cup \terms$.
In this manner, the size of $G'$ is bounded linearly in the size of $G$, while $G'[K \cup L]$ is 3-connected. Compute the decomposition into 3-connected components~\cite{tutte-dekomp,tutte-dekomp2}, which can be done in linear time~\cite{tutte-dekomp-compute}.
It is easy to see that the Pendant Reduction is not applicable if and only if the decomposition consists of a single bag,
and otherwise any leaf bag of the decomposition different than the bag containing $K \cup L$
equals $N[X]$ for some $X$ to which the Pendant Reduction is applicable.
Furthermore, for such a set $X$ with $|N(X)|=2$, one can compute $\min(\lambda, k+1)$ in time $\Oh(km)$ using $\Oh(k)$ rounds of the Ford-Fulkerson algorithm.
\end{proof}

The next three reduction rules consider some special cases of how terminals can lie in the graph.
\begin{reduction}[Lonely Terminal Reduction]
If there exists an unresolved terminal pair $P = \{s,t\}$ such that $s$ is an isolated vertex, delete $P$ from $\terms$ and $V(G)$.
\end{reduction}
The safeness of the Lonely Terminal Reduction follows from the observation that in every terminal separation $(A,B)$ of the reduced
graph, we can always put $t$ on the same side as its neighbor (if it exists) and $s$ on the opposite side.
\begin{reduction}[Adjacent Terminals Reduction]
If there exist two neighboring unresolved terminals $t_1$ and $t_2$, then proceed as follows. If they belong to the same terminal pair,
delete both of them from $G$ and from $\terms$, and reduce $k$ by one.
If they belong to different terminal pairs, say $\{s_1,t_1\}$ and $\{s_2,t_2\}$, then delete both these terminal pairs from $\terms$,
   delete the vertices $t_1$ and $t_2$ from $G$, and add an edge $s_1s_2$.
\end{reduction}
For safeness of the Adjacent Terminals Reduction, first recall that terminals are of degree one in $G$, thus $\{t_1,t_2\}$ is a connected
component of $G$. If they belong to the same terminal pair, the edge $t_1t_2$ always belongs to the solution cut and can be deleted. If they belong to different terminal pairs $\{s_1,t_1\}$ and $\{s_2,t_2\}$,
then the edge $t_1t_2$ is cut by a solution $(A^\ast,B^\ast)$ if and only if $s_1$ and $s_2$ are on different sides of the solution, thus we can just as well account for it by replacing it with an edge $s_1s_2$.
\begin{reduction}[Common Neighbor Reduction]
If there exists an unresolved terminal pair $\{s,t\} \in \terms$, such that $s$ and $t$ share a neighbor $a$, then
delete both terminals from $\terms$ and $G$, and decrease $k$ by one.
\end{reduction}
For safeness of the Common Neighbor Reduction, note that in any solution, exactly one edge $as$ or $at$ is cut.

It is straightforward to check in linear time if any of the last three reductions is applicable, and apply one if this is the case.
It follows from maximality of $(\Az,\Bz)$ and the above safeness arguments that none of these reductions increases $\nu_\inst$:
any potential extension $(A,B)$ of $(\Az,\Bz)$ in the reduced graph can be translated to an extension in the original graph, with a cost larger than $c(A,B)$ by exactly the number of times the budget $k$ has been decreased by the reduction.
Consequently, every application of any of the last three reductions decreases the potential $\mu_\inst$ by at least $\wsp_t$, as each removes at least one terminal pair.

The last reduction is the following.
\begin{reduction}[Majority Neighbour Reduction]
If there exists two non-terminal vertices $u,v \in V(G) \setminus (\Az \cup \Bz)$
such that at least half of the edges incident to $u$ have the second endpoint
in $v$, identify $u$ and $v$.
\end{reduction}
The safeness of the Majority Neighbour Reduction is straightforward: in any integral separation that puts $u$ and $v$ on opposite sides,
changing the side of $u$ does not increase the cost of the separation.
Also, it is straightforward to find vertices $u,v$ for which the Majority Neighbour Reduction applies and execute it in linear time.
Note that, since we require $u,v \notin \Az \cup \Bz$, the considered instance remains maximal.

Two more reduction rules will be introduced in Section~\ref{sec:excess}, where we study sets $A \supseteq \Az$ with 
small $d(A) - d(\Az)$.

\subsection{Branching step}

In every branching step, we identify two terminal separations $(A_1,B_1)$ and $(A_2,B_2)$ extending $(\Az,\Bz)$, and branch into two subcases;
in subcase $i$ we replace $(\Az,\Bz)$ with $(A_i,B_i)$.
We always argue the \emph{correctness} of a branch by showing that there exists a solution $(\Aopt,\Bopt)$
extending $(\Az,\Bz)$ of minimum cost, with the additional property that $(\Aopt,\Bopt)$ extends $(A_i,B_i)$ for some $i=1,2$.
In subcase $i$, we apply the algorithm of Theorem~\ref{thm:ptime} to $(G,\terms,(A_i,B_i),k)$ to obtain a maximal separation
$(\Az_i,\Bz_i)$, and pass the instance $\inst_i = (G,\terms,(\Az_i,\Bz_i),k)$ to a recursive call.

To show the \emph{running time bound} for a branching step, we analyze how the measure $\mu_\inst$ decreases in the subcases,
taking into account the reductions performed in the subsequent recursive calls. More formally,
we say that a branching case \emph{fulfills a branching vector}
$$[t_1,\nu_1,k_1;\ t_2,\nu_2,k_2]$$
if, in subcase $i=1,2$, at least $t_i$ terminal pairs become resolved or reduced with one of the reductions,
the cost of the separation $(\Az_i,\Bz_i)$ grows by at least $\nu_i/2$, and the Boundary Reduction is applied at least $k_i$ times in the instance $(G,\terms,(\Az_i,\Bz_i),k)$.

A branching vector $[t_1,\nu_1,k_1;t_2,\nu_2,k_2]$ is \emph{good} if
$$1.977^{-\wsp_t t_1 - \wsp_\nu \nu_1/2 - \wsp_k k_1} + 1.977^{-\wsp_t t_2 - \wsp_\nu \nu_2/2- \wsp_k k_2} < 1.$$
In other words, if in subcase $i=1,2$, the potential $\mu_\inst$ of the instance $\inst$ decreases by $\delta_i$, then we require that $1.977^{-\delta_1}+1.977^{-\delta_2}<1$.
A standard inductive argument for branching algorithms show that, if in every case we perform a branching step that fulfills some good branching vector,
the branching tree originating from an instance $\inst$ has $\Oh(c^{\mu_\inst})$ leaves for some $c < 1.977$ (so that $c^{\mu_\inst-\delta_1}+c^{\mu_\inst-\delta_2} \leq c^{\mu_\inst}$).
To simplify further exposition, we gather in the next lemma good branching vectors used in the analysis; the fact that they are good can be checked
by direct calculations.
\begin{lemma}\label{lem:good-v}
The following branching vectors are good:
\begin{align*}
[1,1,0;2,1,0] & & [1,1,1;1,2,3] & & [1,2,0;1,3,1] & & [1,1,0;1,4,3] \\
[1,1,2;1,2,2] & & [1,1,1;1,3,2] & & [1,3,0;1,3,0] & & [1,1,0;1,5,2] \\
[1,2,1;1,2,2] & & [1,1,1;1,4,1] & &  & &
\end{align*}
\end{lemma}
Let us stop here to comment that the vectors in Lemma~\ref{lem:good-v} explain our choice of constants
$\wsp_t$, $\wsp_\nu$, $\wsp_k$. The constant $\wsp_t$ is sufficiently large to make the vector $[1,1,0;2,1,0]$ good;
intuitively speaking, we are always done when in one branch we manage to resolve or reduce at least two terminal pairs.
The choice of $\wsp_\nu$ and $\wsp_k$ represents a very delicate tradeoff that makes both $[1,1,1;1,2,3]$ and
$[1,2,0;1,3,1]$ good; note that setting $\wsp_\nu = 1-\wsp_t$ and $\wsp_k = 0$ makes the first vector not good,
while setting $\wsp_\nu = 0$ and $\wsp_k = 1-\wsp_t$ makes the second vector not good.%
\footnote{One can observe that the goodness of all other vectors mentioned in Lemma~\ref{lem:good-v},
can be easily deduced from the goodness of the first four vectors, by using the fact that $\wsp_t \geq \wsp_\nu \geq \wsp_k$
and that a branching vector cannot stop being good if one moves weight from the ``heavier'' side to the ``lighter'' one.}
In fact, arguably the possibility of a tradeoff that makes both the second and the third vector of Lemma~\ref{lem:good-v} good at the same time
is one of the critical insights in our work.

\subsection{Running time bound}

In the subsequent sections, we will only argue that 
\begin{enumerate}
\item every single application of a reduction or
a branching step is executed in $\Oh(k^{\Oh(1)} m)$
time;
\item every reduction either terminates or reduces
$k + |\terms| + |V(G)|$ by at least one; note that this is true for the reductions defined so far;
\item every branching step is correct and fulfills one of the good vectors
mentioned in Lemma~\ref{lem:good-v}.
\end{enumerate}
Observe that these properties guarantee correctness and the claimed running
time of the algorithm.

In a number of places in the branching algorithm, the algorithm \emph{attempts} some branching $(A_1,B_1),(A_2,B_2)$, and withdraws
this decision if the measure decrease is too small. A naive implementation of such behaviour would lead to an additional $n$ factor in the running
time bound, as exhaustive application of our reduction rules may take $\Oh(k^{\Oh(1)} nm)$ time, only to be later withdrawn. To maintain the $\Oh(nm)$ polynomial
factor in our running time bound, we restrict such attempts to only the following procedure: for $i=1,2$, we apply Theorem~\ref{thm:ptime} to
obtain a minimum cost extension $(\Az_i,\Bz_i)$ of $(A_i,B_i)$, and report:
\begin{enumerate}
\item the number of terminal pairs contained in $(\Az_i \cup \Bz_i) \setminus (\Az \cup \Bz)$, i.e., the immediate decrease in $t_\inst$;
\item the difference $c(\Az_i,\Bz_i) - c(\Az,\Bz)$, i.e., the immediate decrease in $\nu_\inst$;
\item the number of immediately applicable Boundary Reductions, defined as follows:
$$\rho_i := |E(\Az_i,\Bz_i)| + \sum_{v \in V(G) \setminus (\Az_i \cup \Bz_i)} \min(|E(v,\Az_i)|,|E(v,\Bz_i)|).$$
\end{enumerate}
Clearly, the aforementioned numbers are computable in $\Oh(k^{\Oh(1)}m)$ time.

\section{Low excess sets}\label{sec:excess}

Let $\inst = (G,\terms,(\Az,\Bz),k)$ be a maximal \termsep{} instance.
A set $A \subseteq V(G)$ is an \emph{$\Az$-extension} if $\Az \subseteq A \subseteq V(G) \setminus \Bz$. It is \emph{terminal-free} if $A \setminus \Az$ does not
contain any terminal.
We denote by $\Delta(A) := d(A)-d(\Az)$ the \emph{excess} of an $\Az$-extension $A$.
An $\Az$-extension $A$ is \emph{compact} if $A \setminus \Az$ is connected and $E(A \setminus \Az, \Az) \neq \emptyset$.

In this section we consider extensions of small excess, and show that their structure can be reduced to have a relatively simple picture.
While in this section we focus on supersets of the set $\Az$,
by symmetry the same conclusion holds if we swap the roles of $\Az$ and $\Bz$.
In our algorithm, we exhaustively apply the reduction rules defined in this section both for the $A$-side and $B$-side of the separation $(\Az,\Bz)$.

Before we start, let us first observe that we can efficiently enumerate all maximal sets of particular constant excess.
\begin{lemma}\label{lem:enum-excess}
For every fixed constant $r$, one can in $\Oh(k^{\Oh(1)}(n+m))$ time enumerate
all inclusion-wise maximal compact $\Az$-extensions of excess at most $r$.
\end{lemma}
\begin{proof}
Our algorithm will in fact enumerate all compact $\Az$-extensions $A$ of excess at most $r$
with the property that every compact $\Az$-extension $A'$ with $A \subsetneq A'$ satisfies $\Delta(A') > \Delta(A)$.
The approach closely follows the algorithm for enumerating important separators (see, e.g.,~\cite{our-book}, Chapter 8).

By the maximality of $(\Az,\Bz)$, $\Az$ is the only such extension of excess $0$.
We initiate a queue $Q$ with $Q = \{\Az\}$. Iteratively, until $Q$ is not empty,
we extract an extension $A$ from $Q$, and proceed as follows. For every $v \in N(A)$,
we compute a set $A_v$ such that $E(A_v,V(G) \setminus A_v)$ is a minimum cut between
$A \cup \{v\}$ and $\Bz \cup (\terms \setminus \Az)$, or take $A_v = \bot$
if for such a set $d(A_v)$ would be larger than $d(\Az)+r$. 
Such a set $A_v$ can be computed using $\Oh(k+r)$ rounds of the Ford-Fulkerson algorithm,
and furthermore it allows us to compute $A_v$ being the unique inclusion-wise maximal set with the required properties.

If $A_v \neq \bot$, we insert $A_v$ into the queue $Q$. Otherwise, if $A_v = \bot$ for every $v \in N(A)$, then
we output $A$ as one of the desired sets.
For correctness, observe that every set $A$ in the queue has excess at most $r$, and the described procedure uses
the definition of compactness to check if there exists any other extension of excess at most $r$ being a strict superset of $A$.
For the time bound, observe that whenever a set $A_v$ is inserted into the queue, it holds that $d(A_v) > d(A)$, while $d(\Az) \leq 2k$ (because of the Terminator Reduction).
Hence, $\Oh((2k+r)^r)$ sets are inserted into the queue. Moreover, the computation for a single set $A$ extracted from the queue
takes $\Oh(k^{\Oh(1)} (n+m))$ time.
\end{proof}

We now proceed to the promised description of reductions.
A straightforward corollary of the assumption that $\inst$ is maximal
is the following.

\begin{lemma}\label{lem:ex0}
If $A$ is a terminal-free $\Az$-extension of excess zero or less,
then $A = \Az$.
\end{lemma}

We now study extensions of excess $1$.

\begin{lemma}\label{lem:ex1}
If $A$ is a terminal-free $\Az$-extension of excess $1$, then there
exists a minimum cost integral terminal separation $(\Aopt,\Bopt)$ extending
$(\Az,\Bz)$, such that $(A\setminus \Az)$ is either completely contained in $\Aopt$
or completely contained in $\Bopt$.
\end{lemma}
\begin{proof}
Let $(\Aopt, \Bopt)$ be a minimum cost integral terminal separation extending $(\Az,\Bz)$.

If $(A\setminus \Az)$ is completely contained in $\Bopt$,
then $(\Aopt, \Bopt)$ proves the claim, so let us assume the contrary:
$(A \setminus \Az) \cap \Aopt \neq \emptyset$. Then $A\cap \Aopt \neq \Az$.
We show that $(\Aopt \cup A, \Bopt \setminus A)$ is
a minimum cost integral separation, proving the claim.

Indeed, since $A$ is terminal-free, $(\Aopt \cup A, \Bopt \setminus A)$ is an integral terminal separation. It suffices to show that it is minimum, that is, $d(\Aopt \cup A) \leq d(\Aopt)$.
By submodularity, $d(\Aopt \cup A) + d(\Aopt \cap A) \leq d(\Aopt) + d(A)$.
Since $A \cap \Aopt$ is a terminal-free $\Az$-extension and $A \cap \Aopt\neq \Az$,
by Lemma~\ref{lem:ex0} we have $\Delta(A \cap \Aopt) > 0$,
which means $d(A \cap \Aopt)\geq 1 + d(\Az)$. By assumption $d(A) = 1+d(\Az)$.
Taking this together, $d(\Aopt \cup A) \leq d(\Aopt) + d(A) - d(\Aopt \cap A) \leq d(\Aopt)$, which concludes the proof.
\end{proof}

Lemma~\ref{lem:ex1} proves safeness of the following reduction rule.
\begin{reduction}[Excess-1 Reduction]
If there exists a terminal-free $\Az$-extension of excess $1$ with 
$|A \setminus \Az| > 1$, merge all vertices of $A \setminus \Az$
into a single vertex.
\end{reduction}

The next lemma shows that one can apply the Excess-1 Reduction efficiently.
\begin{lemma}\label{lem:ex1-apply}
Given a maximal instance $\inst$ for which none of the previously defined
reduction rules is applicable, one can in $\Oh(k^{\Oh(1)} (n+m))$ time
find a set $A$
for which the Excess-1 Reduction rule is applicable, or correctly
conclude that no such set exists.
\end{lemma}
\begin{proof}
Let $A$ be a terminal-free $\Az$-extension of excess 1.
If $A \setminus \Az$ is disconnected, then for any connected component $C$ of $A \setminus \Az$
we have that $d(\Az \cup C) + d(A \setminus C) = d(\Az) + d(A)$, hence either $d(\Az \cup C) \leq d(\Az)$ or $d(A \setminus C) \leq d(\Az)$,
contradicting the maximality of $(\Az,\Bz)$. Thus, $A \setminus \Az$ is connected.
If $E(A\setminus \Az, \Az)$ were empty, then $A\setminus \Az$ would be a terminal-free set with $d(A\setminus \Az)=1$, and would hence be deleted by the Pendant Reduction.

Consequently, every terminal-free $\Az$-extension of excess $1$ is compact.
We can enumerate all such inclusion-wise maximal extensions by Lemma~\ref{lem:enum-excess}, and apply the reduction for any such set $A$ with $|A \setminus \Az| > 1$.
\end{proof}

We can henceforth assume that for every terminal-free $\Az$-extension $A$
of excess $1$, the set $A \setminus \Az$ is a singleton.

We now move to an analysis of sets of excess $2$.
\begin{lemma}\label{lem:ex2}
Assume that the Pendant Reduction and Excess-1 Reduction have been exhaustively applied.
If $A$ is a terminal-free $\Az$-extension of excess $2$, then there exists
a partition $A \setminus \Az = D \uplus C_1 \uplus C_2 \uplus \ldots \uplus C_r$ for some $r \geq 0$, such that:
\begin{enumerate}
\item there exists
a minimum cost integral terminal separation $(\Aopt,\Bopt)$ extending 
$(\Az,\Bz)$, such that one of the following holds:
\begin{itemize}
\item $(A \setminus \Az) \cap \Aopt = \emptyset$;
\item $(A \setminus \Az) \cap \Aopt = C_i$ for some $1 \leq i \leq r$; or
\item $A \subseteq \Aopt$.
\end{itemize}
\item for every $1 \leq i \leq r$, 
the sets $C_i$ and $E(C_i,\Az)$ are nonempty, and $\Az \cup C_i$
is a terminal-free $\Az$-extension of excess $1$;
\item if $D \neq \emptyset$, then for every $1 \leq i \leq r$
the set $E(C_i,D)$ is nonempty and $A\setminus\Az$ is connected;
\item if $D = \emptyset$, then $r=2$;
\item for every $1 \leq i < j \leq r$, there are no edges between $C_i$ and $C_j$.
\end{enumerate}
\end{lemma}
\begin{proof}
Let $C'_1,\dots,C'_r$ be all the inclusion-wise maximal subsets of $A$
that are $\Az$-extensions of excess 1.
Let $C_i = C'_i\setminus \Az$ and let $D=A \setminus (\Az \cup C_1 \cup \dots \cup C_r)$.
We show the claim is true for these sets.
Let $1\leq i\neq j \leq r$.

The Excess-1 Reduction allows us to assume that $C_i$ is a singleton and hence $C_i$ is disjoint from $C_j$.
Since $\Delta(C'_i)=\Delta(C'_j)=1$ and $\Delta(C'_i \cup C'_j) \geq 2$ (by maximality of $C'_i$), there are no edges between $C_i$ and $C_j$, proving point 5.

If $E(C_i,\Az)$ were empty, then $d(C_i)=1$ and $C_i$ would be deleted by the Pendant Reduction; this proves point 2.
If $D\neq \emptyset$ but $A\setminus\Az$ was disconnected, then consider a component $C$ of $A\setminus\Az$.
Then $\Delta(A) = \Delta(\Az \cup C) + \Delta(A \setminus C)$,
hence either $\Delta(\Az \cup C) = \Delta(A \setminus C) = 1$, which would contradict that $D\neq \emptyset$,
or one of $\Az \cup C, A\setminus C$ has excess 0, which would contradict Lemma~\ref{lem:ex0}.
Hence $A\setminus \Az$ is connected and as there are no edges between $C_i$ and $C_j$, 
there must be edges between $C_i$ and $D$, proving point 3.
If $D=\emptyset$, then $\Delta(A) = \sum_{i=1}^r \Delta(\Az \cup C_i) = r$. Hence $r=2$, proving point 4.

To prove point 1, consider a minimum cost integral terminal separation $(\Aopt,\Bopt)$.
Since $C_i$ is a singleton, it is either completely contained in $\Aopt$ or disjoint from it.
If $\Aopt \cap (A\setminus \Az)$ is empty or equal to one of $C_i$, the claim follows.
Otherwise, $\Aopt \cap (A\setminus \Az)$ contains a vertex of $D$ or two of the $C_i$ sets;
by their maximality, the excess of $\Aopt \cap A$ is then at least 2, so $d(\Aopt \cap A) \geq d(A)$.
By submodularity, $d(\Aopt \cup A) + d(\Aopt \cap A) \leq d(\Aopt) + d(A)$
and thus $d(\Aopt \cup A) \leq d(\Aopt)$.
Therefore, since $A$ is terminal-free, $(\Aopt \cup A, \Bopt \setminus A)$ is an integral terminal separation, concluding the proof.
\end{proof}

Lemma~\ref{lem:ex2} ensures safeness of the following reduction rule.
\begin{reduction}[Excess-2 Reduction]
If there exists a terminal-free $\Az$-extension $A$ of excess $2$ 
such that in the partition $D \uplus C_1 \uplus \ldots \uplus C_r$ defined by
Lemma~\ref{lem:ex2}, $|D|>1$,
then merge $D$ into a single vertex.
\end{reduction}

We are left with an efficient implementation of this rule.
\begin{lemma}\label{lem:ex2-apply}
Given a maximal instance $\inst$ for which none of the previously defined
reduction rules is applicable, one can in $\Oh(k^{\Oh(1)} (n+m))$ time
find a set $A$
for which the Excess-2 Reduction is applicable and compute the decomposition
of $A \setminus \Az$ of Lemma~\ref{lem:ex2}, or correctly
conclude that no such set $A$ exists.
\end{lemma}
\begin{proof}
Let $A$ be a terminal-free $\Az$-extension of excess $2$, and let $D, C_1,C_2,\ldots,C_r$ be the sets promised by Lemma~\ref{lem:ex2} and let $|D|>1$.
The inapplicability of the Excess-1 Reduction ensures that every set $C_i$ is a singleton, $C_i = \{c_i\}$.

Let us first deal with the corner case in which $r=0$ and $E(D,\Az) = \emptyset$.
Then, since $A$ is of excess $2$,
we have $d(D) = 2$. However, as $D$ does not contain any terminal,
the Pendant Reduction is applicable to it.

In the remaining cases, Lemma~\ref{lem:ex2} guarantees that $A$ is compact.
We enumerate all inclusion-wise maximal compact excess-$2$ extensions using Lemma~\ref{lem:enum-excess}. 
For every output extension $A$, we first identify the set $C \subseteq A \setminus \Az$ of all vertices $v$ such that $\Az \cup \{v\}$ is of excess
one. By Lemma~\ref{lem:ex2}, we have $D = A \setminus (\Az \cup C)$. If $|D| > 1$, then we can apply the reduction.

To complete the proof, note that if the Excess-2 Reduction is applicable to some compact $\Az$-extension $A$, then it is also applicable
to any compact $\Az$-extension $A'$ of excess $2$ being a superset of $A$: the corresponding set $D$ for $A$ is a subset of the corresponding set $D'$ for $A'$.
\end{proof}

The set $D$ of Lemma~\ref{lem:ex2} is often a very convenient branching pivot: putting it into $\Az$ makes the boundary of $\Az$ extend by two, while putting it into $\Bz$ triggers a number of Boundary Reductions.
In the next few lemmata we summarize the properties of an excess-2 set after reductions, and outcomes on branching on the set $D$.

We start from a slightly more useful presentation of the properties promised by Lemma~\ref{lem:ex2}.

\begin{figure}[H]
	\centering
	\begin{tikzpicture}[scale=1.1]

\Azero;
\node at (0,-0.2) {$\Az$};
\node[A,label=60:$\mathbf{d}$] (d) at (0,2) {};
\node[A,label=$\mathbf{c_1}$] (c1) at (-1.7,1) {};
\node[A,label=$\mathbf{c_2}$] (c2) at (-0.3,1) {};
\node[A,label=$\mathbf{c_3}$] (c3) at (1.5,1) {};
\draw (d) to[bend left=10] (c1);
\draw (d) to[bend right=10] (c1);
\node at (-1.1,1.65) {$p_1$};

\draw (d) to[bend left] (c2);
\draw (d) to (c3);

\draw (d) to (-0.7,2.2);
\draw (d) to (-0.7,2.3);
\draw (d) to (-0.6,2.4);

\draw (c1) to (-1.5,-0.1)  (c1) to (-1.66,-0.15)  (c1) to (-1.84,-0.1);
\draw (c1) to (-2.6,1.15) (c1) to (-2.6,1.3);
\node at (-2.6,0) {$p_1+x_1$};
\node at (-2.9,1.4) {$x_1+1$};
\draw (c2) to (-0.4,-0.1);
\draw (c2) to[bend right=15] (1.2,2.1); 
\draw (c3) to (1.4,-0.1);
\draw (c3) to (1.6,-0.1);
\draw (c3) to (2.2,1.15);
\draw (c3) to (2.2,1.3);
\draw[non] (c1) to[bend right=20] (c2) (c2) to[bend right=20] (c3) (c1) to[bend right] (c3);

\begin{scope}[shift={(6,0)}]

\Azero;
\node at (0,-0.2) {$\Az$};
\node[A,label=$\mathbf{c_1}$] (c1) at (-1.1,1) {};
\node[A,label=$\mathbf{c_2}$] (c2) at (1.1,1) {};
\draw (c1) to (-1,-0.1);
\draw (c1) to (-1.2,-0.1);
\node at (-1.45,0.15) {$x_1$};

\draw (c1) to (-1.9,1.15);
\draw (c1) to (-1.9,1.3);
\draw (c1) to (-1.9,1.45);
\node at (-1.9,1.6) {$x_1+1$};
\draw (c2) to (1.2,-0.1);
\draw (c2) to (1.05,-0.15);
\draw (c2) to (0.9,-0.1);

\draw (c2) to (1.9,1.0);
\draw (c2) to (1.9,1.1);
\draw (c2) to (1.9,1.2);
\draw (c2) to (1.9,1.3);
\draw[non] (c1) to[bend right=10] (c2);

\end{scope}

\end{tikzpicture}
\figspace
	\caption{Examples of sets of excess 2 after reductions (dotted lines are non-edges).}
\label{fig:ex2-red}
\end{figure}
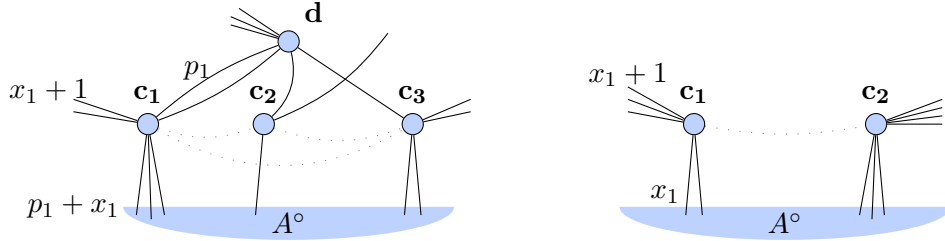

\begin{lemma}\label{lem:ex2-red}
Assume that no reduction is applicable, and let $A$ be a terminal-free $\Az$-extension of excess $2$.
Then one can in $\Oh(k^{\Oh(1)} m)$ time
compute a decomposition $A \setminus \Az = \{d,c_1,c_2,\ldots,c_r\}$ for some $r \geq 0$ or $A \setminus \Az = \{c_1,c_2\}$
with the following properties:
\begin{enumerate}
\item if the vertex $d$ exists, then $A$ is compact and for every $1 \leq i \leq r$, there are $p_i$ edges $dc_i$ for some $p_i \geq 1$; we put $p_1 = p_2 = 0$ if the vertex
$d$ does not exists;
\item for every $1 \leq i \leq r$, the set $\Az \cup \{c_i\}$ is an $\Az$-extension of excess $1$, the vertex $c_i$ has $x_i+1 \geq 1$ edges towards
$V(G) \setminus (A \cup \Bz)$ and $p_i+x_i \geq 1$ edges towards $\Az$, for some $x_i \geq 0$;
\item the vertices $c_i$ are pairwise nonadjacent;
\item the set $\Az \cup \{d\}$ is an $\Az$-extension of excess larger than $1$.
\end{enumerate}
\end{lemma}
\begin{proof}
Most of the enumerated properties are just repetitions of the points of Lemma~\ref{lem:ex2},
after each set of the partition has been identified into a single vertex. Recall that noncompact $\Az$-extensions of excess $2$
are completely reduced by the Pendant Reduction.

For the count on the number of edges incident to a vertex $c_i$,
define $p_i$ as claimed and $x_i := |E(c_i,V(G) \setminus A)|-1$; clearly $x_i \geq -1$.
Since $\Az \cup \{c_i\}$ is of excess $1$, and no two vertices $c_i$ are adjacent, 
we have $|E(c_i,\Az)| = p_i+x_i$.
Furthermore, note that no edge may connect $c_i$ and $\Bz$, as it would trigger a Boundary Reduction.
It remains to refute the case $x_i = -1$, i.e., $E(c_i,V(G) \setminus A) = \emptyset$. 
In this case $p_i + x_i = |E(c_i,\Az)| \geq 0$ implies $p_i \geq 1$, so the vertex $d$ exists.
However, the Majority Neighbour Reduction then applies to $c_i$ and $d$, a contradiction.

If $\Az \cup \{d\}$ is an $\Az$-extension of excess at most $1$,
then $r \geq 1$ as $A$ has excess $2$, but then an edge count shows that $\Az \cup \{d,c_1\}$ would be an $\Az$-extension of nonpositive excess, a contradiction
to the maximality of $\Az$.

Finally, the decomposition of $A\setminus \Az$ can be identified by inspecting the edges incident to every vertex $v \in A \setminus \Az$
to check whether $\Az \cup \{v\}$ is of excess $1$ or larger.
\end{proof}

We now investigate what happens in a branch when we put the vertex $d$ onto the $A$-side.
\begin{lemma}\label{lem:ex2-nested}
Assume that no reduction is applicable, and let $A,A'$ be two terminal-free $\Az$-extensions of excess $2$ with $A \subsetneq A'$.
Then $A' \setminus \Az$ decomposes as $\{d,c_1,c_2,\ldots,c_r\}$ for some $r \geq 2$, and $A \setminus \Az$
consists of two vertices $c_i$ of this decomposition.
\end{lemma}
\begin{proof}
If $A' \setminus \Az = \{c_1,c_2\}$, then there is no choice for the set $A$, as $\Az \cup \{c_i\}$ is of excess $1$ for $i=1,2$.
Hence, $A' \setminus \Az = \{d,c_1,c_2,\ldots,c_r\}$ 
for some $r \geq 1$; note that $|A' \setminus \Az| \geq 2$ as
$\Az \subsetneq A \subsetneq A'$.
A direct edge count using Lemma~\ref{lem:ex2-red} shows that for every $C \subseteq \{c_1,c_2,\ldots,c_r\}$
we have $\Delta(\Az \cup C) = |C|$ and $\Delta(\Az \cup C \cup \{d\}) \geq 2 + (r-|C|)$.
Hence, the only option to get excess $2$ is to have $A = \Az \cup C$ for some $|C|=2$.
\end{proof}

\begin{lemma}\label{lem:ex2-Aside}
Assume that no reduction is applicable, and let $A$ be a terminal-free $\Az$-extension of excess $2$
with $A \setminus \Az = \{d,c_1,c_2,\ldots,c_r\}$ for some $r \geq 0$.
If we furthermore consider a branch $(A_1,B_1)$ such that $d \in A_1$,
but $A_1 \setminus \Az$ does not contain any terminal, then
\begin{enumerate}
\item if $B_1$ contains at least one vertex $c_i$, then there does not exist any minimum cost integral terminal separation $(A^\ast,B^\ast)$
extending $(\Az,\Bz)$ that also extends $(A_1,B_1)$;
\item $d(A_1) \geq d(\Az) + 2$;
\item if $d(A_1) = d(\Az) + 2$, then $A_1 = A$.
\end{enumerate}
\end{lemma}
\begin{proof}
Define $A' := A_1 \cup A$ and $B' := B_1 \setminus A$; note that  $A' \setminus \Az$ is terminal-free and $(A',B')$ is a terminal separation as well.

Observe that if $(A_1,B_1)$ is a terminal separation extending $(\Az,\Bz)$ with $d \in A_1$ but $c_i \notin A_1$ for some $1 \leq i \leq r$,
then a direct edge count from Lemma~\ref{lem:ex2-red} shows that 
$d(A_1 \cup \{c_i\}) < d(A_1)$, $d(B_1 \setminus \{c_i\}) \leq d(B_1)$, hence $c(A_1 \cup \{c_i\}, B_1 \setminus \{c_i\}) < c(A_1,B_1)$.
This proves the first point, and shows that $d(A') \leq d(A_1)$, $d(B') \leq d(B_1)$, thus $c(A',B') \leq c(A_1,B_1)$, and the equality holds only if $(A',B') = (A_1,B_1)$.

Since $A \subseteq A'$, the Excess-1 Reduction is inapplicable, and $\Delta(A) = 2$, we have $\Delta(A') \geq 2$. Consequently, $d(A_1) \geq d(A') \geq d(\Az)+2$,
and $d(A_1) = d(\Az)+2$ only if $d(A_1) = d(A') = d(\Az)+2$.
As discussed in the previous paragraph, this can only happen if $A' = A_1$ and $\Delta(A') = 2$.
By Lemma~\ref{lem:ex2-nested}, this implies $A' = A$, finishing the proof of the lemma.
\end{proof}

In the last lemma we study what happens in a branch when we put the vertex $d$ onto the $B$-side.

\begin{figure}[H]
	\centering
	\begin{tikzpicture}

\Azero;
\node at (1,-0.2) {$\Az$};
\node[A,label=$\mathbf{d}$] (d) at (0,1) {};
\draw (d) to (0,-0.1);
\draw (d) to (-1.1,1);
\draw (d) to (-1.1,1.55);
\draw (d) to (-0.8,1.9);

\begin{scope}[shift={(6,0)}]

\Azero;
\node at (1,-0.2) {$\Az$};
\node[A,label=60:$\mathbf{d}$] (d) at (0,1.5) {};
\node[A,label=0:$\mathbf{c_1}$] (c1) at (0,0.8) {};
\draw (d) to (c1);
\draw (d) to (-0.3,2.2);
\draw (d) to (-0.7,2);
\draw (c1) to (-0.1,-0.1);
\draw (c1) to (0.1,-0.15);
\draw (c1) to (0.3,-0.1);
\node at (-0.25,0.15) {$x$};

\draw (c1) to (-1,0.7);
\draw (c1) to (-1,0.9);
\draw (c1) to (-1,1.1);
\node at (-1.2,0.9) {$x$};
\end{scope}

\end{tikzpicture}
\figspace
	\caption{The two cases when putting $d$ on the unnatural side triggers only one Boundary Reduction.}
\label{fig:ex2-Bside}
\end{figure}
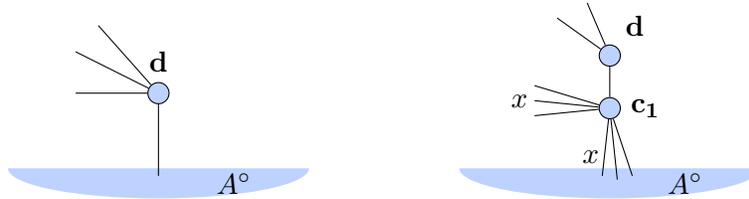

\begin{lemma}\label{lem:ex2-Bside}
Assume that no reduction is applicable, and let $A$ be a terminal-free $\Az$-extension of excess $2$
with $A \setminus \Az = \{d,c_1,c_2,\ldots,c_r\}$ for some $r \geq 0$.
Furthermore, if we consider a branch $(A_1,B_1)$ such that $d \in B_1$,
then at least one Boundary Reduction is immediately triggered.
If only one is triggered, then one of the following holds:

\begin{enumerate}
\item $r = 0$, $A \setminus \Az = \{d\}$, and the vertex $d$ is of degree four, with one incident edge having second endpoint in $\Az$ and the remaining
three edges having second endpoint in $V(G) \setminus (A \cup \Bz)$; or
\item $r=1$, $A \setminus \Az = \{d,c_1\}$, the vertex $d$ is of degree three, with one incident edge being $c_1d$ and the remaining
two edges having second endpoint in $V(G) \setminus A$, and the vertex $c_1$ is of degree $2x+1$ for some $x \geq 1$, with
one incident edge being $c_1d$, $x$ incident edges having second endpoint in $\Az$, and $x$ incident edges having second endpoint in $V(G) \setminus (A \cup \Bz)$.
\end{enumerate}
\end{lemma}
\begin{proof}
In the branch $(A_1,B_1)$,
a Boundary Reduction is immediately triggered for every edge in $E(d,\Az)$, and every vertex $c_i$ triggers $\min(p_i,p_i+x_i) = p_i$ Boundary Reductions.
Note that $r\geq 1$ or $E(D,\Az)\neq\emptyset$, as $A$ is compact by Lemma~\ref{lem:ex2-red}.
Hence at least one reduction is triggered.
If only one reduction is triggered, then $|E(d,\Az)| + \sum_{i=1}^r p_i = 1$. In particular $r$ is either $0$ or $1$. 

If $r = 0$, then $|E(d,\Az)| = 1$ and the assumption that $A$ is of excess $2$ implies that
$|E(d, V(G) \setminus \Az)| = 3$. No edge incident to $d$ may have a second endpoint in $\Bz$, as it would trigger
the Boundary Reduction together with the edge in $E(d,\Az)$.
Thus the first case of the claim holds.

If $r=1$, then $|E(d,\Az)| = 0$ and $p_1=1$. Since $c_1$ has $p_1+x_1$ edges to $\Az$ and $x_1+1$ edges to $V(G)\setminus A$, the assumption that $A$ is of excess $2$ implies that $d$ has exactly two edges to $V(G)\setminus A$.
No edge incident to $c_1$ can have the second endpoint in $\Bz$,
as otherwise it would trigger the Boundary Reduction with any edge
in $E(c_1,\Az)$.
Thus the second case of the claim holds.
\end{proof}

\section{The detailed cases of the branching algorithm}\label{sec:miesko}
In this section we assume we have a maximal instance $\inst = (G,\terms,(\Az,\Bz),k)$ for which none of the previously defined reduction rules is applicable.
Our goal is to find a {\em{branching step}} that fulfils a good vector,
or a set of vertices to merge (a {\em{reduction step}}). Recall that when we consider a branching into terminal separations $(A_1,B_1)$ and $(A_2,B_2)$ that extend $(\Az,\Bz)$, then $t_i,\nu_i,k_i$ for $i=1,2$ measure respectively the number of terminals resolved in branch $i$, two times the growth of the cost of the separation in branch $i$ (i.e., $2(c(A_i,B_i)-c(\Az,\Bz))$), and the decrease in the budget $k$ after applying all the reduction rules when recursing into branch $i$.

Assume that we have identified a branching step into separations
$(A_1,B_1)$ and $(A_2,B_2)$ that both extend, but are different than
$(\Az,\Bz)$. Then, from the maximality of $(\Az,\Bz)$ we infer than
$\nu_1,\nu_2 \geq 1$.
Since $[1,1,0;2,1,0]$ is a good vector, any branching step
in which in both cases we resolve or reduce at least one terminal pair,
while in at least one case we resolve or reduce at least \emph{two} terminal
pairs, is fine for our purposes. 

\subsection{Basic branching and reductions}

Let $\terms' \subseteq \terms$ be the set of unresolved terminal pairs (not in $\Az\cup \Bz$). For every terminal pair $\{s,t\}\in \terms'$, we apply the algorithm of Theorem~\ref{thm:ptime} twice: once for terminal separation $(\Az\cup \{s\},\Bz\cup \{t\})$, and the second time for terminal separation $(\Az\cup \{t\},\Bz\cup \{s\})$. In this manner we obtain two maximal terminal separations $(A_s,B_t)$ and $(A_t,B_s)$ that extend $(\Az\cup \{s\},\Bz\cup \{t\})$ and $(\Az\cup \{t\},\Bz\cup \{s\})$ respectively. Of course, the number of unresolved pairs decreases by at least one in both $(A_s,B_t)$ and $(A_t,B_s)$, due to resolving $\{s,t\}$. If the number of unresolved pairs either in $(A_s,B_t)$ or in $(A_t,B_s)$ decreases by more than one, then, as we argued, performing a branching step $(A_1,B_1)=(A_s,B_t)$ and $(A_2,B_2)=(A_t,B_s)$ leads to the branching vector $[1,1,0;2,1,0]$ or a better one, which is good. We can test in $\Oh(k^{\Oh(1)} m)$ time whether this holds for any pair $\{s,t\}\in \terms'$, and if so then we pursue the branching step. 

\begin{branching}\label{br:twopairs}
If in either $(A_s,B_t)$ or in $(A_t,B_s)$, more than one terminal pair gets resolved, then perform branching into $(A_1,B_1)=(A_s,B_t)$ and $(A_2,B_2)=(A_t,B_s)$.
\end{branching}

Hence, if this branching step cannot be performed, then we assume the following:

\begin{assumption}\label{ass:one-pair}
For every pair $\{s,t\}\in \terms'$ in both $(A_s,B_t)$ and $(A_t,B_s)$ only the pair $\{s,t\}$ gets resolved.
\end{assumption}

We now proceed with some structural observations about the instance at hand.

\begin{lemma}\label{lem:AsConn}
$G[A_s\setminus \Az]$, $G[A_t\setminus \Az]$, $G[B_s\setminus \Bz]$, $G[B_t\setminus \Bz]$ are connected.
\end{lemma}
\begin{proof}
We prove the statement for $G[A_s\setminus \Az]$, since the other statements are symmetric. Suppose $G[A_s\setminus \Az]$ is disconnected, and let $C$ be any of its connected component that does not contain $s$. Then $C$ is terminal-free, so by the maximality of $(\Az,\Bz)$ we infer that $d(C\cup \Az)>d(\Az)$. But then $d(A_s\setminus C)<d(A_s)$, which contradicts the optimality of $(A_s,B_s)$.
\end{proof}

\begin{lemma}\label{lem:lb}
Let $\{s,t\}\in \terms'$, and let $(A_s,B_t)$ and $(A_t,B_s)$ be any optimum-cost terminal separations extending $(\Az\cup\{s\},\Bz\cup \{t\})$ and $(\Az\cup\{t\},\Bz\cup \{s\})$, respectively. Suppose that $(A_s,B_t)$ and $(A_t,B_s)$ do not resolve any terminal pair apart from $\{s,t\}$. Then for any set $A$ with $\Az\cup\{s\}\subseteq A\subseteq V(G)\setminus \Bz$ that has only $s$ among the terminals of $\terms'$, it holds that $\Delta(A)\geq \Delta(A_s)$. Symmetrically, for any set $B$ with $\Bz\cup\{s\}\subseteq B\subseteq V(G)\setminus \Az$ that has only $s$ among the terminals of $\terms'$, it holds that $\Delta(B)\geq \Delta(B_s)$.
\end{lemma}
\begin{proof}
We prove only the first claim for the second one is symmetric. Let $A$ be such a set, and for the sake of contradiction suppose $\Delta(A)<\Delta(A_s)$. Then $d(A)+d(B_t)<2c(A_s,B_t)$. However, from posimodularity of cuts it follows that either $d(B_t\setminus A)+d(A)\leq d(B_t)+d(A)$ or $d(B_t)+d(A\setminus B_t)\leq d(B_t)+d(A)$. Both $(A,B_t\setminus A)$ and $(A\setminus B_t,B_t)$ are terminal separations that extend $(\Az\cup\{s\},\Bz\cup \{t\})$, and one of them has strictly smaller cost than $(A_s,B_t)$. This is a contradiction with the optimality of $(A_s,B_t)$.
\end{proof}

\subsubsection{Pushing $A_s$ and $B_s$}

The problem that we will soon face is that separations $(A_s,B_t)$ and $(A_t,B_s)$ are not uniquely defined. For instance, there can be some set of vertices $Z\subseteq A_s\setminus \Az$ that could be moved from $A_s$ to $B_t$ without changing the cost of the separation. We now make an adjustment of these separations so that we can assume that $A_s$, resp. $B_s$, is maximal. For this, we need the following technical results.

\begin{lemma}\label{lem:patch}
Suppose that $(A_s,B_t)$ and $(A_s',B_t')$ are maximal terminal separations of minimum cost among separations that extend $(\Az\cup \{s\},\Bz\cup \{t\})$. Suppose further that they do not resolve any other terminal pair from $\terms'$. Then
\begin{itemize}
\item[(a)] $d(A_s)=d(A_s')$ and $d(B_t)=d(B_t')$;
\item[(b)] $(A_s\cap A_s',B_t\cup B_t')$ and $(A_s\cup A_s',B_t\cap B_t')$ are also terminal separations of minimum cost among separations that extend $(\Az\cup \{s\},\Bz\cup \{t\})$;
\item[(c)] $A_s\cup B_t=A_s'\cup B_t'$.
\end{itemize}
\end{lemma}
\begin{proof}
\noindent (a) Let $C=c(A_s,B_t)=c(A_s',B_t')$ be the minimum cost of a terminal separation extending $(\Az\cup \{s\},\Bz\cup \{t\})$. Suppose w.l.o.g. that $d(A_s)<d(A_s')$, then we have that $d(B_t)>d(B_t')$. By posimodularity, we have that
\begin{equation}\label{eq1}
d(A_s\setminus B_t')+d(B_t'\setminus A_s)\leq d(A_s)+d(B_t')<2C.
\end{equation}
Observe that $(A_s\setminus B_t',B_t')$ is a terminal separation that extends $(\Az\cup \{s\},\Bz\cup \{t\})$, and hence 
\begin{equation}\label{eq2}
d(A_s\setminus B_t')+d(B_t')=2c(A_s\setminus B_t',B_t')\geq 2C.
\end{equation}
Symmetrically, by considering terminal separation $(A_s,B_t'\setminus A_s)$ we obtain that
\begin{equation}\label{eq3}
d(A_s)+d(B_t'\setminus A_s)=2c(A_s,B_t'\setminus A_s)\geq 2C.
\end{equation}
Thus, from~\eqref{eq1},~\eqref{eq2}, and~\eqref{eq3} we obtain that
$$4C\leq d(A_s)+d(B_t')+d(A_s\setminus B_t')+d(B_t'\setminus A_s)<4C,$$
which is a contradiction.

\smallskip

\noindent (b) Observe that $d(A_s\cap A_s')\geq d(A_s)$, because otherwise $A_s$ could have been replaced with $A_s\cap A_s'$ in separation $(A_s,B_t)$. By submodularity of cuts we have that $d(A_s\cap A_s')+d(A_s\cup A_s')\leq d(A_s)+d(A_s')$, and hence $d(A_s\cup A_s')\leq d(A_s')=d(A_s)$. By posimodularity, we have that
\begin{equation}\label{eq21}
d((A_s\cup A_s') \setminus B_t)+d(B_t\setminus (A_s\cup A_s'))\leq d(A_s\cup A_s')+d(B_t)\leq d(A_s)+d(B_t)=2C
\end{equation}
On the other hand, for terminal separation $((A_s\cup A_s') \setminus B_t,B_t)$ we have that
\begin{equation}\label{eq22}
d((A_s\cup A_s') \setminus B_t)+d(B_t) = 2c((A_s\cup A_s') \setminus B_t,B_t) \geq 2C,
\end{equation}
and for terminal separation $(A_s\cup A_s',B_t\setminus (A_s\cup A_s'))$ we have that
\begin{equation}\label{eq23}
d(A_s\cup A_s')+d(B_t\setminus (A_s\cup A_s')) = 2c((A_s\cup A_s',B_t\setminus (A_s\cup A_s')) \geq 2C.
\end{equation}
Thus, from~\eqref{eq21},~\eqref{eq22}, and~\eqref{eq23}
$$4C\geq d((A_s\cup A_s') \setminus B_t)+d(B_t)+d(A_s\cup A_s')+d(B_t\setminus (A_s\cup A_s'))\geq 4C,$$
which means that all the inequalities above are in fact equalities. In particular:
\begin{itemize}
\item $d(A_s\cap A_s')=d(A_s)=d(A_s\cup A_s')$, and 
\item $c((A_s\cup A_s') \setminus B_t,B_t)=C$.
\end{itemize} 
Symmetric arguments can be used to show that:
\begin{itemize}
\item $d(B_t\cap B_t') = d(B_t)=d(B_t\cup B_t')$, 
\item $c((A_s\cup A_s') \setminus B_t',B_t')=C$, 
\item $c(A_s,(B_t\cup B_t')\setminus A_s)=C$, and
\item $c(A_s',(B_t\cup B_t')\setminus A_s')=C$.
\end{itemize}
Therefore, both $(A_s\cap A_s',B_t\cup B_t')$ and $(A_s\cup A_s',B_t\cap B_t')$ have cost $C$.

\smallskip

\noindent (c) For the sake of contradiction, assume that $A_s\cup B_t\neq A_s'\cup B_t'$. Suppose first that there is an element $u\in A_s$ such that $u\notin A_s'\cup B_t'$. In the proof of (b) we have showed that $c((A_s\cup A_s') \setminus B_t',B_t')=C$. Note that $((A_s\cup A_s') \setminus B_t',B_t')$ is a terminal separation that extends $(A_s',B_t')$, and moreover its left side is has at least one additional element $u$. Since its cost is the same as the cost of $(A_s',B_t')$, we obtain a contradiction with the maximality of $(A_s',B_t')$.
\end{proof}

\begin{lemma}\label{lem:pushing}
Let $\mathcal{F}$ be the family of all maximal terminal separations $(A_s,B_t)$ of minimum cost among separations that extend $(\Az\cup \{s\},\Bz\cup \{t\})$. Suppose that all separations from $\mathcal{F}$ resolve only the pair $\{s,t\}$ among the pairs from $\terms'$. Then there exists a unique maximal terminal separation $(A_s^{\max},B_t^{\min})$ such that $A_s^{\max}\supseteq A_s$ and $B_t^{\min}\subseteq B_t$ for each $(A_s,B_t)\in \mathcal{F}$. Moreover, if $A$ is such that $\Az\cup \{s\}\subseteq A$, $A\cap \Bz=\emptyset$, $A \cap \bigcup \terms' \subseteq \{s\}$, but $A\setminus A_s^{\max}\neq \emptyset$, then $d(A)>d(A_s^{\max})$.
\end{lemma}
\begin{proof}
We set 
$$(A_s^{\max},B_t^{\min})=\left(\bigcup_{(A_s,B_t)\in \mathcal{F}} A_s,\bigcap_{(A_s,B_t)\in \mathcal{F}} B_t\right).$$
From Lemma~\ref{lem:patch} it follows that $(A_s^{\max},B_t^{\min})\in \mathcal{F}$. 

\newcommand{\oA}{\overline{A}}
\newcommand{\oB}{\overline{B}}

We are left with proving the last statement. Take any such $A$, and suppose for the sake of contradiction that $d(A)\leq d(A_s^{\max})$. Let $\oA=A_s^{\max}\cup A$ and $\oB=B_t^{\min}\setminus A$. Observe that $(\oA,\oB)$ is a terminal separation that extends $(\Az\cup \{s\},\Bz\cup \{t\})$. Since $\oA$ has at least one more element than $A_s^{\max}$, from the properties of $(A_s^{\max},B_t^{\min})$ we infer that $c(\oA,\oB)>C$, where $C$ is the cost of every separation from $\mathcal{F}$. Observe that $d(A_s^{\max}\cap A)\geq d(A_s^{\max})$, because otherwise we would substitute $A_s^{\max}$ with $A_s^{\max}\cap A$ in separation $(A_s^{\max},B_t^{\min})$ and obtain a separation of smaller cost that extends $(\Az\cup \{s\},\Bz\cup \{t\})$. Hence, from the submodularity of cuts we infer that $d(\oA)\leq d(A)$, so in particular $d(\oA)\leq d(A_s^{\max})$.

Now, by posimodularity we obtain that
$$d(\oA\setminus B_t^{\min})+d(B_t^{\min}\setminus \oA)\leq d(\oA)+d(B_t^{\min})\leq d(A_s^{\max})+d(B_t^{\min}).$$
On the other hand, observe that $d(\oA\setminus B_t^{\min})\geq d(A_s^{\max})$, because otherwise we could substitute $A_s^{\max}$ with $\oA\setminus B_t^{\min}$ in the terminal separation $(A_s^{\max},B_t^{\min})$ and obtain a terminal separation that extends $(\Az\cup \{s\},\Bz\cup \{t\})$ and has strictly smaller cost. Thus we infer that $d(B_t^{\min}\setminus \oA)\leq d(B_t^{\min})$. As $B_t^{\min}\setminus \oA=B_t^{\min}\setminus A=\oB$, we conclude that $d(\oA)\leq d(A_s^{\max})$, $d(\oB)\leq d(B_t^{\min})$, and hence $c(\oA,\oB)\leq C$. This is a contradiction.
\end{proof}

We modify now separation $(A_s,B_t)$ as follows. For every terminal pair $\{s',t'\}\in \terms'$ that is different from $\{s,t\}$, we verify using Theorem~\ref{thm:ptime} whether $(A_s,B_t)$ can be chosen so that it has a minimum possible cost among the separations that extend $(\Az\cup \{s\},\Bz\cup \{t\})$, but it also resolves $\{s',t'\}$. If this is possible, then we pursue Branching Step~\ref{br:twopairs} with appropriate $(A_s,B_t)$. Otherwise, every minimum-cost separation extending $(\Az\cup \{s\},\Bz\cup \{t\})$ resolves only $\{s,t\}$, and the assumptions of Lemma~\ref{lem:pushing} are satisfied.
Let $(A_s^{\max},B_t^{\min})$ be the terminal extension whose existence is asserted by Lemma~\ref{lem:pushing}.
Observe that we can construct $(A_s^{\max},B_t^{\min})$ in time $\Oh(k^{\Oh(1)} m)$:
we start with any $(A_s,B_t)$ given by Theorem~\ref{thm:ptime},
and observe that Lemma~\ref{lem:pushing} implies that $A_s^{\max}$ is the unique inclusion-wise maximal set
containing $A_s$ such that $E(A_s^{\max},V(G) \setminus A_s^{\max})$ is a minimum cut between $A_s$ and $\Bz \cup (\terms \setminus A_s)$;
such a set can be computed using $\Oh(k)$ rounds of the Ford-Fulkerson algorithm.
  
Hence, we proceed further with the assumption that we have chosen $(A_s,B_t)$ to be $(A_s^{\max},B_t^{\min})$. We do symmetrically in the second branch, assuming that $(A_t,B_s)$ is chosen to be $(A_t^{\min},B_s^{\max})$, that is, the extension of $B$ that contains terminal $s$ is chosen to be maximum possible. Hence, by Lemma~\ref{lem:pushing}, we can from now on use the following assumption.

\begin{assumption}\label{ass:pushing}
For any set $A$ with $\Az\subseteq A\subseteq V(G)\setminus \Bz$ that contains only $s$ from the terminals of $\terms'$ and has at least one vertex outside $A_s$, it holds that $\Delta(A)>\Delta(A_s)$. Symmetrically, for any set $B$ with $\Bz\subseteq B\subseteq V(G)\setminus \Az$ that contains only $s$ from the terminals of $\terms'$ and has at least one vertex outside $B_s$, it holds that $\Delta(B)>\Delta(B_s)$. 
\end{assumption}

\subsubsection{Analyzing $A_s\cap B_s$, $A_s\setminus B_s$, and $B_s\setminus A_s$}

Suppose now that for some pair $\{s,t\}\in \terms'$, we have that $|(A_s\cap B_s)\setminus \{s\}|\geq 2$. Then, by Assumption~\ref{ass:one-pair} $Z=(A_s\cap B_s)\setminus \{s\}$ is a terminal-free set. Since pair $\{s,t\}$ has to be resolved one way or the other, then by persistence (Theorem~\ref{thm:persistence}) we infer that there is some minimum integral terminal separation $(\Aopt,\Bopt)$ such that $Z\subseteq \Aopt$ or $Z\subseteq \Bopt$. Therefore, it is a safe reduction to merge $Z$ into a single vertex.

\begin{reductionstep}
For every $\{s,t\}\in \terms'$, compute $Z_s=(A_s\cap B_s)\setminus \{s\}$ and $Z_t=(A_t\cap B_t)\setminus \{t\}$. Provided $Z_s$ ($Z_t$) contains more than one vertex, merge it.
\end{reductionstep}

We apply this reduction to all terminal pairs from $\terms'$, which takes time $\Oh(k^{\Oh(1)} m)$. Hence, using Lemma~\ref{lem:AsConn} from now on we can assume the following:

\begin{assumption}\label{ass:small-intersections}
For every pair $\{s,t\}\in \terms'$, either $A_s\cap B_s=\{s\}$ or $A_s\cap B_s=\{s,s'\}$, where $s'$ is the only neighbor of $s$. Moreover, either $A_t\cap B_t=\{t\}$ or $A_t\cap B_t=\{t,t'\}$, where $t'$ is the only neighbor of $t$.
\end{assumption}

As every terminal has degree one, for a pair $\{s,t\}\in \terms'$ we have that $d(A_s)\leq d(\Az\cup \{s\})\leq d(\Az)+1$, since otherwise replacing $A_s$ with $\Az\cup \{s\}$ would decrease the cost of $(A_s,B_s)$. On the other hand, we have that $d(A_s)\geq d(\Az)$, since otherwise $(A_s,\Bz\cup \{t\})$ would be a terminal separation extending $(\Az,\Bz)$ of not larger cost, which would contradict the maximality of $(\Az,\Bz)$. Then, we have three possible cases for $(\Delta(A_s),\Delta(B_s))$: $(0,0)$, $(1,0)$ and $(1,1)$; the omitted case $(0,1)$ is symmetric to $(1,0)$. The algorithm behaves differently in each of these cases. Before we proceed to the description of handling each case separately, we prove some useful observations first.

\newcommand{\Atr}{\tilde{A}}
\newcommand{\Btr}{\tilde{B}}

Let us now fix one pair $\{s,t\}$, and let $\Atr=A_s\setminus B_s$ and $\Btr = B_s\setminus A_s$. Observe that since branching on $\{s,t\}$ did not resolve any additional terminal pair, then both $\Atr \setminus \Az$ and $\Btr \setminus \Bz$ are terminal-free. Hence, by the maximality of $(\Az,\Bz)$ we have that
\begin{equation}\label{eq:truncs}
d(\Atr)\geq d(\Az) \qquad \text{and}\qquad d(\Btr)\geq d(\Bz),
\end{equation}
and the equality holds if and only if $\Atr=\Az$ or $\Btr=\Bz$, respectively. Let $R=V(G)\setminus (A_s\cup B_s)$.

\begin{lemma}\label{lem:posi}
One of the following two cases holds:
\begin{itemize}
\item $|E(A_s\cap B_s,R)|=1$, $\Atr=\Az$, $\Btr=\Bz$, and $(\Delta(A_s),\Delta(B_s))=(1,1)$; or
\item $|E(A_s\cap B_s,R)|=0$, and $2\geq \Delta(A_s)+\Delta(B_s)=\Delta(\Atr)+\Delta(\Btr)\geq 0$.
\end{itemize}
\end{lemma}
\begin{proof}
By applying posimodularity of cuts to the sets $A_s$ and $B_s$, we obtain:
\begin{equation}\label{eq:posi}
d(A_s)+d(B_s)=d(\Atr)+d(\Btr)+2|E(A_s\cap B_s,R)|\geq d(\Az)+d(\Bz)+2|E(A_s\cap B_s,R)|.
\end{equation}
On the other hand, we have that $d(A_s)\leq d(\Az)+1$ and $d(B_s)\leq d(\Bz)+1$. Hence we have that $|E(A_s\cap B_s,R)| \leq 1$ and the claimed case distinction follows from~\eqref{eq:truncs} and~\eqref{eq:posi}.
\end{proof}

\subsubsection{Decomposing sets of excess $2$}

Finally, we make a useful observation that will show a generic setting when Lemma~\ref{lem:ex2-red} can be applied.

\begin{lemma}\label{lem:cool}
Suppose $\Delta(A_s)=1$ and $A_s \neq \Az \cup \{s\}$. Then $A_s\setminus \{s\}\supsetneq \Az$ is a terminal-free excess-2 set, and $(A_s\setminus \Az)\setminus \{s\}$ has a decomposition $\{d,c_1,c_2,\ldots,c_r\}$ given by Lemma~\ref{lem:ex2-red}.
Moreover, $d=s'$ is the unique neighbor of $s$ in $G$.
\end{lemma}
\begin{proof}
The fact that $A_s\setminus \{s\}$ is an excess-2 set follows from the assumption that $s$ has degree exactly $1$ (due to the inapplicability of the Lonely Terminal Reduction), and its unique neighbor $s'$ does not belong to $\Az\cup \Bz$
and does belong to $A_s$ (because $G[A_s\setminus \Az]$ is connected by Lemma~\ref{lem:AsConn}).
Since $(A_s\setminus \Az)\setminus \{s\}$ is nonempty and terminal-free (by Assumption~\ref{ass:one-pair}),
it follows from Lemma~\ref{lem:ex2-red} that it has a decomposition of the form $\{c_1,c_2\}$ or $\{d,c_1,c_2,\ldots,c_r\}$, where $c_i$-s are pairwise nonadjacent and $\Az\cup \{c_i\}$ are excess-1 sets. Suppose $s'=c_i$ for some $i$. Then since $\Az\cup \{c_i\}$ is an excess-1 set, we would have that $\Az\cup \{c_i,s\}$ is an excess-0 set, and hence $(\Az\cup \{c_i,s\},B_t)$ would be an extension of $(\Az\cup \{s\},\Bz\cup \{t\})$ of strictly smaller cost than $(A_s,B_t)$, contradicting the definition of $(A_s,B_t)$. Hence $(A_s\setminus \Az)\setminus \{s\}$ has a decomposition of the form $\{d,c_1,c_2,\ldots,c_r\}$ and $s'=d$.
\end{proof}

We will need one more lemma that resolves corner cases when we apply Lemma~\ref{lem:cool}.

\begin{lemma}\label{lem:no-overlap}
Suppose $s$ satisfies the conditions of Lemma~\ref{lem:cool}, and let $\{d=s',c_1,c_2,\ldots,c_r\}$ be the obtained decomposition of $(A_s\setminus \Az)\setminus \{s\}$. Let $t'$ be the unique neighbor of $t$. Then $s'\neq t'$, and if $t'=c_i$ for some $i\in \{1,2,\ldots,r\}$, then there exists an optimum integral terminal separation $(\Aopt,\Bopt)$ that extends $(\Az,\Bz)$ and has $s\in \Bopt$ and $t\in \Aopt$.
\end{lemma}
\begin{proof}
The fact that $s'\neq t'$ follows from the inapplicability of the Common Neighbour Reduction. Suppose then that $t'=c_i$. From Lemma~\ref{lem:ex2-red} it follows that for some $p_i\geq 1$ and $x_i\geq 0$, there are $p_i$ edges between $s'$ and $c_i$, $p_i+x_i$ edges between $c_i$ and $\Az$, and $x_i+1$ edges between $c_i$ and $V(G)\setminus A_s$; one of these $x_i+1$ edges connects $c_i=t'$ with $t$.

Take any optimum integral separation $(\Aopt,\Bopt)$ extending $(\Az,\Bz)$ and suppose that $s\in \Aopt$ and $t\in \Bopt$. We can further assume that $s'\in \Aopt$ and $t'=c_i\in \Bopt$, because otherwise switching the sides of $s$ and $t$ would result in an integral separation of not larger cost that already fulfills the property we aim for. Recall that $c_i$ has $p_i$ edges to $s'$ (which is assigned to $\Aopt$), $p_i+x_i$ edges to $\Az$, and $x_i+1$ edges to other vertices of the graph. Since $p_i\geq 1$, we see that a strict majority of neighbors of $c_i$ are in $\Aopt$. Hence switching the side of $c_i$ from $\Aopt$ to $\Bopt$ strictly decreases the cost of the separation, a contradiction.
\end{proof}

Lemma~\ref{lem:no-overlap} enables us to perform a reduction step whenever a corner case appears in the analysis of vertices close to $s$ and $t$. We choose not to perform this reduction exhaustively, but rather to execute it on demand when such a case appears during branching.


\subsubsection{Fixing an edge $ss'$ or $tt'$}

In a few cases, we consider an improved branching set,
when in one branch we fix $\{s',s\}$ to belong to the left part and $t$ to belong to the right part, whereas in the second branch we fix vice versa. More precisely, we consider branches $(A_{ss'\to A},B_{ss'\to A})$ and $(A_{ss'\to B},B_{ss'\to B})$ that are minimum-cost terminal separations extending $(\Az\cup \{s,s'\},\Bz\cup \{t\})$ and $(\Az\cup \{t\},\Bz\cup \{s,s'\})$, computed using Theorem~\ref{thm:ptime}. Observe that there is some optimum solution that extends one of these branches: If in some optimum solution the vertices $s'$ and $s$ were assigned to different sides, then we could modify this solution by swapping the sides of $s$ and $t$. After this modification then solution has no larger cost due to $t$ having degree one, whereas the edge $ss'$ ceases to be cut by the solution. This justifies the correctness of this branching step; we shall henceforth call it {\em{branching on $\{s,t\}$ with fixing the edge $ss'$}}.
Symmetrically, we can define branching on $\{s,t\}$ with fixing the edge $tt'$.

\subsection{Case $(\Delta(A_s),\Delta(B_s))=(0,0)$}

We show that this case in fact never happens. From Lemma~\ref{lem:posi} we infer that $E(A_s \cap B_s, R) = \emptyset$, $\Atr=\Az$, and $\Btr=\Bz$. Hence, $A_s\setminus \Az=B_s\setminus \Bz=A_s\cap B_s$. As we argued earlier, we can assume that $A_s\cap B_s=\{s\}$ or $A_s\cap B_s=\{s,s'\}$ for $s'$ being the only neighbor of $s'$. 

In the first case, since the degree of $s$ is at most one, from $E(A_s \cap B_s, R) = \emptyset$ and $\Delta(A_s)=\Delta(B_s)=0$ we can infer that $s$ is an isolated terminal, which should have been removed by the Lonely Terminal Reduction. This contradicts the assumptions that no reduction rule is applicable.

In the second case, by $E(A_s \cap B_s, R) = \emptyset$ and $\Delta(A_s)=\Delta(B_s)=0$, we infer that $|E(s',\Az)|=|E(s',\Bz)|=x$ for some $x\geq 0$. If $x=0$, then $s'$ should have been reduced by the Pendant Reduction. On the other hand, if $x>0$ then the Boundary Reduction would have been triggered on $s'$. In both cases this is a contradiction.

\begin{figure}[H]
	\centering
	\begin{tikzpicture}

\draw[Bs] plot [smooth,tension=1.3] coordinates {(-2,2)  (0,0.4) (2,2)} to (-2,2);
\draw[As] plot [smooth,tension=1.3] coordinates {(-2,0)  (0,1.6) (2,0)} to (-2,0);

\Azero;
\begin{scope}[shift={(0,2)}] \Bzero; \node at (1,0.17) {$\Bz$}; \end{scope}
\node at (1,-0.2) {$\Az$};
\node[ABT,label=30:$\mathbf{s}$] (s) at (0.1,1) {};
\node[T, label=10:$\mathbf{t}$] (t) at (2.4,1) {};
\draw[term] (s) to[term,bend right=10] (t);

\begin{scope}[shift={(7,0)}]

\draw[Bs] plot [smooth,tension=1.9] coordinates {(-2,2)  (0,0.2) (2,2)} to (-2,2);
\draw[As] plot [smooth,tension=1.9] coordinates {(-2,0)  (0,1.8) (2,0)} to (-2,0);

\Azero;
\begin{scope}[shift={(0,2)}] \Bzero; \node at (1,0.17) {$\Bz$}; \end{scope}
\node at (1,-0.2) {$\Az$};
\node[ABT,label=30:$\mathbf{s}$] (s) at (0.6,1) {};
\node[AB, label=180:$\mathbf{s'}$] (sp) at (-0.6,1) {};
\node[T, label=10:$\mathbf{t}$] (t) at (2.4,0.9) {};
\draw[term] (s) to[term,bend right=10] (t);
\draw (s) to (sp);
\draw (sp) to (-0.33,2.15);
\draw (sp) to (-0.55,2.2);
\draw (sp) to (-0.77,2.15);
\node at (-0.23,1.6) {$p$};
\draw (sp) to (-0.33,-0.15);
\draw (sp) to (-0.55,-0.2);
\draw (sp) to (-0.77,-0.15);
\node at (-0.2,0.4) {$p$};

\end{scope}
\end{tikzpicture}
\figspace
	\caption{Case  $(\Delta(A_s),\Delta(B_s))=(0,0)$: a reduction is always immediately applicable. Terminal nodes are squares, paired with zig-zags. Extensions $A_s$ and $B_s$ are highlighted with light blue and red, respectively.}
\label{fig:case-00}
\end{figure}
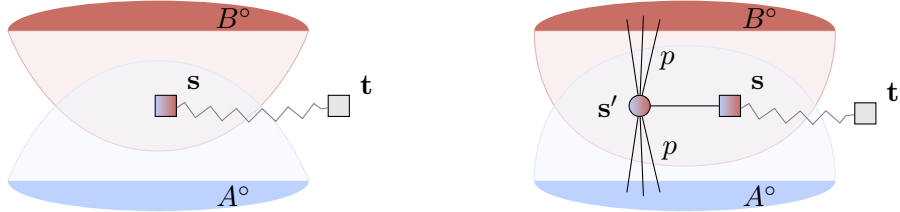


\subsection{Case $(\Delta(A_s),\Delta(B_s))=(1,0)$}

From Lemma~\ref{lem:posi} we infer that $E(A_s \cap B_s, R) = \emptyset$ and $\Delta(\Atr)+\Delta(\Btr)=1$. We have two subcases: either (a) $(\Delta(\Atr),\Delta(\Btr))=(1,0)$, or (b) $(\Delta(\Atr),\Delta(\Btr))=(0,1)$.

\subsubsection{Subcase (a): $(\Delta(\Atr),\Delta(\Btr))=(1,0)$}

By the equality condition in~\eqref{eq:truncs} we have that $\Btr=\Bz$, while $\Atr\supsetneq \Az$ is a terminal-free set of excess $1$. By the inapplicability of the Excess-1 Reduction, we infer that $\Atr=\Az\cup \{a\}$ for some nonterminal vertex $a$.



Set $A_s$ satisfies the conditions of Lemma~\ref{lem:cool}, so we can decompose $(A_s\setminus \Az)\setminus \{s\}$ into $\{d,c_1,c_2,\ldots,c_r\}$, where $d=s'$ is the unique neighbor of $s$. Since $\Delta(\Atr)=1$, we have that $B_s\supsetneq \Bz\cup \{s\}$ and hence by Lemma~\ref{lem:AsConn} it follows that $s'\in B_s$. By Assumption~\ref{ass:small-intersections} we infer that $A_s\cap B_s=\{s,s'\}$ and thus $\{a\}=\Atr\setminus \Az=\{c_1,c_2,\ldots,c_r\}$. Therefore $r=1$ and $c_1=a$.

Since $\Btr=\Bz$, $B_s\setminus \Bz = B_s \cap A_s = \{s,s'\}$.

By Lemma~\ref{lem:ex2-red} we have that $a$ has: $p$ edges to $s'$, $x+1$ edges to $V(G)\setminus (A_s \cup \Bz)$, $p+x$ edges to $\Az$ and no other edges, for some $p\geq 1, x\geq 0$. 
Since $B_s=\Bz\cup \{s',s\}$ is an excess-0 set and $E(s',R)=\emptyset$, we have that $|E(s',\Bz)|=p+|E(s',\Az)|$. In particular $|E(s',\Bz)|>0$, so since Boundary Reductions do not apply to $s'$, we have $E(s',\Az)=\emptyset$ and hence $|E(s',\Bz)|=p$.

\begin{figure}[H]
	\centering
	\begin{tikzpicture}[scale=1.1]

\draw[Bs] plot [smooth,tension=1.4] coordinates {(-2,2.2) (-0.1,1.05) (2,2.2)} to (-2,2.2);
\draw[As] plot [smooth,tension=0.5] coordinates {(-2,0) (-1.5,0.4) (-1.3,1.4) (-0.8,2) (0.8,2) (1.3,1.4) (1.5,0.4) (2,0)} to (-2,0);

\Azero;
\begin{scope}[shift={(0,2.2)}] \Bzero; \node at (1,0.17) {$\Bz$}; \end{scope}
\node at (1,-0.2) {$\Az$};
\node[ABT,label=30:$\mathbf{s}$] (s) at (0.6,1.4) {};
\node[AB, label=170:$\mathbf{s'}$] (sp) at (-0.6,1.4) {};
\node[A, label=0:$\mathbf{a}$] (a) at (-0.6,0.6) {};
\node (t) at (2,1.4) {};
\draw[term] (s) to[term] (t);
\draw (s) to (sp);
\draw (sp) to [bend left=15] (a);
\draw (sp) to [bend right=15] (a);
\node at (-0.37,1) {$p$};

\draw (sp) to (-0.33,2.3);
\draw (sp) to (-0.55,2.3);
\node at (-0.28,1.8) {$p$};
\draw (a) to (-0.11,-0.1);
\draw (a) to (-0.33,-0.15);
\draw (a) to (-0.55,-0.15);
\draw (a) to (-0.77,-0.1);
\node at (-0.44,-0.25) {$p+x$};

\draw (a) to (-1.7,0.5);
\draw (a) to (-1.7,0.65);
\draw (a) to (-1.6,0.8);
\node at (-1.85,0.95) {$x+1$};

\end{tikzpicture}
\figspace
	\caption{Case (1,0)(a): $(\Delta(A_s),\Delta(B_s))=(\Delta(\Atr),\Delta(\Btr))=(1,0)$. Extensions $A_s, B_s$ are highlighted.}
\label{fig:case-10-a}
\end{figure}
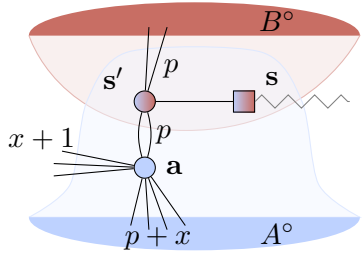

Consider now case $x=0$. Then $a$ has a unique edge $aa'$ with $a'\in R$. Consider first the case when $a'$ is a terminal, so in particular $aa'$ is the only edge incident to $a'$. If $a'=t$, then it is easy to see that $(\Az\cup \{a,t\},\Bz\cup \{s',s\})$ would be an extension of $(\Az,\Bz)$ of the same cost, which contradicts the maximality of $(\Az,\Bz)$. However, if $a'$ belonged to some other pair $\{a',a''\}\in \terms'$, then terminal separation $(A_s\cup \{a'\},B_t\cup \{a''\})$ would have the same cost as $(A_s,B_t)$, which contradicts the maximality of $(A_s,B_t)$. In either case we obtain a contradiction, which means that $a'$ is a nonterminal.

We claim that it is a safe reduction to contract the edge $aa'$; to prove this claim, it suffices to show that there exists an optimum integral terminal separation extending $(\Az,\Bz)$ where $a$ and $a'$ belong to the same side. Take any such integral terminal separation $(\Aopt,\Bopt)$, and assume that $a$ and $a'$ are on opposite sides. Clearly it cannot happen that $a\in \Bopt$ and $a'\in \Aopt$, because then moving $a$ from $\Bopt$ to $\Aopt$ would decrease the cost of the separation. Hence $a\in \Aopt$ and $a'\in \Bopt$. If $s'\in \Bopt$, then moving $a$ from $\Aopt$ to $\Bopt$ would decrease the cost of the separation, so also $s'\in \Aopt$. Construct a new integral separation $(\Aopt_m,\Bopt_m)$ from $(\Aopt,\Bopt)$ by moving $\{a,s'\}$ from $\Aopt$ to $\Bopt$. Then the cost of $(\Aopt_m,\Bopt_m)$ is not larger than that of $(\Aopt,\Bopt)$ (we could have broken the edge $s's$ instead of $aa'$), while both endpoints of $aa'$ belong to $\Aopt_m$.

This reasoning proves the correctness of the following step.
\begin{reductionstep}
Suppose $x=0$ and let $a'$ be the unique neighbor of $a$ in $R$; then $a'$ is a non-terminal. Merge $a$ with $a'$ and restart.
\end{reductionstep}






Henceforth we assume that $x>0$. We claim that now branching on the membership of $a$ leads to a good branch. More precisely, we perform the following branching.

\begin{branching}
If $x\geq 1$, recurse into two branches $(A_{a\to A},B_{a\to A})$ and $(A_{a\to B},B_{a\to B})$ that are minimum-cost maximal terminal separations extending $(\Az\cup \{a\},\Bz)$ and $(\Az,\Bz\cup \{a\})$, respectively.
\end{branching}

Of course, $(A_{a\to A},B_{a\to A})$ and $(A_{a\to B},B_{a\to B})$ are computed using the algorithm of Theorem~\ref{thm:ptime} in time $\Oh(k^{\Oh(1)} m)$. We are left with proving that after applying all the immediate reductions in each branch, we arrive at a good branching vector. For $X\in \{A,B\}$, let $t_{a\to X},\nu_{a\to X},k_{a\to X}$ be the changes of the components of the potential in respective branches, as we denote them in branching vectors.

Consider first the branch $(A_{a\to A},B_{a\to A})$. Then $p$ Boundary Reductions are triggered on vertex $s'$ (regardless of whether it is added or not to one of the sets $A_{a\to A},B_{a\to A}$). Hence $k_{a\to A}\geq p$. Moreover, the terminal pair $\{s,t\}$ either is already resolved by $(A_{a\to A},B_{a\to A})$ or gets reduced by the Lonely Terminal Reduction after applying the Boundary Reductions. Hence $t_{a\to A}\geq 1$. Finally, since $(\Az,\Bz)$ was maximal, we have that $\nu_{a\to A}\geq 1$. So the part of the branching vector corresponding to the branch $(A_{a\to A},B_{a\to A})$ is $[1,1,p]$, or better.

Consider now the second branch $(A_{a\to B},B_{a\to B})$. Then at least $|E(a,\Az)|=p+x$ Boundary Reductions are triggered, hence $k_{a\to B}\geq p+x$. Since $p\geq 1$ and $t$ is of degree $1$, $s'\in B_{a\to B}$ and without loss of generality we can assume $s\in B_{a\to B}$ and $t\in A_{a\to B}$. Hence $t_{a\to B}\geq 1$. If actually $t_{a\to A}\geq 2$ or $t_{a\to B}\geq 2$, then we arrive at a branching vector $[1,1,p;2,1,p]$ or better, which is good, so assume that $t_{a\to A}=t_{a\to B}=1$, that is, only the pair $\{s,t\}$ gets resolved.

We now claim that $\Delta(A_{a\to B})\geq 1$ and $\Delta(B_{a\to B})\geq 1$. The latter claim follows from Assumption~\ref{ass:pushing}, since then $B_{a\to B}$ contains only $s$ among the terminals (due to $t_{a\to B}=1$) and $a\in B_{a\to B}\setminus B_s$. For the former claim, suppose for the sake of contradiction that $d(A_{a\to B})=d(\Az)$. Recall that also $d(B_s)=d(\Bz)$, which means that $d(A_{a\to B})+d(B_s)=c(\Az,\Bz)$. From the posimodularity of cuts it now follows that one of the terminal separations $(A_{a\to B}\setminus B_s,B_s)$ and $(A_{a\to B},B_s \setminus A_{a\to B})$ has cost not larger than $(\Az,\Bz)$, while both of them resolve the terminal pair $\{s,t\}$. This is a contradiction with the maximality of $(\Az,\Bz)$. Hence we infer that $\Delta(A_{a\to B})\geq 1$ and $\Delta(B_{a\to B})\geq 1$, and so $\nu_{a\to B}\geq 2$.

Thus, branching into separations $(A_{a\to A},B_{a\to A})$ and $(A_{a\to B},B_{a\to B})$ leads to a branching vector $[1,1,p;1,2,p+x]$ or better. Recalling that $p,x>0$, observe that this branching vector can be not good only if $p=x=1$ and $\Delta(B_{a\to B})=1$. Hence, from now on let us analyze this case.

Since $\Delta(B_{a\to B})=1$, we have that $B_{a\to B}\setminus \{s\}$ is a terminal-free set of excess $2$, and hence we can apply Lemma~\ref{lem:ex2-red} to it: We have that $B_{a\to B}\setminus \{s\}$ has a decomposition of the form $\{c_1,c_2\}$ or $\{d,c_1,\ldots,c_r\}$. Note that $\Bz\cup \{s'\}$ is an excess-1 set, so $s'=c_i$ for some $i$. As $a\in B_{a\to B}$, $a$ is adjacent to $s'$, and $c_i$-s are pairwise non-adjacent, we must have that $a=d$ and we are dealing with a decomposition of the form $\{d,c_1,\ldots,c_r\}$. Observe that $\Bz\cup \{a,s'\}$ is a $\Bz$-extension of excess at least $1+x+1=3$; hence $B_{a\to B}\supsetneq \Bz \cup \{a,s',s\}$, and in particular $r>1$. Hence there exists some vertex $c_j\neq c_i=s'$. By Lemma~\ref{lem:ex2-red} we have that $c_j$ is adjacent both to $\Bz$ and to $a$. Hence, in the branch $(A_{a\to A},B_{a\to A})$ at least one Boundary Reduction is applied to $c_j$, regardless whether $c_j$ is assigned to $A_{a\to A}$, or $B_{a\to A}$, or neither of these sets. We did not include this Boundary Reduction in the previous calculations; this shows that we in fact pursue a branch with a branching vector $[1,1,2;1,2,2]$ or better, which is a good branching vector.





\subsubsection{Subcase (b): $(\Delta(\Atr),\Delta(\Btr))=(0,1)$}

By the equality condition in~\eqref{eq:truncs} we have that $\Atr=\Az$, while $\Btr\supsetneq \Bz$ is a terminal-free set of excess $1$. By the inapplicability of the Excess-1 Reduction, we infer that $\Btr=\Bz\cup \{b\}$ for some nonterminal vertex $b$. In particular $B_s\supsetneq\Bz\cup \{s\}$, so by Lemma~\ref{lem:AsConn} the unique neighbor $s'$ of $s$ belongs to $B_s$. Since $\Delta(B_s)=0$, we have that $B_s\setminus \{s\}$ is a terminal-free set of excess $1$, so it consists of a single vertex. However, this set already contains $b$. Hence we infer that $b=s'$ is the unique neighbor of $s$, $\Btr=\{s'\}\cup \Bz$, $B_s=\{s,s'\}\cup \Bz$. In particular $s'\notin A_s$, so by Lemma~\ref{lem:AsConn} it follows that $A_s=\Az \cup \{s\}$.

Let $x=|E(s',\Bz)|$. Since $\Delta(B_s)=0$, we also have $x=|E(s',V(G)\setminus B_s)|$. If $x=0$ then $s'$ would be only adjacent to $s$ and thus reducible by the Pendant Reduction. Hence, $x>0$. In particular, we infer that $E(s',\Az)=\emptyset$, since otherwise the Boundary Reduction could be applied to $s'$.







Let us now examine two possible branching steps. Firstly, consider just branching into two branches $(A_s,B_t)$ and $(A_t,B_s)$. In both cases, only one terminal pair $\{s,t\}$ gets resolved. In branch $(A_t,B_s)$, when $s$ is assigned to $B$, we pessimistically have no Boundary Reduction and no increase in the cost of the separation. In branch $(A_s,B_t)$, however, when $s$ is assigned to $A$, we have that $\Delta(A_s)=1$ and one Boundary Reduction is triggered on vertex $s'$ due to having both an edge to $s$ and to $\Bz$.

We now investigate the components of the branching vector when branching on $\{s,t\}$ with fixing $ss'$. If in one of the branches at least one more terminal pair gets resolved, then as argued in the beginning of this section we can just pursue the branching step, because it leads to a good branching vector. Hence, assume from now on that in both branches only the pair $\{s,t\}$ gets resolved. Since $A_s=\{s\}\cup \Az$ and $A_{ss'\to A}$ contains only $s$ among the terminals of $\terms'$, by Assumption~\ref{ass:pushing} we have that $\Delta(A_{ss'\to A})\geq 2$. Also, at least one Boundary Reduction is triggered on an edge between $s'$ and $\Bz$. In branch $(A_{ss'\to B},B_{ss'\to B})$, again we pessimistically have no Boundary Reduction and no increase in the cost of the separation.

\begin{figure}[H]
	\centering
	\begin{tikzpicture}

\draw[Bs] plot [smooth,tension=0.6] coordinates {(-2,2) (-1.5,1.6) (-1.1,0.7) (-0.5,0.2) (0.5,0.2) (1.2,0.7) (1.5,1.6) (2,2)} to (-2,2);
\draw[As] plot [smooth,tension=1.3] coordinates {(-2,0) (-0.1,0.95) (2,0)} to (-2,0);

\Azero;
\node at (1,-0.2) {$\Az$};
\begin{scope}[shift={(0,2)}] \Bzero; \node at (1,0.17) {$\Bz$}; \end{scope}

\node[ABT,label=180:$\mathbf{s}$] (s) at (0.6,0.5) {};
\node[B, label=0:$\mathbf{s'=b}$] (sp) at (-0.6,1.3) {};
\node (t) at (2,0.6) {};
\draw[term] (s) to[term] (t);
\draw (s) to (sp);
\draw (sp) to (-0.33,2.1);
\draw (sp) to (-0.55,2.1);
\draw (sp) to (-0.77,2.1);
\node at (-0.45,2.2) {$x$};

\draw (sp) to (-1.7,0.65);
\draw (sp) to (-1.75,0.8);
\draw (sp) to (-1.7,0.99);
\node at (-1.9,0.8) {$x$};

\begin{scope}[shift={(6,0)}]

\Azero;
\node at (1,-0.2) {$\Az$};

\node[T,label=180:$\mathbf{s}$] (s) at (0.6,1.5) {};
\node[V, label=0:$\mathbf{s'}$] (sp) at (-0.6,0.7) {};
\node (t) at (2,1.4) {};
\draw[term] (s) to[term] (t);
\draw (s) to (sp);

\draw (sp) to (-0.3,-0.1);
\draw (sp) to (-0.66,-0.1);
\node at (-0.45,-0.25) {$x$};

\draw (sp) to (-1.7,0.65);
\draw (sp) to (-1.7,0.99);
\node at (-1.9,0.8) {$x$};

\end{scope}

\end{tikzpicture}
\figspace
	\caption{Case (1,0)(b): $(\Delta(\Atr),\Delta(\Btr))=(0,1)$. This gives rise to an \emph{antenna}, which has to be analyzed together with the other terminal.
	The right side shows an antenna with natural side $\Az$; note the definition does not mention $A_s, B_s$, it only relies on the behaviour of extensions containing $s$ or $s'$.}
\label{fig:case-10-b-antenna}
\end{figure}
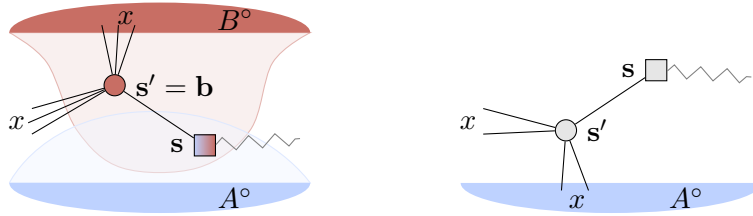

A terminal $s$ with the behaviour as described above will be actually the most problematic case for our branching algorithm. Let us define this setting formally.
\begin{definition}
A terminal $s$ is called an {\em{antenna}} if the following conditions hold:
\begin{itemize}
\item The only neighbor $s'$ of $s$ is a nonterminal, has $x>0$ edges to one of the sets $\Az$ or $\Bz$, no edge to the second one, and $x$ edges to $V(G)\setminus (\Az\cup\Bz\cup \{s\})$. The side $S\in \{\Az,\Bz\}$ to which $s'$ is adjacent is called the {\em{natural}} side of $s$, and the second one is called the {\em{unnatural}} side of $s$.
\item Let $S$ and $\overline{S}$ be the natural and unnatural side of $s$, respectively. Then 
\begin{itemize}
\item For any $X$ with $S\cup \{s,s'\} \subsetneq X \subseteq V(G)\setminus \overline{S}$ that contains only $s$ among the terminals from $\terms'$, it holds that $\Delta(X)\geq 1$.
\item For any $Y$ with $\overline{S}\cup \{s\}\subseteq Y\subseteq V(G)\setminus S$ that contains only $s$ among the terminals from $\terms'$, it holds that $\Delta(Y)\geq 1$. If moreover $Y$ contains at least one more vertex than $\overline{S}\cup \{s\}$, then $\Delta(Y)\geq 2$.
\end{itemize}
\end{itemize}
\end{definition}

The discussion above together with Lemmas~\ref{lem:lb} and Assumption~\ref{ass:pushing} shows that in this case $s$ is an antenna with natural side $\Bz$. Obviously, in the symmetric subcase when $(\Delta(A_s),\Delta(B_s))=(0,1)$ and $(\Delta(\Atr),\Delta(\Btr))=(1,0)$ we obtain that $s$ is an antenna with natural side $\Az$. 

The idea now is {\em{not}} to perform any branching step on an antenna, but rather to branch on the situation around the second terminal $t$, i.e., swap the roles of $t$ and $s$ and restart the analysis.
In other words, we will show that if the analysis of the second terminal $t$ does not reveal that it is an antenna (it conforms to cases $(0,0)$, $(1,0)a$, or $(1,1)$), then a branching step leading to a good branching vector can be found on that side. We will be thus left with the case when both $s$ and $t$ are antennas, which we aim to resolve now by exposing a branching strategy leading to a good branching vector.

Therefore, assume that $s$ and $t$ are both antennas, and let $s'$ and $t'$ be their unique neighbors, respectively. By the inapplicability of the Common Neighbor Reduction, $s' \neq t'$. First, suppose that $s$ and $t$ have different natural sides, say $s$ has natural side $\Az$ and $t$ has natural side $\Bz$. However, then $(\Az\cup \{s,s'\},\Bz\cup \{t,t'\})$ would be a terminal separation that has the same cost as $(\Az,\Bz)$, which contradicts the maximality of $(\Az,\Bz)$.

Hence, assume that $s$ and $t$ have the same natural side. W.l.o.g. suppose that it is $\Bz$. Let $x=|E(s',\Bz)|$ and $y=|E(t',\Bz)|$; recall that $x,y\geq 1$. Consider two possible branching steps: we can branch on $\{s,t\}$ with fixing $ss'$ or with fixing $tt'$. Consider first fixing edge $ss'$, and let $(A_{ss'\to A},B_{ss'\to A})$ and $(A_{ss'\to B},B_{ss'\to B})$ be the branches. By the definition of the antenna we have that $\Delta(A_{ss'\to A})\geq 2$ and $\Delta(A_{ss'\to B})\geq 1$. Also, in branch $(A_{ss'\to A},B_{ss'\to A})$ we have at least $x$ Boundary Reductions triggered on edges incident to $s'$, whereas in branch $(A_{ss'\to B},B_{ss'\to B})$ we have at least one Boundary Reduction triggered edges incident to $t'$. Thus we obtain a branching vector $[1,2,x;1,1,1]$ or better, and a symmetric reasoning for fixing $tt'$ leads to branching vector $[1,1,1;1,2,y]$, or better. Note that one of these vectors is good if $\max(x,y)\geq 3$.
Furthermore, such a branching also leads to a good vector if $s'$ or $t'$ is adjacent to some terminal other than $s$ or $t$, respectively,
  as then in at least one branch a second terminal pair would be resolved.
Hence, if this is the case, we pursue the respective branching step.

\begin{branching}
If $\max(x,y)\geq 3$, or there is a terminal in $\terms'$ different than $s$ or $t$ adjacent to $s'$ or $t'$,
then pursue branching on $\{s,t\}$ with fixing the respective edge $ss'$ or $tt'$.
\end{branching}

From now on we assume that $x,y\leq 2$ and that no other terminal than $s$ and $t$ is adjacent to $s'$ nor $t'$.

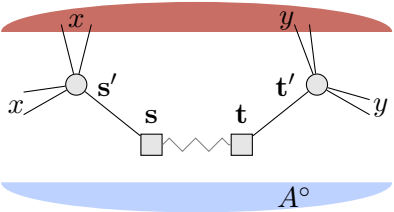
\begin{figure}[H]
	\centering
	\begin{tikzpicture}

\begin{scope}[xscale=1.3]
	\Azero; \node at (1,-0.2) {$\Az$};
	\begin{scope}[shift={(0,2)}]
		\Bzero; 
	\end{scope}
\end{scope}

\node[T,label=90:$\mathbf{s}$] (s) at (-0.6,0.5) {};
\node[V, label=0:$\mathbf{s'}$] (sp) at (-1.6,1.3) {};
\node[T,label=90:$\mathbf{t}$] (t) at (0.6,0.5) {};
\node[V, label=180:$\mathbf{t'}$] (tp) at (1.6,1.3) {};

\draw[term] (s) to[term] (t);
\draw (s) to (sp);
\draw (t) to (tp);

\draw (sp) to (-1.4,2.1);
\draw (sp) to (-1.8,2.1);
\node at (-1.6,2.15) {$x$};
\draw (sp) to (-2.3,0.9);
\draw (sp) to (-2.3,1.2);
\node at (-2.4,1) {$x$};

\draw (tp) to (1.3,2.1);
\draw (tp) to (1.5,2.1);
\node at (1.2,2.15) {$y$};
\draw (tp) to (2.3,1.1);
\draw (tp) to (2.3,0.9);
\node at (2.45,1) {$y$};

\end{tikzpicture}
\figspace
	\caption{Case (1,0)(b) on both $s$ and $t$, where furthermore the antennas have the same natural side.}
\label{fig:case-10-b-two-antennas}
\end{figure}

Consider now the case when there is a vertex $a$ such that all edges of $E(s',V(G)\setminus (\Bz\cup \{s\}))$ have $a$ as the endpoint different than $s'$. This encompasses the cases when $x=1$ and when $x=2$ but the considered edges connecting $s'$ with $V(G)\setminus (\Az\cup \{s\})$ have the same second endpoint. We claim that then it is a safe reduction to merge $a$ and $s'$. To prove this claim, we need to show that there exists an optimum integral terminal separation $(\Aopt,\Bopt)$ extending $(\Az,\Bz)$ where $a$ and $s'$ are on the same side. Take any such integral terminal separation $(\Aopt,\Bopt)$, and assume that $a$ and $s'$ are on opposite sides. Clearly it cannot happen that $a\in \Bopt$ and $s'\in \Aopt$, because then moving $s'$ from $\Aopt$ to $\Bopt$ would decrease the cost of the separation. Hence $a\in \Aopt$ and $s'\in \Bopt$. This implies that $s\in\Bopt$, since otherwise we could improve the cost of the separation by moving $s'$ from $\Bopt$ to $\Aopt$. Therefore $t\in \Aopt$. Consider modifying $(\Aopt,\Bopt)$ into $(\Aopt_m,\Bopt_m)$ by 
\begin{itemize}
\item moving $s'$ and $s$ from $\Bopt$ to $\Aopt$, and
\item moving $t$ and $t'$ from $\Aopt$ to $\Bopt$, provided $t'$ was not already included in $\Bopt$.
\end{itemize}
It is easy to see that $(\Aopt_m,\Bopt_m)$ is still an integral terminal separation extending $(\Az,\Bz)$ and its cost is no larger than that of $(\Aopt,\Bopt)$. Hence, it is optimum as well. However, in $(\Aopt_m,\Bopt_m)$ it holds that $s'$ and $a$ are on the same side.

This reasoning and its symmetric version for $t'$ imply the correctness of the following reduction step.
Note that $a$ is not a terminal, as we have already excluded this case in the previous branching step.

\begin{reductionstep}
If $|N(s')\setminus (\Bz\cup \{s\})|=1$, then merge $s'$ with its unique neighbor in $V(G)\setminus (\Bz\cup \{s\})$ and restart. If $|N(t')\setminus (\Bz\cup \{t\})|=1$, then merge $t'$ with its unique neighbor in $V(G)\setminus (\Bz\cup \{t\})$ and restart.
\end{reductionstep}

We are left with the case when $x=y=2$ and both $s'$ and $t'$ have two neighbors outside $\Bz\cup \{s,t\}$; these neighbors will be called {\em{external}}. We claim that then just pursuing branching on $\{s,t\}$ with fixed $ss'$ leads to a good branching vector.

\begin{branching}
If $x=y=2$ and $|N(s')\setminus (\Bz\cup \{s\})|=|N(t')\setminus (\Bz\cup \{t\})|=2$, then pursue branching on $\{s,t\}$ with fixing $ss'$.
\end{branching}

Let the branches be $(A_{ss'\to A},B_{ss'\to A})$ and $(A_{ss'\to B},B_{ss'\to B})$. Recall that $\Delta(A_{ss'\to A})\geq 2$. If actually $\Delta(A_{ss'\to A})\geq 3$, then we would already have a good branching vector $[1,3,2;1,1,1]$ or better, so assume henceforth that $\Delta(A_{ss'\to A})=2$. Consider now set $A'=A_{ss'\to A}\setminus \{s,s'\}$. If $A'$ did not contain both external neighbors of $s'$, then a simple edge count shows that $A'$ would be a terminal-free set of excess at most $-2$ (if it contains no external neighbor of $s'$) or $0$ (if it contains one external neighbor of $s'$). In both cases this is a contradiction with the maximality of $(\Az,\Bz)$. Hence, $A'$ contains both external neighbors of $s'$, and $A'$ is a terminal-free set of excess $2$. By Lemma~\ref{lem:ex2-red}, we can decompose $A'$ as $\{c_1,c_2\}$ or $\{d,c_1,\ldots,c_r\}$. Since $s'$ has two different neighbors in $A'$, at least one of them is $c_i$ for some $i$. However, by Lemma~\ref{lem:ex2-red} each $c_i$ is adjacent to $\Az$, and hence in branch $(A_{ss'\to B},B_{ss'\to B})$ at least one Boundary Reduction is triggered on vertex $c_i$ (regardless whether this vertex is assigned to $A_{ss'\to B}$, or to $B_{ss'\to B}$, or to neither of these sets). In our earlier calculations we did not account for this Boundary Reduction, so in fact we obtain branching vector $[1,2,2;1,1,2]$ or better, which is a good branching vector.






\subsection{Case $(\Delta(A_s),\Delta(B_s))=(1,1)$}

By Lemma~\ref{lem:posi}, we have three non-symmetric subcases:
\begin{itemize}
\item[(a)] $|E(A_s \cap B_s, R)|=1$, $\Atr=\Az$, $\Btr=\Bz$;
\item[(b)] $E(A_s \cap B_s, R) = \emptyset$, $\Delta(\Atr)=\Delta(\Btr)=1$;
\item[(c)] $E(A_s \cap B_s, R) = \emptyset$, $\Delta(\Atr)=0$, $\Delta(\Btr)=2$.
\end{itemize}
The case when $E(A_s \cap B_s, R) = \emptyset$, $\Delta(\Atr)=2$, $\Delta(\Btr)=0$, is symmetric to case (c).

The algorithm proceeds as follows: It investigates every terminal pair $\{s,t\}\in \terms'$, and investigates the case given by this terminal pair when considered as $\{s,t\}$ (i.e., looking from the side of $s$), and when considered as $\{t,s\}$ (i.e., looking from the side of $t$). If in any of these checks, for any terminal pair, case $(0,0)$ or $(1,0)(a)$ is discovered, the algorithm pursues the respective Reduction Step or Branching Step, as described in the previous sections. Otherwise, we can assume the following:
\begin{assumption}\label{ass:cases-left}
Every terminal of $\terms'$ is either an antenna, or investigating the basic branch of the respective terminal pair from its side yields case $(1,1)$ (has {\em{type (1,1)}}).
\end{assumption}
In the following we will use this property heavily in order to be able to reason about the total increase in the cost of the separation, also on the side of the second terminal from the pair we are currently investigating. 

\subsubsection{Case (a): $|E(A_s \cap B_s, R)|=1$, $\Atr=\Az$, $\Btr=\Bz$}


Let $Z=A_s\cap B_s=A_s\setminus \Az=B_s\setminus \Bz$. By Assumption~\ref{ass:small-intersections}, we have that $Z=\{s\}$ or $Z=\{s,s'\}$, where $s'$ is the unique neighbor of $s$. 

Suppose first that $Z=\{s,s'\}$. Let $x=|E(s',\Az)|$ and $y=|E(s',\Bz)|$. Since $|E(s',R)|=|E(Z,R)|=1$ and both $A_s$ and $B_s$ are excess-1 sets, we infer that $x=y$. Consequently it must hold that $x=y=0$, because otherwise the Boundary Reduction would apply to $s'$. Thus, $s'$ is a vertex of degree $2$ with one neighbor $r$ in $R$ and the second being $s$. Then the Pendant Reduction would apply to $X = \{s'\}$, a contradiction.

Therefore, we have that $Z=\{s\}$. Let $s'$ be the unique neighbor of $s$. We pursue branching on the pair $\{s,t\}$ with fixing edge $ss'$, i.e., branch into two subcases $(A_{ss'\to A},B_{ss'\to A})$ and $(A_{ss'\to B},B_{ss'\to B})$ that are minimum-cost terminal separations extending $(\Az\cup \{s,s'\},\Bz\cup \{t\})$ and $(\Az\cup \{t\},\Bz\cup \{s,s'\})$, respectively.

\begin{branching}
Pursue branching on $\{s,t\}$ with fixing $ss'$.
\end{branching}

Obviously, as explained in the beginning of this section, if any of the resulting branches resolves one more terminal pair, then the branching vector is good. Therefore, suppose that in both branches only the pair $\{s,t\}$ gets resolved. By Assumption~\ref{ass:pushing}, we have that $\Delta(A_{ss'\to A})\geq 2$ and $\Delta(B_{ss'\to B})\geq 2$. By Assumption~\ref{ass:cases-left}, terminal $t$ is either of type $(1,1)$ or is an antenna. In the former case, by Lemma~\ref{lem:lb} we have that $\Delta(B_{ss'\to A})\geq 1$ and $\Delta(A_{ss'\to B})\geq 1$. Hence we arrive at branching vector $[1,3,0;1,3,0]$ or better, which is a good branching vector. In the latter case, by the definition of an antenna we have that $\Delta(B_{ss'\to A})\geq 1$ or $\Delta(A_{ss'\to B})\geq 1$, depending on whether $\Az$ or $\Bz$ is natural for $s$. Also, in the same branch where respective inequality holds, one Boundary Reduction gets applied on the unique neighbor of $t$. Thus we arrive at branching vector $[1,2,0;1,3,1]$, or $[1,3,1;1,2,0]$ or better (depending on which side is natural for $t$), which is a good branching vector.



\subsubsection{Case (b): $E(A_s \cap B_s, R) = \emptyset$, $\Delta(\Atr)=\Delta(\Btr)=1$}


Since $\Atr$ and $\Btr$ are terminal-free sets of excess $1$, by the inapplicability of the Excess-1 Reduction we infer that $\Atr = \Az \cup \{a\}$ and $\Btr = \Bz \cup \{b\}$ for some distinct nonterminal vertices $a,b$. In particular, both $A_s\setminus \Az$ and $B_s\setminus \Bz$ contain at least one more vertex than $s$. Hence, by Lemma~\ref{lem:AsConn} we infer that if $s'$ is the unique neighbor of $s$, then $s'\in A_s$ and $s'\in B_s$. From Assumption~\ref{ass:small-intersections} it follows that $A_s\cap B_s=\{s,s'\}$. Hence $A_s \setminus \Az=\{a,s,s'\}$ and $B_s \setminus \Bz=\{b,s,s'\}$.

Since $\Delta(A_s)=1$ and $A_s\neq \Az\cup \{s\}$, we can apply Lemma~\ref{lem:cool} to it and infer that $A_s\setminus \{s\}$ has a decomposition $\{d,c_1\}$ with $s'=d$ and $c_1=a$. Consequently, by Lemma~\ref{lem:ex2-red} we infer that for some $p\geq 1$ and $x\geq 0$, we have $|E(s',a)|=p$, $|E(a,\Az)|=p+x$, and $|E(a,V(G)\setminus A_s)|=x+1$. Also, there is no edge between $a$ and $\Bz$, because then the Boundary Reduction would be applicable to $a$. A symmetric reasoning shows that for some $q\geq 1$ and $y\geq 0$, we have $|E(s',b)|=q$, $|E(b,\Bz)|=q+y$, $|E(b,V(G)\setminus B_s)|=y+1$, and there is no edge between $b$ and $\Az$.

Vertex $s'$ cannot be connected both to $\Az$ and to $\Bz$, because then the Boundary Reduction would be applicable to it. Hence, w.l.o.g. assume that $E(s',\Az)=\emptyset$.
Let $q'=|E(s',\Bz)|$. Since $E(A_s \cap B_s, R) = \emptyset$ and both $A_s$ and $\Az\cup \{a\}$ are sets of excess $1$, we infer that $p=q+q'$.

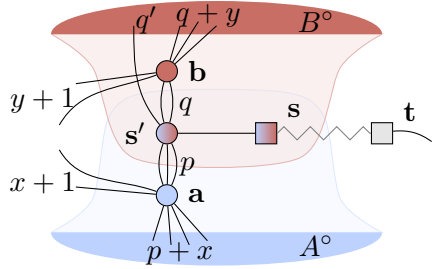
\begin{figure}[H]
	\centering
	\begin{tikzpicture}[scale=1.1]

\draw[As] plot [smooth,tension=0.5] coordinates {(-2,0) (-1.5,0.4) (-1.3,1.4) (-0.8,1.7) (0.8,1.7) (1.3,1.4) (1.5,0.4) (2,0)} to (-2,0);
\draw[Bs] plot [smooth,tension=0.5] coordinates {(-2,2.4) (-1.5,2) (-1,0.9) (1,0.9) (1.5,2) (2,2.4)} to (-2,2.4);

\Azero;
\begin{scope}[shift={(0,2.4)}] \Bzero; \node at (1.2,0.15) {$\Bz$}; \end{scope}
\node at (1.2,-0.16) {$\Az$};
\node[ABT,label=30:$\mathbf{s}$] (s) at (0.6,1.2) {};
\node[AB, label=180:$\mathbf{s'}$] (sp) at (-0.6,1.2) {};
\node[A, label=0:$\mathbf{a}$] (a) at (-0.6,0.45) {};
\node[B, label=0:$\mathbf{b}$] (b) at (-0.6,1.95) {};
\node[T, label=10:$\mathbf{t}$] (t) at (2,1.2) {};
\draw[term] (s) to[term] (t);
\draw (s) to (sp);
\draw (t) to[bend left=20] (2.6,1.1) {};
\draw (sp) to [bend left=2] (a);
\draw (sp) to [bend right=20] (a);
\draw (sp) to [bend left=25] (a);
\node at (-0.35,0.8) {$p$};

\draw (sp) to[bend left=15] (b);
\draw (sp) to[bend right=15] (b);
\node at (-0.37,1.52) {$q$};
\draw (sp) to[out=130,in=-90] (-1,2.5);
\node at (-0.85,2.57) {$q'$};

\draw (a) to (-0.11,-0.1);
\draw (a) to (-0.33,-0.15);
\draw (a) to (-0.55,-0.15);
\draw (a) to (-0.77,-0.1);
\node at (-0.44,-0.25) {$p+x$};

\draw (a) to (-1.7,0.55);
\draw (a) to[in=-60,out=160] (-1.9,1);
\node at (-2.1,0.6) {$x+1$};
\draw (b) to (-0.44,2.5);
\draw (b) to (-0.22,2.55);
\draw (b) to (-0.00,2.5);
\node at (-0.1,2.63) {$q+y$};

\draw (b) to (-1.7,1.8);
\draw (b) to[in=60,out=-160] (-1.9,1.3);
\node at (-2.1,1.6) {$y+1$};
\end{tikzpicture}
\figspace
	\caption{Case (1,1)(b): $\Delta(B_s)=\Delta(A_s)=\Delta(\Atr)=\Delta(\Btr)=1$. Note the edges counted in $x+1$ and $y+1$ may both include a common edge between $a$ and $b$.}
\label{fig:case-11-b}
\end{figure}

For the sake of further argumentation, we now resolve the case when $t'=a$ or $t'=b$, where $t'$ is the unique neighbor of $t$. Then, Lemma~\ref{lem:no-overlap} and its symmetric variant imply that the pair $\{s,t\}$ can be assigned greedily. More precisely, the following reduction step is correct.

\begin{reductionstep}\label{red:(1,1)b-corner}
If $t'=a$ then assign $s$ to the $B$-side and $t$ to the $A$-side, i.e., proceed with instance $(A_t,B_s)$. If $t'=b$ then assign $s$ to the $A$-side and $t$ to the $B$-side, i.e., proceed with instance $(A_s,B_t)$.
\end{reductionstep}




Henceforth we assume that $t'\neq a$ and $t'\neq b$. Since $A_s=\{a,s,s'\}$ and $B_s=\{b,s,s'\}$, by the inpplicability of Common Neighbor Reduction we infer that $t'\notin A_s$ and $t'\notin B_s$.

The crucial observation now is that we can fix both edges $ss'$ and $tt'$ at the same time.

\begin{lemma}\label{lem:double-fixing}
There exists an optimum integral terminal separation $(\Aopt,\Bopt)$ where either $\{s,s'\}\subseteq \Aopt$ and $\{t,t'\}\subseteq \Bopt$, or $\{s,s'\}\subseteq \Bopt$ and $\{t,t'\}\subseteq \Aopt$.
\end{lemma}
\begin{proof}
Let us take any optimum integral terminal separation $(\Aopt,\Bopt)$. If the condition of the lemma is not satisfied, then swapping the sides of $s$ and $t$ does not change the cost of the separation. Let us then assume that the edge $tt'$ is not cut in the solution. Hence, without loss of generality we assume that $s',t,t'\in \Bopt$ and $s\in \Aopt$; the rest of the reasoning will be independent of the choice we made earlier that there are no edges between $s$ and $\Az$, so we are indeed not losing generality here.

Consider $A=\Aopt\cup A_s$. By the submodularity of cuts we have that
$$d(\Aopt\cap A_s)+d(A)\leq d(\Aopt)+d(A_s).$$
However, we have that $d(A_s)\leq d(\Aopt\cap A_s)$ because otherwise we would be able to replace $A_s$ with $\Aopt\cap A_s$ in separation $(A_s,B_t)$ thus decreasing its cost while preserving the fact that it extends $(\Az\cup \{s\},\Bz\cup \{t\})$. Hence, we infer that $d(\Aopt)\geq d(A)$. Observe that since $A_s\setminus (\Az\cup \{s\})$ is terminal-free, then $(A,V(G)\setminus A)$ is also an integral terminal separation, and its cost is $d(A)\leq d(\Aopt)=c(\Aopt,\Bopt)$. Hence, $(A,V(G)\setminus A)$ is also an optimum integral terminal separation. Since $s'\in A_s$ and $t'\notin A_s$, we infer that edges $ss'$ and $tt'$ are not cut in $(A,V(G)\setminus A)$, as was requested.
\end{proof}

Lemma~\ref{lem:double-fixing} justifies the correctness of branching on $\{s,t\}$ with both $ss'$ and $tt'$ fixed. More precisely, we branch into separations $(A_{ss'\to A},B_{ss'\to A})$ and $(A_{ss'\to B},B_{ss'\to B})$ that are minimum-cost terminal separations extending $(\Az\cup \{s,s'\},\Bz\cup \{t,t\})$ and $(\Az\cup \{t,t'\},\Bz\cup \{s,s'\})$, computed using Theorem~\ref{thm:ptime}. 

\begin{branching}
Pursue branching on $\{s,t\}$ with fixing both $ss'$ and $tt'$.
\end{branching}

As we argued at the beginning of this section, if in any of these branches at least one more terminal pair gets resolved, then we arrive at a good branching vector; hence assume that this is not the case.

By Lemma~\ref{lem:lb} we have that $\Delta(A_{ss'\to A})\geq 1$ and $\Delta(B_{ss'\to B})\geq 1$. Also, in both branches the Boundary Reduction will be applied at least $p$ times: either $p$ times on $a$ (provided $s'$ is assigned to the $B$-side), or $q$ times on $b$ and $q'$ times on edges between $s'$ and $\Bz$ (provided $s'$ is assigned to the $A$-side).

We now calculate the branching vectors when performing this branching.

Suppose first that $t$ is an antenna, then by the definition of the antenna we have that $\Delta(A_{ss'\to B})\geq 2$ or $\Delta(B_{ss'\to A})\geq 2$, depending whether $\Bz$ or $\Az$ is the natural side of $t$. Moreover, in the same branch, one Boundary Reduction is triggered on an edge between $t'$ and the natural side of $t$ in the branch. By $t'\notin A_s\cup B_s$, we know that this Boundary Reduction was not accounted for in the previous calculations. Hence, we obtain branching vector $[1,1,p;1,3,p+1]$, or $[1,3,p+1;1,1,p]$, or better, depending on the natural side of $t$. Since $p\geq 1$, these branching vectors are good.

Suppose now that $t$ is of type $(1,1)$, and moreover that investigation of its situation also leads to the same case (b). Then we have that $\Delta(B_{ss'\to A})\geq 1$ and $\Delta(A_{ss'\to B})\geq 1$ by Lemma~\ref{lem:lb}. Moreover, in both of the branches, at least one Boundary Reduction is triggered that reduces some edge incident to $t'$. It is easy to see that the applicability of this Boundary Reduction could not be spoiled by the application of the $p$ Boundary Reductions on the side of $s$, because $t'\notin A_s\cup B_s$ and $s'$ and $t'$ are assigned to different sides. Thus, we arrive at branching vector $[1,2,p+1;1,2,p+1]$, which is $[1,2,2;1,2,2]$ or better, and hence good.

Finally, we are left with the case when $t$ is of type $(1,1)$, and the investigation of its situation also leads to case (c). Similarly as in the previous paragraph, we have that $\Delta(B_{ss'\to A})\geq 1$ and $\Delta(A_{ss'\to B})\geq 1$. Moreover, as we shall see in the next section, in at least one branch, one additional Boundary Reduction will be triggered that will reduce an edge incident to $t'$. Moreover, the applicability of this Boundary Reduction will not be spoiled by the application of the previous $p$ Boundary Reductions on the side of $s$, for the same reason as in the previous paragraph; that is, $t'\notin A_s\cup B_s$ and $s'$ and $t'$ are assigned to different sides. Hence we arrive at branching vector $[1,2,p+1;1,2,p]$, or $[1,2,p;1,2,p+1]$, or better. All these vectors are good for $p>0$.

\subsubsection{Case (c): $E(A_s \cap B_s, R) = \emptyset$, $\Atr=\Az$, $\Delta(\Btr)=2$}


Since $\Delta(B_s)=1$, we have that $B_s\setminus \{s\}$ is a terminal-free extension of $\Bz$ of excess $2$ and we can apply Lemma~\ref{lem:cool} to decompose it as $\{d,c_1,c_2,\ldots,c_r\}$, where $d=s'$ is the unique neighbor of $s$. Let $p_i=|E(c_i,s')|$, for $i=1,2,\ldots,r$. Recalling Lemma~\ref{lem:ex2-Bside}, let $\sigma=|E(s',\Bz)|+\sum_{i=1}^r p_i=|E(s',\{c_1,\ldots,c_r\}\cup \Bz)|$ be the number of Boundary Reductions that are immediately triggered within $B_s\setminus \{s\}$ in any branch when $s'$ is assigned to the $A$-side. By Lemma~\ref{lem:ex2-Bside} we have that $\sigma>0$. This justifies the claim that was left in our analysis of Case $(1,1)b$, where we argued for the applicability of one additional Boundary Reduction.

Before we proceed, let us exclude the corner case when $t'=c_i$ for some $i\in \{1,2,\ldots,r\}$, where $t'$ is the unique neighbor of $t$. Lemma~\ref{lem:no-overlap} justifies the correctness of the following reduction step.

\begin{reductionstep}
If $t'=c_i$ for some $i\in \{1,2,\ldots,r\}$, then assign $s$ to the $A$-side and $t$ to the $B$-side, i.e., proceed with instance $(A_s,B_t)$.
\end{reductionstep}

Since $t'\neq s'$ by the inapplicability of the Common Neighbor Reduction, henceforth we can assume that $t'\notin B_s$. 

Since $\Atr=\Az$, by Assumption~\ref{ass:small-intersections} we have two cases: either $A_s \setminus \Az=\{s,s'\}$ or $A_s \setminus \Az=\{s\}$.

\paragraph*{Subcase (c.i): $A_s\setminus \Az=\{s,s'\}$.} 

Let $p=|E(s',\Az)|$. Since $E(A_s \cap B_s, R) = \emptyset$ and $A_s$ has excess $1$, we infer that there are $p+1$ edges from $s'$ to $\Btr=\{c_1,c_2,\ldots,c_r\}\cup \Bz$, and hence $\sigma=p+1$. Observe that $p\geq 1$, because if $p=0$ the $s'$ would be adjacent only to $s$ and to a vertex in $\Btr$, and hence the Pendant Reduction would be applicable to $\{s'\}$. Hence in this case $\sigma\geq 2$.

\begin{figure}[H]
	\centering
	\begin{tikzpicture}

\draw[Bs] plot [smooth,tension=0.6] coordinates {(-2,2) (-1.5,1.6) (-1.1,0.7) (-0.5,0.2) (0.5,0.2) (1.2,0.7) (1.5,1.6) (2,2)} to (-2,2);
\draw[As] plot [smooth,tension=1.3] coordinates {(-2,0) (-0.1,0.95) (2,0)} to (-2,0);

\Azero;
\node at (1,-0.2) {$\Az$};
\begin{scope}[shift={(0,2)}] \Bzero; \end{scope}

\node[ABT,label=30:$\mathbf{s}$] (s) at (0.6,0.5) {};
\node[AB, label=180:$\mathbf{s'}$] (sp) at (-0.6,0.5) {};
\node[B, label={right:$\mathbf{c_1}$}] (c1) at (-0.8,1.5) {};
\node[B, label=0:$\mathbf{c_2}$] (c3) at (0.5,1.5) {};
\node (t) at (2,0.6) {};
\draw[term] (s) to[term] (t);
\draw (s) to (sp);
\draw (sp) to[bend right=2] (c1);
\draw (sp) to[bend left=15] (c1);
\draw (sp) to[bend right=18] (c1);
\draw (sp) to (c3);

\draw (sp) to (-0.33,-0.1);
\draw (sp) to (-0.55,-0.1);
\draw (sp) to (-0.77,-0.1);
\node at (-0.55,-0.25) {$p$};

\draw (c1) to (-1.0,2.1);
\draw (c1) to (-0.8,2.1);
\draw (c1) to (-0.6,2.1);
\draw (c1) to[bend left] (-2,1);


\draw (c3) to (0.6,2.1);
\draw (c3) to (0.8,2.1);

\draw (c3) to[out=-40,in=180] (1.9,1.2);
\draw (c3) to[out=-30,in=180] (1.9,1.3);

\end{tikzpicture}
\figspace
	\caption{Case (1,1)(c.i): $\Atr=\Az, \Delta(\Btr)=1, A_s\cap B_s = \{s,s'\}$. A careful reader might notice that since the excess of $B_s$ is 1, an edge count implies $r=2$.}
\label{fig:case-11-ci}
\end{figure}
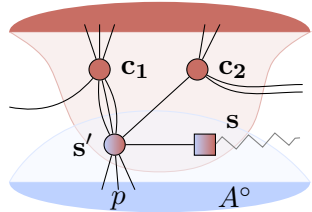

Having this structure, it is natural to make the following branching.

\begin{branching}
If $A_s\setminus \Az=\{s,s'\}$, then pursue branching on $\{s,t\}$ with fixing $ss'$.
\end{branching}

Let $(A_{ss'\to A},B_{ss'\to A})$ and $(A_{ss'\to B},B_{ss'\to B})$ be the respective branches, i.e., minimum-cost maximal terminal separations extending $(\Az\cup \{s,s'\},\Bz\cup \{t\})$ and $(\Az\cup \{t\},\Bz\cup \{s,s'\})$, respectively. Of course, if in any of these branches one more terminal pair got resolved, then we have a good branching vector. Assume therefore that this is not the case. In branch $(A_{ss'\to A},B_{ss'\to A})$ from Lemma~\ref{lem:lb} we have that $\Delta(\Az\cup \{s,s'\})\geq 1$ and $\sigma\geq 2$ Boundary Reductions are triggered within $B_s\setminus \{s\}$. In branch $(A_{ss'\to B},B_{ss'\to B})$ we again have that $\Delta(\Bz\cup \{s,s'\})\geq 1$, and $p\geq 1$ Boundary Reduction are triggered for edges between $s'$ and $\Az$.

We now calculate the obtained branching vector depending on whether $t$ is an antenna or is of type $(1,1)$.

If $t$ is an antenna with natural side $\Az$, then in branch $(A_{ss'\to A},B_{ss'\to A})$ we have $\Delta(B_{ss'\to A})\geq 1$ and one Boundary Reduction is triggered on edges incident to $t'$. Since $t'\notin B_s$, it is easy to see that the execution of the $\sigma$ previous Boundary Reductions on the side of $s$ could not spoil the applicability of this Boundary Reduction. In branch $(A_{ss'\to B},B_{ss'\to B})$, we do not account for any gain on the side of $t$. Thus we arrive at branching vector $[1,2,1+\sigma;1,1,p]$, which is $[1,2,3;1,1,1]$ or better, and hence good.

If $t$ is an antenna with natural side $\Bz$, then in branch $(A_{ss'\to A},B_{ss'\to A})$ we do not account for any gain on the side of $t$. However, in branch $(A_{ss'\to B},B_{ss'\to B})$ we have that $\Delta(A_{ss'\to B})\geq 1$ and one Boundary Reduction is triggered on edges incident to $t'$. Since $t'\neq s'$, again it is easy to see that the execution of the $p$ previous Boundary Reductions on the side of $s$ could not spoil the applicability of this Boundary Reduction. Thus we arrive at branching vector $[1,2,\sigma;1,1,1+p]$, which is $[1,2,2;1,1,2]$ or better, and hence good.

Finally, suppose $t$ is of type $(1,1)$. Then by Lemma~\ref{lem:lb} we infer that $\Delta(B_{ss'\to A})\geq 1$ and $\Delta(A_{ss'\to B})\geq 1$. Hence we have a branching vector $[1,2,\sigma;1,2,p]$ or better, which is good because $\sigma\geq 2$ and $p\geq 1$.






\paragraph*{Subcase (c.ii): $A_s\setminus\Az=\{s\}$.}


We again investigate the branch on $\{s,t\}$ with fixing $ss'$; for now we do not state that we indeed perform it, because its execution will take place only if the progress will be large enough. Let $(A_{ss'\to A},B_{ss'\to A})$ and $(A_{ss'\to B},B_{ss'\to B})$ be the respective branches, i.e., minimum-cost maximal terminal separations extending $(\Az\cup \{s,s'\},\Bz\cup \{t\})$ and $(\Az\cup \{t\},\Bz\cup \{s,s'\})$, respectively. Of course, if in any of these branches an additional terminal pair gets resolved, then we already have a good branching vector, so assume henceforth that this is not the case. Since $A_s=\Az\cup \{s\}$, by Assumption~\ref{ass:pushing} we infer that $\Delta(A_{ss'\to A})>\Delta(A_s)=1$, because in $A_{ss'\to A}$ at least one more vertex (namely $s'$) is assigned to the $A$-side. As before, in branch $(A_{ss'\to A},B_{ss'\to A})$ we have that $\sigma\geq 1$ Boundary Reductions are triggered inside $B_s$. In branch $(A_{ss'\to B},B_{ss'\to B})$, by Lemma~\ref{lem:lb} we have that $\Delta(B_{ss'\to B})\geq 1$, and we do not account for any Boundary Reductions.

\begin{figure}[H]
	\centering
	\begin{tikzpicture}

\draw[Bs] plot [smooth,tension=0.6] coordinates {(-2,2) (-1.5,1.6) (-1.1,0.7) (-0.5,0.2) (0.5,0.2) (1.2,0.7) (1.5,1.6) (2,2)} to (-2,2);
\draw[As] plot [smooth,tension=1.3] coordinates {(-2,0) (-0.1,0.7) (2,0)} to (-2,0);

\Azero;
\node at (1,-0.2) {$\Az$};
\begin{scope}[shift={(0,2)}] \Bzero; \end{scope}

\node[ABT,label=30:$\mathbf{s}$] (s) at (0.3,0.4) {};
\node[B, label=180:$\mathbf{s'}$] (sp) at (-0.3,1) {};
\node[B, label={[label distance=3]left:$\mathbf{c_1}$}] (c1) at (-1.05,1.5) {};
\node[B, label={[label distance=-3]left:$\mathbf{c_2}$}] (c2) at (-0.1,1.5) {};
\node[B, label=0:$\mathbf{c_3}$] (c3) at (0.7,1.5) {};
\node (t) at (2,0.6) {};
\draw[term] (s) to[term] (t);
\draw (s) to (sp);
\draw (sp) to[bend left=10] (c1);
\draw (sp) to[bend right=10] (c1);
\draw (sp) to (c2);
\draw (sp) to[bend left=10] (c3);
\draw (sp) to[bend right=10] (c3);

\draw (sp) to[out=-155,in=0] (-1.8,0.6);

\draw (sp) to (-0.33,-0.1);
\draw (sp) to (-0.55,-0.1);

\draw (c1) to (-1.2,2.1);
\draw (c1) to (-1.0,2.1);
\draw (c1) to[bend left] (-2,1);

\draw (c2) to (-0.1,2.1);
\draw (c2) to[out=-30,in=180] (2,1);

\draw (c3) to (0.55,2.1);
\draw (c3) to (0.77,2.1);
\draw (c3) to (0.99,2.1);

\draw (c3) to[out=-40,in=180] (1.9,1.2);
\draw (c3) to[out=-30,in=180] (1.9,1.3);

\end{tikzpicture}
\figspace
	\caption{Case (1,1)(c.ii): $\Atr=\Az, \Delta(\Btr)=1, A_s\cap B_s = \{s\}$.}
\label{fig:case-11-cii}
\end{figure}

Let us now investigate what happens in respective branches on the side of terminal $t$, depending on the type of $t$. Suppose first that $t$ is of type $(1,1)$. Then, by Lemma~\ref{lem:lb} it follows that $\Delta(B_{ss'\to A})\geq 1$ and $\Delta(A_{ss'\to B})\geq 1$, and hence together with the account of the progress on the side of $s$, we obtain a branching vector $[1,3,\sigma;1,2,0]$ or better, which is good because $\sigma\geq 1$.

We are left with the case when $t$ is an antenna, where it can be easily verified that the reasoning as above does not lead to a good branching vector without any deeper analysis. We distinguish two subsubcases, depending on the natural side of $t$.



\paragraph*{Subsubcase (c.ii.A): the natural side of $t$ is $\Az$}

Let $t'$ be the unique neighbor of $t$. Since $t$ is an antenna, there are $x$ edges from $t'$ to $\Az$ and $x$ edges from $t'$ to $V(G)$, for some $x\geq 1$.

\newcommand{\Ant}{A_{\text{nt}}}
\newcommand{\Bnt}{B_{\text{nt}}}
\newcommand{\Aunt}{A_{\text{unt}}}
\newcommand{\Bunt}{B_{\text{unt}}}
\newcommand{\Auntext}{A_{\text{unt}}^{\text{ext}}}
\newcommand{\Buntext}{B_{\text{unt}}^{\text{ext}}}

We now introduce a new type of a branching step that we shall call {\em{skewed branching}}. Namely, we will branch into separations $(\Ant,\Bnt)$ and $(\Aunt,\Bunt)$ that are minimum-cost terminal separations extending $(\Az\cup \{t\},\Bz\cup \{s\})$ and $(\Az\cup \{s,s'\},\Bz\cup \{t,t'\})$, respectively. It is easy to see that this branching step is correct, because there is always an optimum integral separation $(\Aopt,\Bopt)$ extending $(\Az,\Bz)$ where (1) $s\in \Bopt$ and $t\in \Aopt$, or (2) $\{s,s'\}\in \Aopt$ and $\{t,t'\}\in \Bopt$. Namely, if neither the first nor the second property is satisfied, then swapping the sides of $s$ and $t$ does not increase the cost of the separation (because the second property is not satisfied), but it makes the first property satisfied.

The reader should think of the skewed branching in the following way. For terminal $t$, the side $\Az$ is the natural side to be assigned to, whereas for $s$ it is $\Bz$ that is more natural. More precisely, in the branch where we have such assignment, we are not able to reason about any Boundary Reductions being triggered. We do, however, hope for a large decrease in the potential in the opposite branch, where both terminals are assigned to their unnatural sides. Therefore, in this unnatural branch we fix both edges $ss'$ and $tt'$ to maximize the progress measured in the potential function, while in the natural branch we do not fix anything, because this would not lead to any profit in the analysis.

Let us now calculate the branching vector that we obtain when we perform the described skewed branching; of course we assume that no other terminal pair gets resolved in either of the branches, because then we immediately obtain a good branching vector. In branch $(\Ant,\Bnt)$, by Lemma~\ref{lem:lb} we have that $\Delta(\Ant)\geq 0$ and $\Delta(\Bnt)\geq 1$, and we do not account for any applications of the Boundary Reduction. In branch $(\Aunt,\Bunt)$, however, we have $\Delta(\Aunt)\geq 2$ by Assumption~\ref{ass:pushing}, because $A_s=\Az\cup \{s\}$, and $\Delta(\Bunt)\geq 2$, by the definition of an antenna and the fact that $\Bz$ is the unnatural side of $t$. Moreover, in this branch $x\geq 1$ Boundary Reductions are applicable to the edges between $t'$ and $\Az$ and $\sigma\geq 1$ Boundary Reductions are applicable within $B_s$. Since $t'\notin B_s$, these applications do not interfere with each other. Thus, we arrive at a branching vector $[1,1,0;1,4,x+\sigma]$, or better. This branching vector is good unless $x=\sigma=1$. Also, even if $x=\sigma=1$ but $\Delta(\Bunt)\geq 3$, then this leads to branching vector $[1,1,0;1,5,2]$ or better, which is good. Thus, we can state the following branching step.

\begin{branching}
Unless $x=\sigma=1$ and $\Delta(\Bunt)=2$, pursue skewed branching into separations $(\Ant,\Bnt)$ and $(\Aunt,\Bunt)$.
\end{branching}





Henceforth we assume that $x=\sigma=1$ and $\Delta(\Bunt)=2$. Therefore, the degree of $t'$ in $G$ is equal to $3$, and it is adjacent to $t$, one vertex in $\Az$, and one vertex in $V(G)\setminus (\Az\cup \Bz)$ that shall be whence called $v$.

Consider now the branch $(\Aunt,\Bunt)$, and suppose there is some optimum terminal separation $(\Aopt,\Bopt)$ extending $(\Az,\Bz)$ that conforms to this branch, i.e., it also extends $(\Aunt,\Bunt)$. Suppose that $v\in \Aopt$. Then this is clearly a contradiction with the optimality of $(\Aopt,\Bopt)$, because $t'$ has $2$ neighbors in $\Aopt$ and $1$ in $\Bopt$, so moving it from $\Bopt$ to $\Aopt$ would decrease the cost of the separation. Hence we can assign $v$ greedily to the $B$-side. More precisely, instead of $(\Aunt,\Bunt)$ we will from now on consider terminal separation $(\Auntext,\Buntext)$ defined as the minimum-cost terminal separation extending $(\Az\cup \{s,s'\},\Bz\cup \{t,t',v\})$. In case $s'=v$, the reasoning above shows that the branch where $\{s,s'\}$ is assigned to the $A$-side and $\{t,t'\}$ is assigned to the $B$-side cannot lead to an optimum solution, so we can greedily pursue the branch where $s$ and $t$ are assigned to respective natural sides.

\begin{reductionstep}
If $v=s'$, then recurse into terminal separation $(A_t,B_s)$.
\end{reductionstep}

Hence, from now on we assume that $v\neq s'$ and we branch into $(\Ant,\Bnt)$ and $(\Auntext,\Buntext)$, where the latter is defined as above. As usual, we assume that $(\Auntext,\Buntext)$ does not resolve any new terminal pair, because then we would have a good branching vector. The same reasoning as for $(\Aunt,\Bunt)$ shows that $\Delta(\Auntext)\geq 2$ and $\Delta(\Buntext)\geq 2$. As before, if we had that $\Delta(\Buntext)\geq 3$, then branching into $(\Ant,\Bnt)$ and $(\Auntext,\Buntext)$ would lead to a branching vector $[1,1,0;1,5,2]$ or better. Hence, we can again assume that $\Delta(\Buntext)=2$.

Therefore, a straightforward edge count shows that $B_q=\Buntext\setminus \{t,t'\}$ is a terminal-free $\Bz$-extension of excess $2$. Hence, we can apply Lemma~\ref{lem:ex2-red} to decompose it. By the inapplicability of the Excess-2 Reduction, we have that $B_q\setminus \Bz$ has a decomposition of the form $\{c_1,c_2\}$ or $\{d,c_1,\ldots,c_r\}$ (from now on we drop the earlier notation for the decomposition of $B_s\setminus (\Bz\cup \{s\})$, and use the notation $d,c_1,\ldots,c_r$ for the decomposition of $B_q\setminus \Bz$). Suppose first that $v=c_i$ for some $i$. Then this is a contradiction with the optimality of $(\Auntext,\Buntext)$, because then $\{t,t',c_i\}$ would be a set of excess $1$, so replacing $\Buntext$ with it would decrease the cost of separation $(\Auntext,\Buntext)$. Therefore, $B_q$ has a decomposition of the form $\{d,c_1,\ldots,c_r\}$ where $v=d$. 

By Lemma~\ref{lem:ex2-red}, we have that $|E(v,c_i)|=p_i$, $|E(c_i,\Bz)|=p_i+x_i$ and $|E(c_i,V(G)\setminus B_q)|=x_i+1$, for some integers $p_i\geq 1$ and $x_i\geq 0$. Let $\sigma_2=|E(v,\Bz)|+\sum_{i=1}^r p_i$ be the number of Boundary Reductions triggered within $B_q$ when the vertex $v$ is assigned to the $A$-side. By Lemma~\ref{lem:ex2-Bside}, $\sigma_2\geq 1$.

Before we proceed, we need to resolve a corner case when $s'\in B_q$. We claim that then it is safe to greedily assign $t$ to the $A$-side and $s$ to the $B$-side.

\begin{reductionstep}\label{rdstep:corner2}
If $s'\in B_q$, then recurse with terminal separation $(A_t,B_s)$.
\end{reductionstep}

To argue the correctness of this reduction step, we need to prove that there exists an optimum terminal separation extending $(\Az,\Bz)$ where $s$ is assigned to the $B$-side and $t$ is assigned to the $A$-side. Let us take any optimum terminal separation $(\Aopt,\Bopt)$
that satisfies point 1 of Lemma~\ref{lem:ex2}.
Assume $s\in \Aopt$ and $t\in \Bopt$, as otherwise we are done.
We can further assume that $s'\in \Aopt$ and $t'\in \Bopt$, because otherwise switching the sides of $s$ and $t$ would not increase the cost of the separation, however it would make it satisfy the condition we seek.
Suppose first that $s'=v$; then we have an immediate contradiction, because moving $t'$ from $\Aopt$ to $\Bopt$ would decrease the cost.
Suppose then that $s'=c_i$ for some $i=\{1,2,\ldots,r\}$.
Since $(\Aopt,\Bopt)$ satisfies point 1 of Lemma~\ref{lem:ex2},
we infer that $\Bopt\cap B_q=\Bz$ or $\Bopt\cap B_q=\Bz\cup \{c_j\}$ for some $j\neq i$. In particular, $v\in \Aopt$. This is, however, a contradiction, because moving $t'$ from $\Bopt$ to $\Aopt$ would decrease the cost of the separation. This justifies the correctness of Reduction Step~\ref{rdstep:corner2}.

Whence we assume that $s'\notin B_q$.

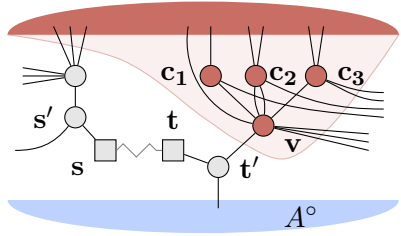
\begin{figure}[H]
	\centering
	\begin{tikzpicture}[yscale=1.1]

\draw[Bs] plot [smooth,tension=0.6] coordinates {(-2.6,2) (-1.1,1.7) (-0,1) (0.5,0.7) (1.1,0.5) (1.7,0.9) (2.3,1.6) (2.6,2)} to (-2.6,2);

\begin{scope}[xscale=1.3]
	\Azero;
	\node at (1,-0.2) {$\Az$};
	\begin{scope}[shift={(0,2)}] \Bzero; \end{scope}
\end{scope}

\node[T,label=-170:$\mathbf{s}$] (s) at (-1.3,0.6) {};
\node[V, label=left:$\mathbf{s'}$] (sp) at (-1.7,1) {};
\node[V] (oldc1) at (-1.7,1.5) {};
\node[T, label=above:$\mathbf{t}$] (t) at (-0.4,0.6) {};
\node[V, label=0:$\mathbf{t'}$] (tp) at (0.2,0.4) {};
\node[B, label=-20:$\mathbf{v}$] (v) at (0.8,0.9) {};
\node[B, label=left:$\mathbf{c_1}$] (c1) at (0.1,1.5) {};
\node[B, label={[label distance=-3]right:$\mathbf{c_2}$}] (c2) at (0.7,1.5) {};
\node[B, label=right:$\mathbf{c_3}$] (c3) at (1.5,1.5) {};

\draw (s) to (sp);
\draw (sp) to (oldc1);
\draw (sp) to[out=-130,in=0] (-2.5,0.6);

\draw (oldc1) to (-1.55,2.1);
\draw (oldc1) to (-1.77,2.1);
\draw (oldc1) to (-1.99,2.1);

\draw (oldc1) to (-2.4,1.6);
\draw (oldc1) to (-2.4,1.5);
\draw (oldc1) to (-2.4,1.4);

\draw[term] (s) to (t);
\draw (t) to (tp);
\draw (tp) to (0.2,-0.1);
\draw (tp) to (v);

\draw (v) to (c1);
\draw (v) to[bend left=15] (c2);
\draw (v) to[bend right=15] (c2);
\draw (v) to (c3);
\draw (v) to[out=170,in=-100] (-0.2,2.1);

\draw (v) to (2.2,0.8);
\draw (v) to (2.2,0.7);
\draw (v) to (2.2,0.6);

\draw (c1) to (0.1,2.1);
\draw (c2) to (0.6,2.1);
\draw (c2) to (0.8,2.1);
\draw (c3) to (1.4,2.1);
\draw (c3) to (1.6,2.1);

\draw (c1) to[out=-25,in=175] (2.4, 1);
\draw (c2) to[out=-25,in=180] (2.4, 1.1);
\draw (c3) to[out=-25,in=180] (2.4, 1.2);
\draw (c3) to[out=-20,in=180] (2.4, 1.3);

\end{tikzpicture}
\figspace
	\caption{Case (1,1)(c.ii.A), where furthermore $x=\sigma=1$ and $\Delta(\Bunt)=2$. The set $B_q=\Buntext \setminus\{t,t'\}$ of excess 2 is highlighted in red.}
\label{fig:case-11-ciiA}
\end{figure}

We will now branch on vertex $v$. More precisely, we recurse into branches $(A_{v\to A},B_{v\to A})$ and $(A_{v\to B},B_{v\to B})$, defined as minimum-cost terminal separations extending $(\Az\cup \{v\},\Bz)$ and $(\Az,\Bz\cup \{v\})$, respectively.

\begin{branching}
Pursue branching on $v$, that is, recurse into branches $(A_{v\to A},B_{v\to A})$ and $(A_{v\to B},B_{v\to B})$.
\end{branching}

The remainder of the description of this subcase is devoted to proving that the execution of this branching step leads to a good branching vector.

Consider first branch $(A_{v\to A},B_{v\to A})$. By the optimality of $(A_{v\to A},B_{v\to A})$ we infer that $t'\in A_{v\to A}$, because otherwise assigning it to $A_{v\to A}$, or moving from $B_{v\to A}$ to $A_{v\to A}$, would decrease the cost of the separation. Also, we can assume that $t\in A_{v\to A}$ and $s\in B_{v\to A}$ for the following reason. If this terminal pair was not resolved in $(A_{v\to A},B_{v\to A})$, then assigning $t$ to the $A$-side and $s$ to the $B$-side would not increase the cost of the separation while extending $(A_{v\to A},B_{v\to A})$, a contradiction with the maximality of $(A_{v\to A},B_{v\to A})$. However, if $t\in B_{v\to A}$ and $s\in A_{v\to A}$, then we can modify separation $(A_{v\to A},B_{v\to A})$ by switching the sides of $s$ and $t$, and because $t'\in A_{v\to A}$, then this modification does not increase the cost. 

Hence, in branch $(A_{v\to A},B_{v\to A})$ the terminal pair $\{s,t\}$ gets resolved. Since $\{v,t,t'\}\subseteq A_{v\to A}$ and $\{s\}\subseteq B_{v\to A}$, by Lemma~\ref{lem:lb} and the definition of an antenna we obtain that $\Delta(A_{v\to A})\geq 1$ and $\Delta(B_{v\to A})\geq 1$. Notice also that at least $\sigma_2\geq 1$ Boundary Reductions are triggered within $B_q$.

Consider the second branch $(A_{v\to B},B_{v\to B})$. In it, one Boundary Reduction is triggered on the edges incident to $t'$, and the terminal pair $\{s,t\}$ is either resolved by this branch or is immediately removed by the Lonely Terminal Reduction (possibly preceded by the Pendant Reduction that removes $t'$). Hence, $\{s,t\}$ also gets resolved in this branch.

From now on, we assume that neither of the considered branches resolves any terminal pair other than $\{s,t\}$, because then, as argued at the beginning of this section, we would immediately achieve a good branching vector. With this assumption in mind, we now claim that $\Delta(B_{v\to B})\geq 2$. 

\begin{myclaim}\label{cl:even-if-terminal1}
$\Delta(B_{v\to B})\geq 2$.
\end{myclaim}
\begin{proof}
If $B_{v\to B}$ is a terminal-free extension of $\Bz$, then this follows from Lemma~\ref{lem:ex2-Aside}. We have two cases left to investigate: either $t\in B_{v\to B}$ or $s\in B_{v\to B}$.

In the first case, since $v\in B_{v\to B}$ by the optimality of $(A_{v\to B},B_{v\to B})$ it follows that also $t'\in B_{v\to B}$. But then if $B_{v\to B}$ was an extension of excess at most $1$, then $B_{v\to B}\setminus \{t,t'\}$ would be a terminal-free extension of excess at most $1$, a contradiction with Lemma~\ref{lem:ex2-Aside}.

In the second case, assume for the sake of contradiction that $\Delta(B_{v\to B})=1$ (it cannot happen that $\Delta(B_{v\to B})=0$, because then $(\Az\cup \{t,t'\},B_{v\to B})$ would be an extension of $(\Az,\Bz)$ of the same cost, a contradiction with the maximality of $(\Az,\Bz)$). Then $B_{v\to B}\setminus \{s\}$ is a terminal-free set of excess $2$. Since $v\in B_{v\to B}$, by the optimality of $(A_{v\to B},B_{v\to B})$ we obtain that $c_i\in B_{v\to B}$ for each $i\in \{1,2,\ldots,r\}$, so $B_q\subseteq B_{v\to B}\setminus \{s\}$. 

On the other hand, $s'\in B_{v\to B}$ because $G[B_{v\to B}\setminus \Bz]$ is connected by the same reasoning as in Lemma~\ref{lem:AsConn}. But we are currently working with the assumption that $s'\notin B_q$, so $B_q$ is a strict subset of $B_{v\to B}\setminus \{s\}$. 

Thus, we obtain a contradiction with Lemma~\ref{lem:ex2-nested}. Indeed, from this lemma it follows that $B_q$ consists of two vertices that are adjacent to $\Bz$ and each of them forms an excess-1 extension of $\Bz$, but we know that $v\in B_q$ is the vertex $d$ of the decomposition of $B_q$ and by Lemma~\ref{lem:ex2-red}, $\Delta(\Bz\cup \{v\})>1$.
\end{proof}

We conclude that the considered branching leads to branching vector $[1,2,\sigma_2;1,2,1]$, or better. This vector is good unless $\sigma_2=1$. Also, if in fact $\Delta(A_{v\to A})\geq 3$, then we also arrive at a good branching vector $[1,3,1;1,2,1]$, or better. Hence, from now on assume that $\sigma_2=1$ and $\Delta(A_{v\to A})=2$.

If $\Delta(A_{v\to A})=2$, then $A_{v\to A}\setminus \{t,t'\}$ is a terminal-free extension of $\Az$ of excess $2$. Let us apply Lemma~\ref{lem:ex2-red} to it. Regardless of the form of the decomposition, from Lemma~\ref{lem:ex2-Bside} we infer that in the branch $(A_{v\to B},B_{v\to B})$ at least one Boundary Reduction will be triggered within $A_{v\to A}\setminus \{t,t'\}$. This Boundary Reduction is applied independently of the Boundary Reduction triggered on edges incident to $t'$ that we previously counted in branch $(A_{v\to B},B_{v\to B})$. This gives one additional Boundary Reduction that we did not account for previously, which leads to a good branching vector $[1,2,1;1,2,2]$, or better.





\paragraph*{Subsubcase (c.ii.B): the natural side of $t$ is $\Bz$}

Recall that we investigated the branch $(A_{ss'\to A},B_{ss'\to A})$ and $(A_{ss'\to B},B_{ss'\to B})$, and we concluded that $\Delta(A_{ss'\to A})\geq 2$, $\Delta(B_{ss'\to B})\geq 1$, and $\sigma\geq 1$ Boundary Reductions are triggered inside $B_s$ in branch $(A_{ss'\to A},B_{ss'\to A})$. From the definition of an antenna we have that $\Delta(A_{ss'\to B})\geq 1$ and one Boundary Reduction is triggered on edges incident to $t'$ in branch $(A_{ss'\to B},B_{ss'\to B})$, when $t$ is assigned to the $A$-side. This gives us branching vector $[1,2,\sigma;1,2,1]$ or better, which is good unless $\sigma=1$. 
Also, note that this branching is good if $s'$ is adjacent to a second terminal different than $s$: due to inapplicability
of the Common Neighbor Reduction, this terminal would belong to a second terminal pair that would get resolved
in the branching.
This justifies the execution of the following branching step.

\begin{branching}
If $\sigma > 1$ or $s'$ is adjacent to a second terminal different than $s$,
then recurse into branches $(A_{ss'\to A},B_{ss'\to A})$ and $(A_{ss'\to B},B_{ss'\to B})$.
\end{branching}


Recall that from Lemma~\ref{lem:cool} we obtained a decomposition $\{d,c_1,c_2,\ldots,c_r\}$ of $(B_{s}\setminus \{s\})\setminus \Bz$, where $d=s'$ is the unique neighbor of $s$. By Lemma~\ref{lem:ex2-Bside}, if $\sigma=1$ then we have two cases:
\begin{itemize}
\item either $r=0$ and $s'$ has degree $4$: one edge to $\Bz$, one edge to $s$, and two edges to $V(G)\setminus (\Bz\cup \{s\})$;
\item or $r=1$ and $s'$ has degree $3$: one edge to $c_1$, one edge to $s$, and one edge to $V(G)\setminus (\Bz\cup \{s,c_1\})$. Moreover, $c_1$ is incident exactly on the following edges: $1$ edge to $s'$, $x$ edges to $\Bz$ and $x$ edges to $V(G)\setminus \Bz\cup (\{s,s'\})$, for some $x\geq 1$.
\end{itemize}
Let $y=|E(t',\Bz)|$; then $y\geq 1$. We now investigate the cases separately.



\paragraph*{Subsubsubcase (c.ii.B.1): $r=0$}

We first investigate the possibility of branching on $\{s,t\}$ with fixing $tt'$. That is, we examine branches $(A_{tt'\to B},B_{tt'\to B})$ and $(A_{tt'\to A},B_{tt'\to A})$ that are minimum-cost maximal extensions of $(\Az\cup \{s\},\Bz\cup \{t,t'\})$ and $(\Az\cup \{t,t'\},\Bz\cup \{s\})$, respectively. Of course if any of these branches resolves some terminal pair other than $\{s,t\}$, then we obtain a good branching vector. Hence, assume this is not the case. 

Consider the first branch $(A_{tt'\to B},B_{tt'\to B})$. By Lemma~\ref{lem:lb} we have that $\Delta(A_{tt'\to B})\geq 1$. Also, one Boundary Reduction is triggered on edges incident to $s'$. 

Consider the second branch $(A_{tt'\to A},B_{tt'\to A})$. By the definition of an antenna, we have that $\Delta(A_{tt'\to A})\geq 2$ and by Lemma~\ref{lem:lb} we have that $\Delta(B_{tt'\to A})\geq 1$. Also, $y$ Boundary Reductions are triggered on edges between $t'$ and $\Bz$ in this branch.

This leads to a branching vector $[1,1,1;1,3,y]$ or better, which is good unless $y=1$.
This justifies executing the following step.

\begin{branching}
If $y>1$, then recurse into branches $(A_{tt'\to B},B_{tt'\to B})$ and $(A_{tt'\to A},B_{tt'\to A})$.
\end{branching}

From now on we assume that $y=1$, and, consequently, we have that $t'$ has degree $3$: it neighbors $t$, one vertex in $\Bz$, and one vertex $w$ outside $\Bz\cup \{t\}$. 

We now resolve the corner case when $w=s'$. Then $t'$ is adjacent to $t$, $s'$ and one vertex in $\Bz$, whereas $s'$ is adjacent to $s$, $t'$, one vertex in $\Bz$, and one vertex outside $\Bz\cup \{s,s',t,t'\}$. We claim that we can assign greedily $s$ to the $A$-side and $t$ to the $B$-side. To argue this, take any minimum-cost integral terminal separation $(\Aopt,\Bopt)$ extending $(\Az,\Bz)$, and assume that $s\in \Bopt$ and $t\in \Aopt$. We can further assume that $s'\in \Bopt$ and $t'\in \Aopt$, because otherwise switching the sides of $s$ and $t$ produces an integral terminal separation of no larger cost where $s$ and $t$ are on the sides we aimed for. But then $t'$ has two neighbors in $\Bopt$ and one in $\Aopt$, which is a contradiction with the optimality of $(\Aopt,\Bopt)$ --- moving $t'$ from $\Aopt$ to $\Bopt$ would decrease the cost. This justifies the correctness of the following step.

\begin{figure}[H]
	\centering
	\begin{tikzpicture}[yscale=0.8,xscale=0.9]


\begin{scope}
	\Azero;
	\node at (1,-0.2) {$\Az$};
	\begin{scope}[shift={(0,2)}] \Bzero; \end{scope}
\end{scope}

\node[T,label=-170:$\mathbf{s}$] (s) at (-0.7,0.6) {};
\node[V, label=left:$\mathbf{w=s'}$] (sp) at (-1.2,1.4) {};
\node[T, label=-0:$\mathbf{t}$] (t) at (0.7,0.6) {};
\node[V, label=0:$\mathbf{t'}$] (tp) at (1.2,1.4) {};

\draw (s) to (sp);
\draw (sp) to (-1.2,2.1);
\draw (sp) to (-2,0.8);

\draw[term] (s) to (t);
\draw (t) to (tp);
\draw (tp) to (1.2,2.1);
\draw (tp) to (sp);

\end{tikzpicture}
\figspace
	\caption{Case (1,1)(c.ii.B.1), where furthermore $s'$ is a neighbor of $t'$.}
\label{fig:case-11-ciiB1corner}
\end{figure}
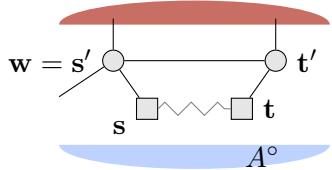

\begin{reductionstep}
If $w=s'$, then recurse into branch $(A_s,B_t)$.
\end{reductionstep}

From now on, we assume that $w\neq s'$.

Suppose now that $|E(w,\Az)|>0$. Then in branch $(A_{tt'\to B},B_{tt'\to B})$ one additional Boundary Reduction is triggered on edges incident to $w$. Since $w\neq s'$, the application of this Boundary Reduction cannot spoil the applicability of the Boundary Reduction on edges incident to $s'$ that we counted in the same branch. This results in branching vector $[1,1,2;1,3,1]$ or better, which is good, so we can do the following.

\begin{branching}
If $|E(w,\Az)|>0$, then recurse into branches $(A_{tt'\to B},B_{tt'\to B})$ and $(A_{tt'\to A},B_{tt'\to A})$.
\end{branching}

Henceforth we assume that $E(w,\Az)=\emptyset$.

\newcommand{\Aext}{A^{\textrm{ext}}_{tt'\to A}}
\newcommand{\Bext}{B^{\textrm{ext}}_{tt'\to A}}

As in Case~(c.ii.A), we argue that in branch $(A_{tt'\to A},B_{tt'\to A})$ we can greedily assign $w$ to the $A$-side. More precisely, instead of $(A_{tt'\to A},B_{tt'\to A})$ we will from now on consider terminal separation $(\Aext,\Bext)$ defined as the minimum-cost terminal separation extending $(\Az\cup \{t,t',w\},\Bz\cup \{s\})$. The argumentation for the correctness of this step is as before: Suppose there is some optimum terminal separation $(\Aopt,\Bopt)$ extending $(\Az,\Bz)$ that conforms to this branch, i.e., it also extends $(A_{tt'\to A},B_{tt'\to A})$. Suppose that $w\in \Bopt$. Then this is clearly a contradiction with the optimality of $(\Aopt,\Bopt)$, because $t'$ has $2$ neighbors in $\Bopt$ and $1$ in $\Aopt$, so moving it from $\Aopt$ to $\Bopt$ would decrease the cost of the separation. 

As usual, $(\Aext,\Bext)$ resolving an additional terminal pair would immediately lead to a good branching vector when branching into $(A_{tt'\to B},B_{tt'\to B})$ and $(\Aext,\Bext)$, so assume this is not the case.

In branch $(\Aext,\Bext)$, by the definition of an antenna, we have that $\Delta(\Aext)\geq 2$ and by Lemma~\ref{lem:lb} we have that $\Delta(\Bext)\geq 1$. Also, a Boundary Reduction is triggered on the edge between $t'$ and $\Bz$ in this branch. Observe that if we in fact had that $\Delta(\Aext)\geq 3$, then branching into $(A_{tt'\to B},B_{tt'\to B})$ and $(\Aext,\Bext)$ results in branching vector $[1,1,1;1,4,1]$ or better, which is good. This justifies the execution of the following.

\begin{branching}
If $\Delta(\Aext)\geq 3$, then recurse into branches $(A_{tt'\to B},B_{tt'\to B})$ and $(\Aext,\Bext)$.
\end{branching}

\begin{figure}[H]
	\centering
	\begin{tikzpicture}[yscale=-1.1]

\draw[As] plot [smooth,tension=0.5] coordinates {(-2.6,2) (-1.1,1.7) (-0.2,1) (0.3,0.7) (1.1,0.5) (1.6,0.7) (2.3,1.6) (2.6,2)} to (-2.6,2);

\begin{scope}[xscale=1.3,yscale=-1,shift={(0,-2)}]
	\Azero;
	\begin{scope}[shift={(0,2)}] \Bzero; \node at (1,0.2) {$\Bz$}; \end{scope}
\end{scope}

\node[T,label=-100:$\mathbf{s}$] (s) at (-1.3,0.6) {};
\node[V, label=left:$\mathbf{s'}$] (sp) at (-1.7,0.4) {};
\node[T, label=above:$\mathbf{t}$] (t) at (-0.4,0.6) {};
\node[V, label=0:$\mathbf{t'}$] (tp) at (0.2,0.4) {};
\node[A, label=left:$\mathbf{w}$] (v) at (0.8,0.9) {};
\node[A, label=left:$\mathbf{c_1}$] (c1) at (0.1,1.5) {};
\node[A, label={[label distance=-3]-10:$\mathbf{c_2}$}] (c2) at (0.7,1.5) {};
\node[A, label=-10:$\mathbf{c_3}$] (c3) at (1.5,1.5) {};

\draw (s) to (sp);
\draw (sp) to (-1.7,-0.1);
\draw (sp) to (-2.4,0.8);
\draw (sp) to (-2.3,1.1);

\draw[term] (s) to (t);
\draw (t) to (tp);
\draw (tp) to (0.2,-0.1);
\draw (tp) to (v);

\draw (v) to (c1);
\draw (v) to[bend left=15] (c2);
\draw (v) to[bend right=15] (c2);
\draw (v) to (c3);

\draw (v) to (2.2,0.5);
\draw (v) to (2.2,0.7);

\draw (c1) to (0.1,2.1);
\draw (c2) to (0.6,2.1);
\draw (c2) to (0.8,2.1);
\draw (c3) to (1.4,2.1);
\draw (c3) to (1.6,2.1);

\draw (c1) to[out=-25,in=175] (2.4, 1);
\draw (c2) to[out=-25,in=180] (2.4, 1.1);
\draw (c3) to[out=-25,in=180] (2.4, 1.2);
\draw (c3) to[out=-20,in=180] (2.4, 1.3);

\end{tikzpicture}
\figspace
	\caption{Case (1,1)(c.ii.B.1), where furthermore $w\neq s'$ and $\Delta(\Aext)=2$. The set $A_q=\Aext\setminus \{t,t'\}$ of excess 2 is highlighted in blue.}
\label{fig:case-11-ciiB1}
\end{figure}
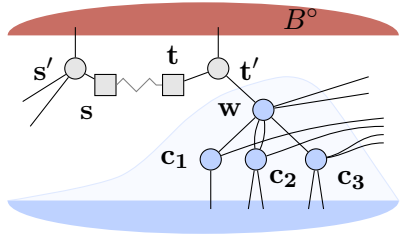

From now on we assume that $\Delta(\Aext)=2$. This means that $A_q=\Aext\setminus \{t,t'\}$ is a terminal-free extension of $\Az$ of excess $2$, and moreover $w\in A_q$. Hence, we can apply Lemma~\ref{lem:ex2-red} to $A_q$ and obtain a decomposition of the form $\{c_1,c_2,\}$ or $\{d,c_1,\ldots,c_r\}$ (we now drop the notation for the decomposition of $B_s$, and use it for the decomposition of $A_q$ instead). By Lemma~\ref{lem:ex2-red}, each $c_i$ is connected with $\Az$ via at least one edge, but we assumed that there is no edge between $w$ and $\Az$. Hence the decomposition has the form $\{d,c_1,\ldots,c_r\}$ and $d=w$. By Lemma~\ref{lem:ex2-Bside}, at least one Boundary Reduction is triggered within $A_q$ in any branch where $w$ is assigned to $B$.

We will pursue branching on vertex $w$.  More precisely, we consider recursing into branches $(A_{w\to A},B_{w\to A})$ and $(A_{w\to B},B_{w\to B})$, defined as minimum-cost terminal separations extending $(\Az\cup \{w\},\Bz)$ and $(\Az,\Bz\cup \{w\})$, respectively.

\begin{branching}
Pursue branching on $w$, that is, recurse into branches $(A_{w\to A},B_{w\to A})$ and $(A_{w\to B},B_{w\to B})$.
\end{branching}

The remainder of this case is devoted to arguing that this branching step leads to a good branching vector.

Consider branch $(A_{w\to A},B_{w\to A})$. Then a Boundary Reduction is triggered on edges incident to $t'$, and consequently the pair $\{s,t\}$ either gets resolved in this branch, or is removed by an application of the Lonely Terminal Reduction (possibly preceded by the Pendant Reduction that removes $t'$). Hence, at least the terminal pair $\{s,t\}$ gets resolved or removed in this branch. 

In branch $(A_{w\to B},B_{w\to B})$ we have that $t'\in B_{w\to B}$ due to the optimality of $(A_{w\to B},B_{w\to B})$. Moreover, we can assume that $t\in B_{w\to B}$ and $s\in A_{w\to B}$ for the following reason. If this terminal pair was not resolved in $(A_{w\to B},B_{w\to B})$, then assigning $t$ to the $B$-side and $s$ to the $A$-side would not increase the cost of the separation while extending $(A_{w\to B},B_{w\to B})$, a contradiction with the maximality of $(A_{w\to B},B_{w\to B})$. However, if $t\in A_{w\to B}$ and $s\in B_{w\to B}$, then we can modify separation $(A_{w\to B},B_{w\to B})$ by switching the sides of $s$ and $t$, and because $t'\in B_{w\to B}$, then this modification does not increase the cost. 

Hence, in both branches, the terminal pair $\{s,t\}$ gets eventually removed. Suppose then that neither $(A_{w\to A},B_{w\to A})$ nor $(A_{w\to B},B_{w\to B})$ resolves another terminal pair, because otherwise we would immediately have a good branching vector. We continue the calculation of the obtained branching vector with this assumption. First, we need an analogue of Claim~\ref{cl:even-if-terminal1}, whose proof is very similar.

\begin{myclaim}\label{cl:even-if-terminal2}
$\Delta(A_{w\to A})\geq 2$.
\end{myclaim}
\begin{proof}
If $A_{w\to A}$ is a terminal-free extension of $\Az$, then this follows from Lemma~\ref{lem:ex2-Aside}. We have two cases left to investigate: either $t\in A_{w\to A}$ or $s\in A_{w\to A}$.

In the first case, since $w\in A_{w\to A}$ by the optimality of $(A_{w\to A},B_{w\to A})$ it follows that also $t'\in A_{w\to A}$. But then if $A_{w\to A}$ was an extension of excess at most $1$, then $A_{w\to A}\setminus \{t,t'\}$ would be a terminal-free extension of excess at most $1$, a contradiction with Lemma~\ref{lem:ex2-Aside}.

In the second case, assume for the sake of contradiction that $\Delta(A_{w\to A})=1$ (it cannot happen that $\Delta(A_{w\to A})=0$, because then $(A_{w\to A},\Bz\cup \{t,t'\},)$ would be an extension of $(\Az,\Bz)$ of the same cost, a contradiction with the maximality of $(\Az,\Bz)$). Then $A_{w\to A}\setminus \{s\}$ is a terminal-free set of excess $2$ that contains $w$. Since $w\in A_{w\to A}$, by the optimality of $(A_{w\to A},B_{w\to A})$ we obtain that $c_i\in A_{w\to A}$ for each $i\in \{1,2,\ldots,r\}$, so $A_q\subseteq A_{w\to A}\setminus \{s\}$. 

On the other hand, $s'\in A_{w\to A}$ because $G[A_{w\to A}\setminus \Az]$ is connected by the same reasoning as in the proof of Lemma~\ref{lem:AsConn}. However, observe that $s'\notin A_q$. Indeed, we are working with the assumption that $s'\neq w$, and moreover $s'\neq c_i$ for each $i=1,2,\ldots,r$ because otherwise the Boundary Reduction would apply to $s'$. Hence $A_q$ is a strict subset of $A_{w\to A}\setminus \{s\}$. 

Thus, we obtain a contradiction with Lemma~\ref{lem:ex2-nested}. Indeed, from this lemma it follows that $A_q$ consists of two vertices that are adjacent to $\Az$, but we know that $w\in A_q$ and $E(w,\Az)=\emptyset$.
\end{proof}

Observe that in branch $(A_{w\to A},B_{w\to A})$ we also have one Boundary Reduction triggered on the edges incident to $t'$.

On the other hand, in branch $(A_{w\to B},B_{w\to B})$ we have $\Delta(A_{w\to B})\geq 1$ by Lemma~\ref{lem:lb} because $s\in A_{w\to B}$. Also, $\Delta(B_{w\to B})\geq 1$ by the definition of an antenna and the fact that $\{w,t,t'\}\subseteq B_{w\to B}$. Finally, one Boundary Reduction is triggered on edges incident to $s'$ and one Boundary Reduction is triggered within $A_q$. Since $s'\neq w$, the application of one of these reductions cannot spoil the applicability of the other.

Thus we obtain branching vector $[1,2,1;1,2,2]$ or better, which is a good branching vector.





\paragraph*{Subsubsubcase (c.ii.B.2): $r=1$.}

Recall that $s'$ has degree $3$ and has one edge to $c_1$, one edge to $s$, and one edge to $V(G)\setminus (\Bz\cup \{s,c_1\})$, whereas $c_1$ has $x$ edges to $\Bz$ and $x$ edges to $V(G)\setminus (\Bz\cup \{s,c_1\})$, for some $x\geq 1$. Let $v$ be the neighbor of $s'$ other than $c_1$ and $s$.

First, note that $v$ is not a terminal, as we have already excluded the case 
when $s'$ is adjacent to a second terminal different than $s$.

We claim that there is an optimum integral terminal separation $(\Aopt,\Bopt)$ extending $(\Az,\Bz)$ where vertices $s'$ and $v$ are assigned to the same side. Take any optimum integral separation $(\Aopt,\Bopt)$ and suppose that $s'$ and $v$ are assigned to different sides.

Assume first that $s'\in \Aopt$ and $v\in \Bopt$. Then $s\in \Aopt$ and $c_1\in \Aopt$ because otherwise moving $s'$ from $\Aopt$ to $\Bopt$ would decrease the cost of the separation. Hence $t\in \Bopt$, and again by the optimality of $(\Aopt,\Bopt)$ we have $t'\in \Bopt$. Consider a modified integral terminal separation $(\Aopt_m,\Bopt_m)$ obtained from $(\Aopt,\Bopt)$ by moving $\{c_1,s'\}$ from the $A$-side to the $B$-side. Recall that $c_1$ has $x$ edges to $\Bz$ and $x$ edges going outside of $\Bz\cup \{c_1,s'\}$. Then it is easy to see that the cost of $(\Aopt_m,\Bopt_m)$ is not larger than the cost of $(\Aopt,\Bopt)$, whereas $s'$ and $v$ are both assigned to the $B$-side.

Assume second that $s'\in \Bopt$ and $v\in \Aopt$. Similarly as before, by the optimality of $(\Aopt,\Bopt)$ we have that $s\in \Bopt$ and $c_1\in \Bopt$. Consequently, $t\in \Aopt$. If we had that $t'\in \Bopt$, then moving $t$ from $\Aopt$ to $\Bopt$ and $\{s,s'\}$ from $\Bopt$ to $\Aopt$ would strictly decrease the cost of the separation, a contradiction with the optimality of $(\Aopt,\Bopt)$. Hence we have that $t'\in \Aopt$. Consider a modified integral terminal separation $(\Aopt_m,\Bopt_m)$ obtained from $(\Aopt,\Bopt)$ by moving $\{s,s'\}$ from the $B$-side to the $A$-side, and moving $\{t,t'\}$ from the $A$-side to the $B$-side. Recall that $t'$ has $y$ edges going to $\Bz$ and $y$ edges going outside of $\Bz\cup \{t,t'\}$. Hence it is easy to see that the cost of $(\Aopt_m,\Bopt_m)$ is not larger than the cost of $(\Aopt,\Bopt)$, whereas $s'$ and $v$ are both assigned to the $A$-side.

This justifies the execution of the following reduction in this case.

\begin{reductionstep}
Merge $s'$ and $v$.
\end{reductionstep}

As the case study is exhaustive, this finishes the description of the branching algorithm. We hope that the reader shares the joy of the writer after getting to this line.



\section{Conclusions}\label{sec:conc}
In this work we have developed an algorithm for \edgebip that has running time $\Oh(1.977^k\cdot nm)$, which is the first one to achieve a dependence on the parameter better than $2^k$. Our result shows that in the case of \edgebip the constant $2$ in the base of the exponent is not the ultimate answer, as is conjectured for {\sc{CNF-SAT}}. Also, it improves some recent works where the FPT algorithm for \edgebip is used as a black-box~\cite{KolayP0S16}. However, our work leaves a number of open questions that we would like to highlight.

\begin{itemize}
\item Reducing the dependence on the parameter from $2^k$ to $1.977^k$ can be only considered a ``proof of concept'' that such an improvement is possible. Even though we believe that it is an important step in understanding the optimal parameterized complexity of graph separation problems, we put forward the question of designing a reasonably simple algorithm with the running time dependence on the parameter substantially better than $2^k$. Last but not least, it is also interesting whether the dependence of the running time on the input size can be improved from $\Oh(nm)$ to linear.
\item There is a simple reduction that reduces back \termsep to \edgebip
without changing the size of the cutset.\footnote{We thank an anonymous reviewer for suggesting this reduction.} Given a \termsep instance:
$(G,\terms,(\Az,\Bz),k)$:
\begin{enumerate}
\item subdivide once every edge of $G$;
\item add a new terminal pair $(a_0,b_0)$, connect $a_0$ to every vertex of $\Az$
with $k+1$ parallel edges, and connect $b_0$ to every vertex of $\Bz$ with $k+1$ parallel edges;
\item for every terminal pair $(s,t)$, including $(a_0,b_0)$, connect $s$ and $t$
with $k+1$ parallel edges.
\end{enumerate}
Recall that our approach can be summarized as follows: having observed that \termsep admits a simple $\Ohstar(2^{|\terms|})$-time algorithm and an $\Ohstar(4^k)$-time algorithm using the CSP-guided technique of Wahlstr\"om~\cite{magnus}, we develop an algorithm for a joint parameterization $(|\terms|,k)$ that for $|\terms|=k+1$ achieves running time $\Ohstar(1.977^k)$. 
The aforementioned reduction yields an $\Ohstar(1.977^k)$-time algorithm for \termsep. The remaining question is: can \termsep be solved in time $\Ohstar(c^{|\terms|})$ for some $c<2$? Maybe one can prove matching lower bounds under the Strong Exponential Time Hypothesis (SETH), or under the Set Cover Conjecture (SeCoCo)~\cite{secoco}?
\item Finally, we would like to reiterate two related open questions. 
\begin{itemize}
\item The currently fastest algorithm for \oct runs in time $\Ohstar(2.3146^k)$~\cite{LP-guided-dl}. It is reasonable to expect that this base of the exponent is an artifact of the technique. Is it possible to design an algorithm with running time simply $\Ohstar(2^k)$?
\item In their work, Cygan et al.~\cite{above-LP} presented an algorithm for {\sc{Node Multiway Cut}} with running time $\Ohstar(2^k)$, based on the idea of LP-guided branching. Is it possible to obtain an algorithm with running time $\Ohstar(c^k)$ for some $c<2$? As shown by Cao et al.~\cite{CaoCF14}, this is indeed the case for the edge deletion variant.
\end{itemize}\end{itemize}





\bibliographystyle{abbrv}
\bibliography{edge-bip}

\end{document}